\documentclass{snkart}

\usepackage{microtype}
\usepackage{stmaryrd}
\usepackage{mathrsfs}
\usepackage{tikz}
\usepackage{thm-restate}

\date{29 September 2022}

\newtheorem{observation}[theorem]{Observation}

\numberwithin{equation}{section}

\makeatletter

\let\rel\mathbf         \let\gr\mathbb          \let\clo\mathscr               \let\tup\mathbf         \let\cox\mathsf         \let\top\mathcal              

\let\@psilon\epsilon
\let\epsilon\varepsilon
\let\varepsilon\@psilon
\let\eps\epsilon

\let\Union\bigcup
\newcommand{\defeq}{:=}

\newcommand{\ZZ}{\mathbb{Z}}
\newcommand{\NN}{\mathbb{N}}

\newcommand{\RR}{\mathbb{R}}

\DeclareMathOperator{\CSP}{CSP}
\DeclareMathOperator{\PCSP}{PCSP}
\DeclareMathOperator{\Pol}{Pol}
\DeclareMathOperator{\Inv}{Inv}
\DeclareMathOperator{\Aut}{Aut}
\DeclareMathOperator{\Hom}{\text{\normalfont \sffamily Hom}}

\newcommand{\pHom}{\Hom^{\text{prod}}}

\newcommand\Sphere{\mathcal S}
\newcommand\Torus{\mathcal T}

\newcommand\Bip[1]{\Hom(\rel K_2,\mathord{#1})}
\newcommand\gBip[1]{\geom{\Hom(\rel K_2,\mathord{#1})}}
\newcommand\gHom[1]{\geom{\Hom(#1)}}

\newcommand{\equivariant}{$\mathbb Z_2$}
\newcommand{\fg}{\pi_1}
\newcommand{\Fg}[1]{\fg(\mathord{#1})}

\newcommand{\geom}[1]{\lvert \mathord{#1}\rvert}
\newcommand{\size}[1]{\lvert \mathord{#1}\rvert}
\newcommand{\floor}[1]{\lfloor \mathord{#1}\rfloor}

\newcommand{\NP}{\textsf{\upshape NP}}
\newcommand{\Ptime}{\textsf{\upshape P}}

\newcommand{\yes}{{\scshape yes}}
\newcommand{\no}{{\scshape no}}

\DeclareMathOperator{\sub}{sub}
\DeclareMathOperator{\sym}{sym}

\newcommand\B[1]{{\binom{#1}{\lfloor #1/2 \rfloor}}}

\DeclareMathOperator{\id}{id}

\newcommand{\gr@}{\text{\normalfont\sffamily Gr}}
\newcommand{\dgr@}{\text{\normalfont\sffamily Dgr}}
\newcommand{\graf}{\gr@}
\newcommand{\dgraf}{\dgr@}
\newcommand{\dgra}{\dgr@_{\infty}}

\makeatother

\hypersetup{pdftitle  = {Topology and adjunction in promise constraint satisfaction},
  pdfauthor = {A. Krokhin, J. Opršal, M. Wrochna, and S. Živný}}

\begin{document}

\title{Topology and adjunction in promise constraint satisfaction}

\author{Andrei Krokhin}
\address{Department of Computer Science, Durham University, Durham, UK}
\email{andrei.krokhin@durham.ac.uk}

\author{Jakub Opršal}
\address{Institute of Science and Technology Austria, Klosterneuburg, Austria}
\email{jakub.oprsal@ist.ac.at}

\author{Marcin Wrochna}
\address{Faculty of Mathematics, Informatics, and Mechanics, University of Warsaw, Warsaw, Poland}
\email{m.wrochna@mimuw.edu.pl}

\author{Stanislav Živný}
\address{Department of Computer Science, University of Oxford, Oxford, UK}
\email{standa.zivny@cs.ox.ac.uk}

\thanks{\copyright\ licensed under Creative Commons License CC-BY 4.0.\\ \indent Preliminary versions of parts of this paper appeared in the proceedings of \emph{60th Annual IEEE Symposium on Foundations of Computer Science} \cite{KO19} and \emph{Fourteenth Annual ACM-SIAM Symposium on Discrete Algorithms (SODA20)} \cite{WZ20}.\\
\indent Andrei Krokhin and Jakub Opr\v{s}al were supported by the UK EPSRC grant EP/R034516/1.
Jakub Opršal has received funding from the European Union's Horizon 2020 research and innovation programme under the Marie Skłodowska-Curie Grant Agreement No 101034413.
Stanislav \v{Z}ivn\'y was supported by a Royal Society University Research Fellowship.
This project has received funding from the European Research Council (ERC) under the European Union's Horizon 2020 research and innovation programme (grant agreement No 714532). The paper reflects only the authors' views and not the views of the ERC or the European Commission. The European Union is not liable for any use that may be made of the information contained therein.}

\begin{abstract}
  The approximate graph colouring problem, whose complexity is unresolved in most cases, concerns finding a $c$-colouring of a graph that is promised to be $k$-colourable, where $c\geq k$. This problem naturally generalises to promise graph homomorphism problems and further to promise constraint satisfaction problems. The complexity of these problems has recently been studied through an algebraic approach. In this paper, we introduce two new techniques to analyse the complexity of promise CSPs: one is based on topology and the other on adjunction. We apply these techniques, together with the previously introduced algebraic approach, to obtain new unconditional NP-hardness results for a significant class of approximate graph colouring and promise graph homomorphism problems.
\end{abstract}


\maketitle

\section{Introduction}

  In this paper we investigate the complexity of finding an approximate solution to fully satisfiable instances of constraint satisfaction problems.
  For example, for the classical problem of $k$-colouring a~graph, one~natural approximation version is the \emph{approximate graph colouring} problem: The goal is to find a~$c$-colouring of a~given $k$-colourable graph, where $c\ge k\ge 3$.
  There is a~huge gap in our understanding of the complexity of this problem. For $k=3$, the best known efficient algorithm uses roughly $c=O(n^{0.199})$ colours where $n$ is the number of vertices of the graph \cite{KT17}. It has been long conjectured the problem is \NP-hard for any fixed constants $c\ge k\ge 3$, but, say for $k=3$, the state-of-the-art has only recently been improved from $c=4$ \cite{KLS00,GK04} to $c=5$ \cite{BKO19,BBKO19}.

  Graph colouring problems naturally generalise to graph homomorphism problems
  and further to constraint satisfaction problems (CSPs). In a~graph
  homomorphism problem, one is given two graphs and needs to decide whether
  there is a~homomorphism (an~edge-preserving map) from the first graph to the
  second \cite{HN04}. The CSP is a generalisation that uses arbitrary relational
  structures in place of graphs. One particularly important case that attracted
  much attention is when the second graph/structure is fixed; this is the
  so-called non-uniform CSP~\cite{BKW17,FV98}. This is also the only case we consider in this article. For graph homomorphisms, this is known as the $\rel H$-\emph{colouring} problem: decide whether a~given graph has a~homomorphism to a~fixed graph $\rel H$ \cite{HN04}. The \Ptime\ vs.\ \NP-complete dichotomy of $\rel H$-colouring given in \cite{HN90} was one of the base cases that supported the Feder-Vardi dichotomy conjecture for CSPs \cite{FV98}.
  The study of the complexity of the CSP and the complete resolution
  of the CSP dichotomy conjecture \cite{Bul17,Zhu17,Zhu20} was greatly influenced by
  the algebraic approach~\cite{Bulatov05:classifying} (see survey \cite{BKW17}). This approach has also made important contributions to the study of approximability of CSPs (e.g.\ \cite{BK16}).

  Brakensiek and Guruswami \cite{BG16-graph,BG18-structure} suggested that perhaps  progress on approximate graph colouring and similar open problems can be made by looking at a~broader picture, by extending it to \emph{promise graph homomorphism} and further to the \emph{promise constraint satisfaction problem} (PCSP).
  Promise graph homomorphism is an approximation version of the graph homomorphism problem in the following sense: in $\PCSP(\rel H,\rel G)$, we fix (not one but) two graphs $\rel H$ and $\rel G$ such that there is a~homomorphism from $\rel H$ to $\rel G$ (we write $\rel H\rightarrow \rel G$ to denote this). The goal is then to find a~$\rel G$-colouring for a~given graph when an $\rel H$-colouring is guaranteed to exist (but not given as part of input). The \emph{promise} is that the input graph is always $\rel H$-colourable and hence $\rel G$-colourable as well. The PCSP is a~natural generalisation of this to arbitrary relational structures, or in other words, a~generalisation of the decision CSP to the promise setting. Brakensiek and Guruswami proposed a conjecture that $\PCSP(\rel H,\rel G)$ is NP-hard for all non-bipartite loopless graphs $\rel H$ and $\rel G$ such that $\rel H\rightarrow\rel G$. This would generalise the approximate graph colouring conjecture and greatly extend the Hell-Ne\v{s}et\v{r}il dichotomy for $\rel H$-colouring~\cite{HN90}.

  Given the huge success of the algebraic approach to the CSP, it is natural to investigate what it can do for PCSPs. This investigation was started by Austrin, H{\aa}stad, and Guruswami \cite{AGH17}, with an application to a~promise version of SAT. It was further developed by Brakensiek, Guruswami and others \cite{BG16-graph,BG18-structure,BG19,BG20,BGWZ20,BWZ19} and applied to a~range of problems, including versions of approximate graph and hypergraph colouring. A~recent paper \cite{BKO19,BBKO19} describes a~general abstract algebraic theory for PCSPs, which shows, in particular, how algebraic properties precisely capture the power of gadget reductions in PCSPs. However, the algebraic theory of PCSPs is still very young and much remains to be done both in further developing it and in applying it to specific problems. We note that the aforementioned \NP-hardness of 5-colouring a~given 3-colourable graph was proved in \cite{BKO19,BBKO19} by applying this abstract theory.

  The gist of the algebraic theory is that the complexity of $\PCSP(\rel H,\rel G)$ depends only on (certain properties of) \emph{polymorphisms}, which are multi-variable functions that can be defined as homomorphisms from direct powers $\rel H^n$ into $\rel G$.
  However, the analysis of polymorphisms is in general a highly non-trivial task, and powerful tools are needed to conduct it. For resolving the CSP dichotomy conjecture, the structural theory of finite universal algebras provided such a tool. However, it is not clear how much this theory can be applied to the promise setting. In this paper, we show that algebraic topology gives a very useful tool to analyse polymorphisms and pinpoint the complexity of PCSPs. We do this by explaining how general PCSPs are naturally equipped with a topological structure, called homomorphism complexes, and how polymorphisms of a given PCSP can be understood through the continuous maps they induce.
  Homomorphism complexes (as well as several related constructions) have been actively studied in topological combinatorics~\cite{Koz08-book,Mat03}, though mainly to give obstructions to the existence of homomorphisms and mostly for the case of graphs.
  However, methods of algebraic topology can also be used to obtain important information about polymorphisms: for example, to identify ``influential'' variables.
  We demonstrate how this new methodology can be applied to resolve a significant part of the Brakensiek-Guruswami conjecture.

  We also show that the simple notion of adjunction, which is a certain form of homomorphism duality, provides a powerful tool to reason about reductions between PCSPs. We observe that adjunctions always give rise to reductions between PCSPs. Moreover, we prove that many reductions between PCSPs work because of the presence of adjunction. This includes, in particular, all gadget reductions (that are captured by the algebraic approach) and all reductions satisfying very mild technical conditions. We demonstrate how adjunction can be applied by significantly improving the state-of-the-art in approximate graph colouring --- via reductions that provably cannot be explained via the algebraic approach from~\cite{BBKO19}.

\subsection*{Related work}

The notion of PCSP was coined in \cite{AGH17}, though one of the main examples of problems of this form, approximate graph colouring, has been around for a~long time~\cite{GJ76}. The complexity landscape of PCSPs (beyond CSPs) is largely unknown, even for the Boolean case (see~\cite{BG18-structure,Ficak19:icalp}) or for graph colouring and homomorphisms.

Most notable examples of PCSPs studied before are related to graph and hypergraph colouring. We already mentioned some results concerning colouring 3-colourable graphs with a~constant number of colours.  Without additional complexity-theoretic assumptions, the strongest known \NP-hardness results for  colouring $k$-colourable graphs are as follows.  For any $k\ge 3$, it is \NP-hard to colour a~given $k$-colourable graph with $2k-1$ colours \cite{BKO19,BBKO19}. For large enough $k$,  it is \NP-hard to colour a~given $k$-colourable graph with $2^{\Omega(k^{1/3})}$ colours \cite{Hua13}.  The only earlier result about promise graph homomorphisms (with $\rel H\ne \rel G$) that involves more than approximate graph colouring is the \NP-hardness of 3-colouring for graphs that admit a~homomorphism to $\rel C_5$, the five-element cycle \cite{BBKO19}.

Under stronger assumptions (Khot's 2-to-1 Conjecture~\cite{Khot02stoc} for
$k\geq 4$ and its non-standard variant for $k=3$), Dinur, Mossel, and Regev
showed that finding a $c$-colouring of a $k$-colourable graph is NP-hard for all
constants $c\geq k\geq 3$~\cite{DMR09}.  It was shown in~\cite{GS19} that the above result for $k = 2d$ still holds if one assumes the $d$-to-1 Conjecture of Khot~\cite{Khot02stoc} for any fixed $d\ge 2$ instead of the 2-to-1 Conjecture (which is the strongest in the family of $d$-to-1 conjectures).  A variant of Khot's 2-to-1 Conjecture with imperfect completeness has recently been proved~\cite{DinurKKMS18,KhotMS18}, which implies hardness for approximate colouring variants for the weaker promise that most but not all of the graph is guaranteed to be $k$-colourable.

A~colouring of a~hypergraph is an assignment of colours to its vertices that
leaves no edge monochromatic. It is known that, for any constants $c\geq k\geq 2$, it is \NP-hard to find a $c$-colouring of a~given 3-uniform $k$-colourable hypergraph \cite{DRS05}.
  Further variants of approximate hypergraph colouring, e.g.\ relating to strong or rainbow colourings, were studied in \cite{ABP20,BG16-graph,BG17,GL18,GS20:rainbow}, but most complexity classifications related to them are still open in full generality.

Some results are also known for colourings with a super-constant number of colours.  For graphs, conditional hardness can be found in~\cite{Dinur10:approx}, and for hypergraphs, \NP-hardness results were obtained in~\cite{ABP19,Bha18}.

An accessible exposition of the algebraic approach to the CSP can be found in \cite{BKW17}, where many ideas and results leading to (but not including) the resolution \cite{Bul17,Zhu17,Zhu20} of the Feder-Vardi conjecture are presented. The volume \cite{KZ17} contains surveys concerning many aspects of the complexity and approximability of CSPs.

The first link between the algebraic approach and PCSPs was found by Austrin, H\aa{}stad, and Guruswami \cite{AGH17},  where they studied a promise version of $(2k+1)$-SAT called $(2+\epsilon)$-SAT. They use a~notion of \emph{polymorphism} (which is the central concept in the algebraic theory of CSP) suitable for PCSPs.
In \cite{BG18-structure}, it was shown that the complexity of a~PCSP is fully determined by its polymorphisms --- in the sense that two PCSPs with the same set of polymorphisms have the same complexity. They also use polymorphisms to prove several hardness and tractability results.
The algebraic theory of PCSP was lifted to an abstract level in \cite{BKO19,BBKO19}, where it was shown that abstract properties of polymorphisms determine the complexity of PCSP.

The topological methods that we develop in this paper originate in topological combinatorics, specifically in Lovász's celebrated proof~\cite{Lov78} that gives a tight lower bound on the chromatic number of Kneser graphs.
We refer to \cite{Mat03} for an approachable introduction, and to \cite{Koz08-book} for an in-depth technical reference.
The modern view of this method is to assign a topological space to a graph in such a way that combinatorial properties of the graph (e.g.\ the chromatic number) are influenced by topological properties of the resulting space (e.g.\ topological connectivity).
An intermediate step in the construction of the topological space is to assign a certain abstract simplicial complex to a~graph (we introduce these below). In our proof, we use so-called \emph{homomorphism complexes} that give a simplicial structure on the set of homomorphisms between two graphs (or other structures). We remark that restricting those complexes to vertices and edges (so called 1-skeletons) gives graphs of homomorphisms which have been used in CSP research before (see, e.g., \cite{BBDL19,LLT07}).

We remark that three earlier results on the complexity of approximate hypergraph colouring \cite{ABP20,Bha18,DRS05} were based on results from topological combinatorics using the Borsuk-Ulam theorem or similar~\cite{Lov78,Mat03}. Their use of topology seems different from ours, and it remains to be seen whether they are all occurrences of a common pattern.

Topological methods and adjunction (including some specific cases that we use in this paper) have also been actively used in research around Hedetniemi's conjecture about the chromatic number of graph products~\cite{FoniokT17,Matsushita17,Tardif08:survey,TardifW18,Wrochna17,Wrochna17b} (recently disproved by Shitov~\cite{Shitov19}).  A few ideas in this paper are inspired by this line of research.  A survey on adjunction and graph homomorphisms can be found in~\cite{FoniokT13} (see also~\cite{FoniokT15}), which also discusses several specific constructions that we use in this paper.

\subsection*{Our contributions}

We first describe our methodological contributions related to topology and adjunction and then specific applications to promise graph homomorphism and approximate graph colouring.
For simplicity, we will present the general theory for the case of graphs, which is what our applications are about;
nevertheless, it generalises immediately to arbitrary relational structures.
We do not assume that the reader is familiar with topological combinatorics or algebraic topology and provide the necessary definitions and explanations here and in later sections.

It will be clear to the reader familiar with category theory that much of what we do in this paper can be naturally expressed in category-theoretic language. However, we avoid using this language, for the benefit of the readers less familiar with category theory.

\subsubsection*{Topological analysis of polymorphisms}

As we mentioned before, the complexity of a problem $\PCSP(\rel H,\rel G)$ is completely determined by certain abstract properties of polymorphisms from $\rel H$ to $\rel G$. Our first contribution is the introduction of topology as a tool to analyse polymorphisms.
The basis for such analysis is the fact that the set of all homomorphisms from a graph $\rel H$ to another graph $\rel G$ can be made into an abstract simplicial complex denoted by $\Hom(\rel H,\rel G)$.

An \emph{abstract simplicial complex} $\cox K$ is a downwards closed family of non-empty subsets of a vertex set $V(\cox K)$ --- subsets in the family are called \emph{faces} (or \emph{simplices}), their elements are \emph{vertices}.
A~simplical complex describes a topological space: the \emph{geometric realisation} of $\cox K$, denoted $\geom{\cox K}$, is the subspace of $\RR^{V(\cox K)}$ obtained by identifying vertices with affinely independent points and, for each face, adding to the space the convex hull of the vertices in the face.
Thus a pair $\{u,v\} \in \cox K$ becomes an edge, a triple (i.e., 3-element face) becomes a filled triangle, a quadruple becomes a filled tetrahedron, and so on.
The resulting space can be analysed by using algebraic topology and the algebraic structures (groups, rings) that it associates with a topological space.

The vertex set of the complex $\Hom(\rel H,\rel G)$ is the set of all homomorphisms from $\rel H$ to $\rel G$. A~non-empty set $\{h_1,\ldots,h_\ell\}$ of such homomorphisms is a face if every function $h \colon V(H) \to V(G)$ satisfying $h(v) \in \{h_1(v),\dots,h_\ell(v)\}$ for all $v$ is a homomorphism $\rel H \to \rel G$.
For example, if two homomorphisms $h_1, h_2$ differ at only one vertex $v \in V(H)$, then they are connected by a line in $\geom{\Hom(\rel H,\rel G)}$.
Note the definition generalises in a straightforward way from graphs to arbitrary relational structures.

There are several ways to use this notion for analysis of polymorphisms. One is to directly use the topological structure of $\geom{\Hom(\rel H^n,\rel G)}$ --- for example, by looking at various connectivity properties in this space and asking when polymorphisms (as points in this space) belong to the same component.
Another one, and this is what we use in the paper, goes as follows. Any (say, $n$-ary) polymorphism $f$ from $\rel H$ to $\rel G$, i.e., a homomorphism from $\rel H^n$ to $\rel G$,  induces in a natural way a continuous map $\tilde{f}$ from the space $\geom{\Hom(\rel K_2,\rel H^n)}$ to $\geom{\Hom(\rel K_2,\rel G)}$, where $\rel K_2$ is the two-element clique. One can then obtain information about $f$ from algebraic invariants of~$\tilde{f}$.

As an important example, suppose that $\rel H, \rel G$ are (possibly different) odd cycles.
It is well known and not hard to check that $\geom{\Hom(\rel K_2,\rel H)}$ is topologically equivalent to the circle $\Sphere^1$ (we do this later in Example~\ref{ex:cycle}) and $\geom{\Hom(\rel K_2,\rel H^n)}$ to the $n$-torus $\Torus^n = \Sphere^1 \times \dots \times \Sphere^1$. A~homomorphism $f$ from $\rel H$ to $\rel G$ induces a continuous map $\tilde{f}$ from $\Sphere^1$ to $\Sphere^1$, and the main algebraic invariant of such a map is its \emph{degree}, or winding number, which is an integer that intuitively measures how many times $\tilde{f}$ wraps the domain circle around the range circle (and in which direction).
The degree of $\tilde{f}$ will be bounded because it arises from a discrete map $f$.
Similarly, when analysing a homomorphism $f$ from $\rel H^n$ to $\rel G$, we study $\tilde{f}$, which is now a continuous map from $\Torus^n$  to $\Sphere^1$.
Each of the $n$ variables of $\tilde{f}$ corresponds to a circle in $\Torus^n$ and thus to a degree of $\tilde{f}$ restricted to that circle.
We show that the number of variables whose degrees are non-zero is bounded, again because $\tilde{f}$ arises from a discrete function $f$. In this way, we obtain that each polymorphism $f$
has a bounded number of coordinates (independent of $n$) that are ``important'', or
``influential'', and we can then use this information, together with the previously developed algebraic theory~\cite{BBKO19}, to show that $\PCSP(\rel H,\rel G)$ is \NP-hard.

\subsubsection*{Adjunction}

We use symbols $\Lambda,\Gamma$ for functions from the class of all (finite) graphs to itself.  It will be convenient to write $\Lambda \rel H$ instead of $\Lambda(\rel H)$ for the image of $\rel H$ under $\Lambda$. The definitions and general properties again extend to all relational structures. Adjunction is a certain duality property between functions, best introduced with a concrete example.
\begin{example}
\label{ex:walk-power}
For a graph $\rel G$ and an odd integer $k$ one can consider the following functions: $\Lambda_k \rel G$ is defined to be the graph obtained by subdividing each edge of $\rel G$ into a path of $k$ edges, and $\Gamma_k \rel G$ is the graph obtained by taking the $k$-th power of the adjacency matrix (with zeroes on the diagonal; equivalently, the vertex set remains unchanged and two vertices are adjacent if and only if there is a walk of length exactly $k$ in $\rel G$). For example, $\Gamma_3 \rel G$ has loops if $\rel G$ has triangles.
\end{example}

Two functions $\Lambda, \Gamma$ are called \emph{adjoint}~if
\[
  \Lambda \rel H \to \rel G \text{ if and only if } \rel H \to \Gamma \rel G
\]
for all graphs $\rel G,\rel H$.  In this case $\Lambda$ is also called the \emph{left adjoint to $\Gamma$}, and $\Gamma$ is \emph{the right adjoint to~$\Lambda$}.  For example, it is well known and easy to check that $\Lambda_k,\Gamma_k$ are adjoint, for any fixed odd~$k$~\cite{FoniokT13}. Adjoint functions are always \emph{monotone} with respect to the homomorphism preorder, i.e., $\rel H \to \rel G$ implies both $\Lambda\rel H \to \Lambda \rel G$ and $\Gamma\rel H \to \Gamma \rel G$ (see Lemma~\ref{lem:adj-technical}).

Adjoint functions give us a way to reduce one PCSP to another.  Indeed, consider any function~$\Lambda$.  We can always attempt to use it as a reduction between \emph{some} PCSPs: if an instance graph $\rel I$ is guaranteed to be $\rel H$-colourable, then $\Lambda \rel I$ is guaranteed to be $\Lambda \rel H$-colourable if $\Lambda$ is monotone.  On the other hand if we find $\Lambda \rel I$ to be $\rel G$-colourable, this may imply that $\rel I$ is $\rel X$-colourable for some graph $\rel X$.  In such a case $\Lambda$ would be a reduction from $\PCSP(\rel H,\rel X)$ to $\PCSP(\Lambda \rel H, \rel G)$.  What is the best possible $\rel X$?  It is a graph $\rel X$ such that for any instance $\rel I$, $\Lambda \rel I \to \rel G$ holds if and only if $\rel I \to \rel X$.  If such an $\rel X$ exists, it is essentially unique (since we just defined what homomorphisms $\rel X$ admits).  The function that assigns to a graph $\rel G$ this best possible $\rel X$ is exactly the right adjoint to $\Lambda$.  In this way, adjoints help us identify the best possible reduction a function gives, even though the proof that the reduction works might not need to mention the right adjoint.

\subsubsection*{Applications}
  Our applications of the above methodologies aim towards resolving the Bra\-ken\-siek-Guruswami conjecture mentioned earlier:

  \begin{conjecture}[Brakensiek and Guruswami~\cite{BG18-structure}]
    \label{conj:main} \label{conj:bg}
    Let $\rel H$ and $\rel G$ be any non-bipartite loopless graphs with $\rel H\to \rel G$. Then $\PCSP(\rel H,\rel G)$ is \NP-hard.
  \end{conjecture}

We remark that the Hell-Ne\v{s}et\v{r}il theorem~\cite{HN90} confirms Conjecture~\ref{conj:main} for the case $\rel H = \rel G$.
We also remark that Conjecture~\ref{conj:bg} covers all graphs: As discussed in
Section~\ref{sec:preliminaries}, if either $\rel H$ or $\rel G$ is bipartite or
contains a loop then $\PCSP(\rel H, \rel G)$ can be easily solved in polynomial time.

  It is not hard to see that the conjecture is equivalent to the statement that
  $\PCSP(\rel C_k,\rel K_c)$ is \NP-hard for all $k\geq 3$ odd and $c \geq 3$,
  where $\rel C_k$ is a cycle on $k$ vertices and $K_c$ is a clique on $c$
  vertices. This is because we have a chain of homomorphisms
  \begin{equation}\label{eq:chain}
    \dots \to \rel C_k \to \dots \to \rel C_5 \to \rel C_3 =
    \rel K_3 \to \rel K_4 \to \dots \to \rel K_c \to \dots
  \end{equation}
  and, for each $(\rel H, \rel G)$ with a homomorphism $\rel H \to \rel G$, the problem $\PCSP(\rel H, \rel G)$ admits a trivial reduction from $\PCSP(\rel C_k, \rel K_c)$, where the promise is strengthened by requiring the input graph to be $\rel C_k$-colourable, for an odd cycle $\rel C_k$ in $\rel H$, and the goal is weakened to that of finding a $\rel K_c$-colouring, where $c$ is the chromatic number of $\rel G$ (so we have $\rel C_k \to \rel H$ and $\rel G \to \rel K_c$).

  The chain~\eqref{eq:chain} has a natural middle point $\rel K_3$.  The right half corresponds to the classical approximate graph colouring: find a $c$-colouring of a 3-colourable graph. Our applications make progress on the right half and show hardness for all of the left half.

  For the left half, we use the topological analysis of polymorphisms, as described above, to confirm Conjecture~\ref{conj:bg} for $\rel G = \rel K_3$:

  \begin{theorem}\label{thm:K3}
    $\PCSP(\rel H,\rel K_3)$ is \NP-hard for every non-bipartite 3-colourable $\rel H$.
  \end{theorem}

  Equivalently, $\PCSP(\rel C_k,\rel K_3)$ is \NP-hard for all odd $k \geq 3$.  We in fact prove a more general result which covers other graphs $\rel G$ with similar topological properties to $\rel K_3$ --- namely that $\geom{\Hom(\rel K_2,\rel G)}$ maps to the circle $\Sphere^1$ via a~\emph{\equivariant-map} (see Definition \ref{def:z2-map}).

  \begin{restatable}{theorem}{maintheoremtopology}
    \label{thm:main-s1}
    Let $\rel H,\rel G$ be non-bipartite loopless graphs such that $\rel H \to \rel G$, and there is a \equivariant-map from $\gBip{\rel G}$ to $\Sphere^1$.
    Then $\PCSP(\rel H,\rel G)$ is \NP-hard.
  \end{restatable}

  We give two specific classes of graphs $\rel G$ satisfying the assumptions of Theorem \ref{thm:main-s1}: certain circular cliques and all square-free graphs.

  For positive integers $p,q$ such that $p\ge 2q$, the \emph{circular clique} $\rel K_{p/q}$ is the graph that has the same vertex set as the cycle $\rel C_p$ and two vertices in it are connected by an edge if and only if they are at distance at least $q$ in $\rel C_p$ (see Fig.~\ref{fig:circular-cliques}). It well known that $\rel K_{n/1}$ is isomorphic to $\rel K_n$, $\rel K_{(2n+1)/n}$ is isomorphic to $\rel C_{2n+1}$,
  and that $\rel K_{p/q}\to \rel K_{p'/q'}$ if and only if $p/q\le p'/q'$ (see, e.g., Theorem~6.3 in~\cite{HN04}), thus circular cliques refine the homomorphism order~\eqref{eq:chain} on odd cycles and cliques described above.
  The \emph{circular chromatic number} of $\rel G$, $\chi_c(\rel G)$, is defined as $\inf\{p/q \mid \rel G\to \rel K_{p/q}\}$. Note that we always have $\chi(\rel G) = \lceil \chi_c(\rel G)\rceil$ and also $\chi_c(\rel G)\le 2+\frac{1}{n}$ if and only if $\rel G\to \rel C_{2n+1}$.

\begin{figure}
  \[\begin{array}{c@{\qquad}c@{\qquad}c}
    \begin{tikzpicture}[scale = 1.3, baseline={([yshift=-.5ex]current bounding box.center)}]
      \draw [gray] (0,0) circle (1cm);
      \foreach \i/\c in {0/0,72/1,144/2,-144/3,-72/4}
        \node (\c) [circle,fill,inner sep=1.5,label={\i:$\c$}] at (\i:1) {};
      \foreach \i/\c in {0/0,72/1,144/2,-144/3,-72/4} {
        \draw [thick] (\i:1) -- (72+\i:1);
        \draw [thick] (\i:1) -- (144+\i:1);
      }
    \end{tikzpicture}
    &
    \begin{tikzpicture}[scale = 1.3, baseline={([yshift=-.5ex]current bounding box.center)}]
      \draw [gray] (0,0) circle (1cm);
      \foreach \i/\c in {0/0,72/1,144/2,-144/3,-72/4}
        \node (\c) [circle,fill,inner sep=1.5,label={\i:$\c$}] at (\i:1) {};
      \foreach \i/\c in {0/0,72/1,144/2,-144/3,-72/4} {
        \draw [thick] (\i:1) -- (144+\i:1);
      }
    \end{tikzpicture}
    &
    \begin{tikzpicture}[scale = 1.3, baseline={([yshift=-.5ex]current bounding box.center)}]
      \draw [gray] (0,0) circle (1cm);
      \foreach \i/\c in {0/0,51.42/1,102.85/2,154.28/3,205.71/4,257.14/5,308.57/6}
        \node (\c) [circle,fill,inner sep=1.5,label={\i:$\c$}] at (\i:1) {};
      \foreach \i/\j in {0/2,1/3,2/4,3/5,4/6,5/0,6/1,0/3,1/4,2/5,3/6,4/0,5/1,6/2}
        \draw [thick] (\i) -- (\j);
    \end{tikzpicture}
    \\
    \rel K_{5/1} \simeq \rel K_5 &
    \rel K_{5/2} \simeq \rel C_5 &
    \rel K_{7/2}
  \end{array}\]
  \caption{Examples of circular cliques.}
  \label{fig:circular-cliques}
\end{figure}

  The fact that circular cliques $\rel K_{p/q}$ with $2<p/q<4$ satisfy the topological condition of Theorem \ref{thm:main-s1} is folklore, though we prove it later for completeness.
  The theorem in this case can be viewed as \NP-hardness of colouring $(2+\eps)$-colourable graphs with $4-\eps$ colours:

  \begin{corollary} \label{thm:circular-cliques}
    $\PCSP(\rel K_{p/q},\rel K_{p'/q'})$ is \NP-hard for all $2<p/q\le p'/q'<4$.
  \end{corollary}

  A graph is said to be \emph{square-free} if it does not contain the 4-cycle $\rel C_4$ as a subgraph. This includes all graphs of girth at least $5$ and thus graphs of arbitrarily high chromatic number.
  Again, it will be a simple observation that square-free graphs satisfy the condition of Theorem~\ref{thm:main-s1}.
  Therefore, we confirm Conjecture \ref{conj:bg} for square-free graphs $\rel G$.

  \begin{corollary}\label{thm:square-free}
    $\PCSP(\rel H,\rel G)$ is \NP-hard for all non-bipartite loopless graphs $\rel H, \rel G$ such that $\rel H\rightarrow \rel G$ and $\rel G$ is square-free.
  \end{corollary}

  Since the key assumption of Theorem~\ref{thm:main-s1} is topological, this raises a question whether topology is in some sense necessary to settle Conjecture~\ref{conj:bg}. Using adjointness, we argue that it is indeed the case, proving the following (see Theorem~\ref{thm:topoOnly} for a~formal statement).

  \begin{theorem}[\normalfont{informal}]\label{thm:topoOnlyInformal}
    For any graph $\rel G$, the property that $\PCSP(\rel H, \rel G)$ is
    \NP-hard for all non-bipartite $\rel G$-colourable graphs $\rel H$ depends
    only on the topology (and \equivariant-action) of $\geom{\Hom(\rel K_2,\rel G)}$.
  \end{theorem}

  \bigskip

  Returning to the right half of the chain~\eqref{eq:chain} (the classical colouring problem), we first show that,
  to prove \NP-hardness of $c$-colouring $k$-colourable graphs for all constants $c\geq k\geq 3$,
  it is enough to prove it for {\em any fixed} $k$ (and all $c\ge k$).

  \begin{theorem} \label{thm:conditional}
    Suppose there is an integer $k$ such that $\PCSP(\rel K_k,\rel K_c)$ is \NP-hard for all $c \geq k$.
    Then $\PCSP(\rel K_3,\rel K_c)$ is \NP-hard for all $c \geq 3$.
  \end{theorem}

  Following the reasoning in~\cite{GS19}, the above theorem implies \NP-hardness of all problems $\PCSP(\rel K_k,\rel K_c)$ with $c \geq k\ge 3$ if the $d$-to-1 conjecture of Khot holds for any fixed $d\ge 2$.
  (The paper~\cite{GS19} used an earlier version of Theorem~\ref{thm:conditional} with 4 in place of 3).

  Furthermore, we strengthen the best known asymptotic hardness: Huang~\cite{Hua13} showed that $\PCSP(\rel K_k, \rel K_c)$ is NP-hard for all sufficiently large $k$ and $c=2^{\Omega(k^{1/3})}$.  We improve this in two ways, using Huang's result as a black-box.  First, we improve the asymptotics from sub-exponential $c=2^{\Omega(k^{1/3})}$ to single-exponential $c = \B{k} \in \Theta ( 2^k / \sqrt k )$.  Second, we show the claim holds for $k$ starting as low as $4$.

  \begin{restatable}{theorem}{maintheoremasymptotics}\label{thm:asymp}
    For all $k \geq 4$ and $c = \B{k} - 1$, $\PCSP(\rel K_k, \rel K_c)$ is \NP-hard.
  \end{restatable}

  In comparison, the previous best result relevant for all integers $k$ was obtained in~\cite{BBKO19} where \NP-hardness of $\PCSP(\rel K_k, \rel K_{2k-1})$ is proved for all $k\geq 3$.
  For $k = 3, 4$ we obtain no new results and for $k = 5$ the two bounds coincide: $\B{k}-1 = 9 = 2k-1$.
  However, already for $k = 6$ we improve the bound from $2k-1 = 11$ to $\B{k}-1 = 19$, and, for larger $k$, the improvement is even more dramatic.

\subsubsection*{The organisation of the paper}
Section~\ref{sec:preliminaries} briefly describes the algebraic framework of \cite{BBKO19}: minions (sets of polymorphisms of a PCSP), minion homomorphism (which provide log-space reductions between corresponding problems), and a condition on minions that guarantees \NP-hardness.
Section~\ref{sec:topology} details the topological method and its application: Theorem~\ref{thm:main-s1}.
The bulk of its content is devoted to expounding standard definitions with examples and then proving these definitions behave well when identifying variables of polymorphisms.
Section~\ref{sec:adjunction} introduces adjunction in a wider context, in particular relating it to gadget reductions and minion homomorphisms.
Adjoint functions that give reductions for approximate graph colouring are presented in Section~\ref{sec:righthard}.
Finally Section~\ref{subsec:secondMainProof} uses another adjoint function to prove
Theorem~\ref{thm:topoOnlyInformal}: that whether a graph $\rel G$ satisfies the Brakensiek-Guruswami conjecture for all $\rel H$ depends only on the topology of $\rel G$.

\section{Preliminaries}\label{sec:preliminaries}

\subsection{Promise graph homomorphism problems}

The approximate graph colouring problem and promise graph homomorphism problem are special cases of the PCSP, and we use the theory of PCSPs. However, we will not need the general definitions, so we define everything only for digraphs. For general definitions, see, e.g.\ \cite{BBKO19}.

A \emph{digraph} $\rel H$ is a pair $\rel H=(V(H),E(H))$, where $V(H)$ is a set of vertices and $E(H) \subseteq \{(u,v) \mid u,v\in V(H)\}$ is a set of
(directed) edges. Unless stated otherwise, our digraphs are finite and can have loops. We view undirected graphs as digraphs where each (non-loop) edge is presented in both directions.

\begin{definition}
A \emph{homomorphism} from a~digraph $\rel H=(V(H),E(H))$ to another digraph $\rel G=(V(G),E(G))$ is a~map $h\colon V(H)\to V(G)$ such that $(h(u),h(v))\in E(G)$ for every $(u,v)\in E(H)$.  In this case we write $h\colon \rel H\to \rel G$, and simply $\rel H\to \rel G$ to indicate that a~homomorphism exists.
\end{definition}

We now define formally the promise (di)graph homomorphism problem.

\begin{definition}
Fix two digraphs $\rel H$ and $\rel G$ such that $\rel H \rightarrow \rel G$.
\begin{itemize}
  \item The \emph{search} variant of $\PCSP(\rel H,\rel G)$ is, given an~input digraph $\rel I$ that maps homomorphically to $\rel H$, \emph{find} a~homomorphism $h\colon \rel I\to \rel G$.
  \item The \emph{decision} variant of $\PCSP(\rel H,\rel G)$ requires, given an input digraph $\rel I$ such that either $\rel I\to \rel H$ or $\rel I\not\to\rel G$, to output \yes{} in the former case, and \no{} in the latter case.
\end{itemize}
\end{definition}

We remark that the (decision) problem $\PCSP(\rel H,\rel H)$ is nothing else but the constraint satisfaction problem $\CSP(\rel H)$, also known as $\rel H$-colouring.

There is an obvious reduction from the decision variant of each PCSP to the search variant, but it is not known whether the two variants are equivalent for each PCSP. The hardness results in this paper hold for the decision (and hence also for the search) version of $\PCSP(\rel H,\rel G)$.

It is obvious that if at least one of $\rel H, \rel G$ is undirected and bipartite then the problem can be solved in polynomial time by using an algorithm for 2-colouring. If one of the graphs contains a loop, the problem is trivial.  Recall that Brakensiek and Guruswami conjectured (see Conjecture \ref{conj:main}) that, for undirected graphs, the problem is \NP-hard in all the other cases.

All applications in this paper concern undirected graphs, even though some proofs use digraphs.
We remark that, as shown in Theorem F.3 of the arXiv version of~\cite{BG18-structure}
(generalising the corresponding result for CSPs~\cite{FV98}), a complexity classification of all problems $\PCSP(\rel H,\rel G)$ for digraphs is equivalent to such a classification for all PCSPs (for arbitrary relational structures).

Two (di)graphs $\rel H_1$ and $\rel H_2$ are called {\em homomorphically equivalent} if both $\rel H_1\rightarrow \rel H_2$ and
$\rel H_2\rightarrow \rel H_1$.
The binary relation $\rel H_1\rightarrow \rel H_2$ defines a preorder on the class of all digraphs (or all graphs), called the {\em homomorphism preorder}. We will use this preorder in Section~\ref{sec:adjunction}.

We also define digraph powers, which are essential for the notion of polymorphisms.

\begin{definition}\label{def:nth-power}
The \emph{$n$-th direct (or tensor) power} of a~digraph $\rel H$ is the digraph $\rel H^n$ whose vertices are all $n$-tuples of vertices of $\rel H$ (i.e., $V(H^n) = V(H)^n$), and whose edges are defined as follows: we have an edge from $(u_1,\dots,u_n)$ to $(v_1,\dots,v_n)$ in $\rel H^n$ if and only if $(u_i,v_i)$ is an edge of $\rel H$ for all $i \in \{1,\dots,n\}$.
\end{definition}

\subsection{Polymorphisms}
We use the notions of polymorphisms \cite{AGH17,BG18-structure}, minions and minion homomorphisms \cite{BKO19,BBKO19}. We introduce these notions in the special case of digraphs below. General definitions and more insights can be found in \cite{BBKO19,BKW17}.

\begin{definition}
An $n$-ary \emph{polymorphism} from a~digraph $\rel H$ to a~digraph $\rel G$ is a~homomorphism from $\rel H^n$ to $\rel G$.
To spell this out, it is a~mapping $f\colon V(H)^n \to V(G)$ such that, for all tuples $(u_1,v_1)$, \dots, $(u_n,v_n)$ of edges of $\rel H$, we have
\[
  (f( u_1,\dots,u_n ), f( v_1,\dots,v_n )) \in E(G)
.\]
We denote the set of all polymorphisms from $\rel H$ to $\rel G$ by $\Pol(\rel H,\rel G)$.
\end{definition}

\begin{example} The $n$-ary polymorphisms from a~digraph $\rel H$ to the $k$-clique $\rel K_k$ are the $k$-colourings of $\rel H^n$.
\end{example}

The set of all polymorphisms between any two digraphs has a certain algebraic structure, which we now describe.
We denote by $[n]$ the set $\{1,2,\ldots,n\}$.

\begin{definition}\label{def:minor} An~$n$-ary function $f\colon A^n\to B$ is called a~\emph{minor} of an $m$-ary function $g\colon A^m \to B$ if there is a~map $\pi \colon [m] \to [n]$ such that
\[
  f(x_1,\dots,x_n) = g(x_{\pi(1)},\dots,x_{\pi(m)})
\]
for all $x_1,\dots,x_n \in A$. In this case, we write $f=g^{\pi}$.
\end{definition}

Alternatively, one can say that $f$ is a~minor of $g$ if it is obtained from $g$ by identifying variables, permuting variables, and introducing inessential variables.

\begin{definition}
  For sets $A,B$, let $\clo O(A,B) = \{f\colon A^n\rightarrow B\mid n\ge 1\}$. A~\emph{(function) minion} $\clo M$ on a pair of sets $(A,B)$ is a~non-empty subset of $\clo O(A,B)$ that is closed under taking minors. For fixed $n\ge 1$, let $\clo M^{(n)}$ denote the set of $n$-ary functions from $\clo M$.
\end{definition}

It is easy to see that $\Pol(\rel H,\rel G)$ is a minion whenever $\rel H\rightarrow\rel G$.

An important notion in our analysis of polymorphisms is that of an essential coordinate.

\begin{definition}
  \label{def:essential}
A coordinate $i$ of a function $f\colon A^n\rightarrow B$ is called \emph{essential} if $f$ depends on it, that is,
if there exist $a_1,\ldots,a_n$ and $b_i$ in $A$ such that
\[
  f(a_1,\ldots,a_{i-1},a_i,a_{i+1},\ldots, a_n)\ne f(a_1,\ldots,a_{i-1},b_i,a_{i+1},\ldots, a_n).
\]
A coordinate of $f$ that is not essential is called \emph{inessential}.
\end{definition}

\begin{definition}
A~minion $\clo M$ is said to have \emph{essential arity at most $k$}, if each function $f\in \clo M$ has at most $k$ essential variables. It is said to have \emph{bounded essential arity} if it
has essential arity at most $k$ for some $k$.
\end{definition}

\begin{example}
It is well known (see, e.g., \cite{GL74}), and not hard to check, that the minion $\Pol(\rel K_3,\rel K_3)$ has essential arity at most 1.
However for any odd $k>3$, the minion $\Pol(\rel C_k,\rel K_3)$ does not have bounded essential arity. Indeed, fix a~homomorphism $h\colon \rel C_k\to\rel K_3$ such that $h(0)=h(2)=0$ and $h(1)=1$ and define the following function from $\rel C_k^n$ to $\rel K_3$:
\[
  f(x_1,\ldots,x_n) =
    \begin{cases}
      2 & \text{if }x_1=\ldots=x_n=1,\\
      h(x_1) & \text{otherwise.}
    \end{cases}
\]
It is easy to check that $f\in \Pol(\rel C_k,\rel K_3)$.
By using Definition~\ref{def:essential} with $a_1=\ldots=a_n=1$ and $b_i=0$, one can verify that every coordinate $i$ of $f$ is essential.
\end{example}

\begin{definition}\label{def:minion-homomorphism}
  Let $\clo M$ and $\clo N$ be two minions (not necessarily on the same pairs of sets). A~mapping $\xi\colon \clo M \to \clo N$ is called a~\emph{minion homomorphism} if
  \begin{enumerate}
    \item it preserves arities, i.e., maps $n$-ary functions to $n$-ary functions for all $n$, and
\item it preserves taking minors, i.e., for each $\pi\colon \{1,\dots,m\} \to \{1,\dots,n\}$ and each $g\in \clo M^{(m)}$ we have $\xi(g)^{\pi}=\xi(g^{\pi})$, i.e.,
    \[
      \xi (g)(x_{\pi(1)},\dots,x_{\pi(m)}) = \xi (g(x_{\pi(1)},\dots,x_{\pi(m)}))
    .\]
  \end{enumerate}
\end{definition}

We refer to \cite[Example 2.22]{BBKO19} for examples of minion homomorphisms.

Our proof of Theorem~\ref{thm:main-s1} is based on the following result. It is a special case of a result in \cite{BBKO19} (that generalised \cite[Theorem~4.7]{AGH17}). We remark that the proof of this theorem is by a reduction from Gap Label Cover, which is a~common source of inapproximability results.

\begin{theorem}[{\cite[Proposition 5.15]{BBKO19}}] \label{thm:bounded-arity}
  Let $\rel H, \rel G$ be digraphs such that $\rel H \to \rel G$.
  Assume that there exists a minion homomorphism $\xi\colon\Pol(\rel H,\rel G) \rightarrow \clo M$ for some minion $\clo M$ on a pair of (possibly infinite) sets such that $\clo M$ has bounded essential arity and does not contain a~constant function (i.e., a~function without essential variables). Then $\PCSP(\rel H,\rel G)$ is \NP-hard.
\end{theorem}

To prove Theorem~\ref{thm:main-s1}, we will use Theorem~\ref{thm:bounded-arity} with the minion $\clo M=\clo Z_{\leq N}$, for some constant $N>0$.
The set $\clo Z_{\leq N}$ is defined to consist of all linear functions $f$ on $\mathbb Z$ of the form
\[
  f(x_1,\dots,x_n) = \sum_{i=1}^n c_i x_i
\]
for some $c_i\in \mathbb Z$ such that $\sum_{i=1}^n c_i$ is odd and $\sum_{i=1}^n \size{c_i} \leq N$.
It easy to see that $\clo Z_{\leq N}$ is indeed a minion and that all functions in it have between 1 and $N$ non-zero coefficients, meaning that it has bounded essential arity and contains no constant function.

\section{Topology}
  \label{sec:topology}

All graphs in this section are assumed to be undirected and loopless.

\subsection{Simplicial complexes}

An \emph{(abstract) simplicial complex} is a~family of non-empty sets $\cox K$ that is downwards closed, i.e., if $\sigma_1 \in \cox K$, $\sigma_2\neq \emptyset$ and $\sigma_2 \subseteq \sigma_1$, then $\sigma_2 \in \cox K$. Each $\sigma \in \cox K$ is called a \emph{face}. The elements in these sets are \emph{vertices} of $\cox K$. We denote the set of all vertices of $\cox K$ by $V(\cox K)$, i.e., $V(\cox K) \defeq \Union_{\sigma \in \cox K} \sigma$.  A~\emph{simplicial map} between complexes $\cox K$ and $\cox K'$ is a function $f \colon V(\cox K) \to V(\cox K')$ that preserves faces, i.e., if $\sigma \in \cox K$ then $f(\sigma) \defeq \{f(v) \mid v\in \sigma\} \in \cox K'$.
Two simplicial complexes $\cox K$ and $\cox K'$ are \emph{isomorphic} if there are simplicial maps $\alpha\colon \cox K \to \cox K'$ and $\beta\colon \cox K' \to \cox K$ such that both $\alpha\beta$ and $\beta\alpha$ are identity maps.

We will use the following notion of a~product of simplicial complexes (see also \cite[Section 2.2]{Matsushita17} and \cite[Definition 4.25]{Koz08-book}).

\begin{definition}
  Let $\cox K_1,\dots,\cox K_n$ be simplicial complexes. We define the product $\cox K_1\times \dots \times \cox K_n$ to be the simplicial complex with vertices
  \[
    V( \cox K_1 \times \dots \times \cox K_n ) = V(\cox K_1) \times \dots \times V(\cox K_n),
  \]
  so that $\sigma \subseteq V(\cox K_1\times \dots\times \cox K_n)$ is a~face if there are faces $\sigma_1 \in \cox K_1$, \dots, $\sigma_n \in \cox K_n$ such that $\sigma \subseteq \sigma_1\times \dots \times \sigma_n$.
\end{definition}

\subsubsection{From graphs to simplicial complexes.}

As mentioned before, there are several ways to assign a simplicial complex to a graph. For our use, the most convenient is the \emph{homomorphism complex}. Our definition of this complex is slightly different from that in~\cite{BK06,Koz08-book}, but the difference is superficial (as we explain in Appendix~\ref{app:hom-and-box}). The vertices of such a complex are homomorphisms, while faces are determined by multihomomorphisms defined below.

\begin{definition} \label{def:multimorphism}
  A~\emph{multihomomorphism} from $\rel K$ to $\rel G$ is a~mapping $f\colon V(K) \to 2^{V(G)}$ such that, for each $(u,v) \in E(K)$, we have $f(u)\times f(v) \subseteq E(G)$.
\end{definition}

\begin{definition} \label{def:homcomplex}
  Let $\rel K$ and $\rel G$ be two graphs. We define a~simplicial complex $\Hom(\rel K,\rel G)$ as follows. Its vertices are homomorphisms from $\rel K$ to $\rel G$, and $\sigma = \{ f_1,\dots,f_\ell \}$ is a~face if the mapping $u\mapsto \{ f_1(u), \dots, f_\ell(u) \}$ is a~multihomomorphism from $\rel K$ to $\rel G$.
\end{definition}

We work almost exclusively with complexes $\Bip{\rel G}$, where $\rel K_2$ is the two-element clique.
Such complexes (with our definition) appeared before, e.g.\ in~\cite{Mat17-box}, where they are called box complexes
(which is not the traditional use of this name) and in~\cite{MZ04}, where these complexes appear under the name $\mathsf{B_{edge}}(\rel G)$.
The complex $\Bip{\rel G}$ can be also described in the following way.
The vertices of $\Bip{\rel G}$ are all (oriented) edges of $\rel G$.
The faces are directed bipartite subgraphs that can be extended to a~complete directed bipartite subgraph of $\rel G$ (with all edges directed from one part to the other); more precisely, $\sigma$ is a~face if there are $U,V\subseteq V(G)$ such that $\sigma \subseteq U\times V\subseteq E(G)$.
The complexes $\Bip{\rel G}$ have an additional structure obtained from the automorphism of $\rel K_2$ that switches the two vertices.
The group $\mathbb Z_2$ then acts on the vertices of $\Bip{\rel G}$ by reversing the direction of edges, i.e., $-(a,b) = (b,a)$.

\begin{figure}
  \raisebox{-.5\height}{\begin{tikzpicture} [scale = 1.3, every node/.style = {font={\small}}]
    \draw (0:1) \foreach \i in {36,72,...,360} { -- (\i:1) };
    \foreach \i/\c/\d in {0/(0/4),36/(0/1),72/(2/1),108/(2/3),144/(4/3),180/(4/0),216/(1/0),252/(1/2),288/(3/2),324/(3/4)}
      \node [circle,fill,inner sep=1,label={\i:\c,\d}] at (\i:1) {};
  \end{tikzpicture}}
  \qquad
  \raisebox{-.5\height}{\begin{tikzpicture} [scale = 1]
    \draw (0:1) \foreach \i in {60,120,...,360} { -- (\i:1) };
    \node [circle,fill,inner sep=1,label={left:(0,5)}] at (0:1) {};
    \node [circle,fill,inner sep=1,label={60:(0,1)}] at (60:1) {};
    \node [circle,fill,inner sep=1,label={120:(2,1)}] at (120:1) {};
    \node [circle,fill,inner sep=1,label={180:(2,3)}] at (180:1) {};
    \node [circle,fill,inner sep=1,label={240:(4,3)}] at (240:1) {};
    \node [circle,fill,inner sep=1,label={300:(4,5)}] at (300:1) {};

    \begin{scope}[shift={(3,0)},xscale=-1, every node/.style = {font={\small}}]
      \draw (0:1) \foreach \i in {60,120,...,360} { -- (\i:1) };
      \node [circle,fill,inner sep=1,label={right:(5,0)}] at (0:1) {};
      \node [circle,fill,inner sep=1,label={120:(1,0)}] at (60:1) {};
      \node [circle,fill,inner sep=1,label={60:(1,2)}] at (120:1) {};
      \node [circle,fill,inner sep=1,label={0:(3,2)}] at (180:1) {};
      \node [circle,fill,inner sep=1,label={300:(3,4)}] at (240:1) {};
      \node [circle,fill,inner sep=1,label={240:(5,4)}] at (300:1) {};
    \end{scope}
  \end{tikzpicture}}
  \caption{Representations of $\Bip{\rel C_5}$ and $\Bip{\rel C_6}$.}
    \label{fig:box-k3}
\end{figure}

\begin{example}\label{ex:cycle}
  Let us consider the complex $\Bip{\rel C_k}$. Its vertices are all oriented edges of the $k$-cycle which means pairs of the form $(i,i+1)$ and $(i+1,i)$ where the addition is considered modulo $k$. It is not hard to see that the only directed complete bipartite subgraphs of $\rel C_k$ are either two outgoing edges from a single vertex, or two incoming edges to a single vertex.  The only non-trivial faces of $\Bip{\rel C_k}$ are therefore of the form $\{ (i-1,i),(i+1,i) \}$ or $\{ (i,i-1),(i,i+1) \}$. The resulting complex can be drawn as a graph (see Fig.~\ref{fig:box-k3} for such a drawing of $\Bip{\rel C_5}$). The exact structure depends on the parity of $k$. If $k$ is odd, the complex is a single $2k$-cycle where $(i,j)$ is opposite to $(j,i)$. The \equivariant-action acts as the central reflection. If $k$ is even, the complex consists of two disjoint $k$-cycles such that one contains all edges of the form $(2i,2i\pm 1)$ and the other all edges of the from $(2i\pm 1,2i)$. The \equivariant-action in this case switches the two parts.
\end{example}

\begin{example}
  A slightly more complicated example is $\Bip{\rel K_4}$. See Fig.~\ref{fig:box-k4} for graphical representations of this complex. There are two types of maximal directed complete bipartite subgraphs of $\rel K_4$: either all three in/outgoing edges of a single vertex, or 4 directed edges from a two-element subset of $\rel K_4$ to its complement.
These, and all their non-empty subsets, are the faces of $\Bip{\rel K_4}$. In the pictures, the in/outgoing edges correspond to the triangular faces, and the faces containing 4 edges correspond to tetrahedrons that are represented as tetragons. Naturally, all subsets of these tetragons are also faces, nevertheless they are omitted from the picture for better readability. Also note that the outer face of the left diagram forms such a tetrahedron (corresponding to the bipartite subgraph $\{1,3\} \times \{0,2\}$).
  The \equivariant-symmetry of this complex is given by reversing edges; this corresponds to the antipodality on the spherical representation.
\end{example}

  \begin{figure}[t]
    \centering
    \raisebox{-.5\height}{\begin{tikzpicture} [scale = 1, label distance = -1mm,
                                               every node/.style = {font={\small}}]
      \draw [fill=gray!20!white] (2,2) -- (2,-2) -- (-2,-2) -- (-2,2) -- cycle;

      \node [circle,fill,inner sep=1,label={0:(3,1)}] (east) at (0:1.366) {};
      \node [circle,fill,inner sep=1,label={90:(0,2)}] (north) at (90:1.366) {};
      \node [circle,fill,inner sep=1,label={180:(1,3)}] (west) at (180:1.366) {};
      \node [circle,fill,inner sep=1,label={270:(2,0)}] (south) at (270:1.366) {};
      \node [circle,fill,inner sep=1,label={45:(0,1)}] at (.5,.5) {};
      \node [circle,fill,inner sep=1,label={-45:(2,1)}] at (.5,-.5) {};
      \node [circle,fill,inner sep=1,label={-135:(2,3)}] at (-.5,-.5) {};
      \node [circle,fill,inner sep=1,label={135:(0,3)}] at (-.5,.5) {};
      \node [circle,fill,inner sep=1,label={45:(3,2)}] at (2,2) {};
      \node [circle,fill,inner sep=1,label={-45:(3,0)}] at (2,-2) {};
      \node [circle,fill,inner sep=1,label={-135:(1,0)}] at (-2,-2) {};
      \node [circle,fill,inner sep=1,label={135:(1,2)}] at (-2,2) {};

      \draw (.5,.5) -- (.5,-.5) -- (-.5,-.5) -- (-.5,.5) -- cycle;
      \draw (.5,-.5) -- (east) -- (.5,.5) -- (north) -- (-.5,.5) -- (west) -- (-.5,-.5) -- (south) -- (.5,-.5);
      \draw (2,-2) -- (east) -- (2,2) -- (north) -- (-2,2) -- (west) -- (-2,-2) -- (south) -- (2,-2);

    \end{tikzpicture}}
    \qquad
    \raisebox{-.5\height}{
      \includegraphics[width=.4\textwidth]{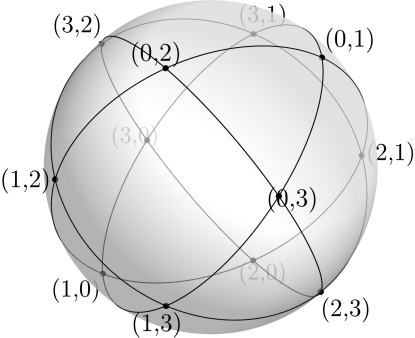}
    }
    \caption{Two representation of $\Bip{\rel K_4}$. The tetragons in both pictures represent tetrahedrons with the same vertices.}
      \label{fig:box-k4}
  \end{figure}

\begin{definition}
  A~\emph{\equivariant-(simplicial) complex} is a~simplicial complex $\cox K$ with a simplicial map ${-}\colon \cox K\to \cox K$ such that $-(-v) = v$ for each $v\in V(\cox K)$.
  We also write $-\sigma$ for the face $\{ -v \mid v\in \sigma \}$.
  We say that a~simplicial map $f$ from one \equivariant-complex $\cox K$ to another \equivariant-complex $\cox K'$ is a~\emph{\equivariant-simplicial map}, if $f(-v) = -f(v)$ for each $v\in V(\cox K)$ (note that the first $-$ is taken in $\cox K$, while the second is taken in $\cox K'$).
\end{definition}

The map $v\mapsto -v$ can be also viewed as an action of the group $\mathbb Z_2$ on $\cox K$ by simplicial maps. We remark that a~product $\cox K_1\times \dots \times \cox K_n$ of \equivariant-complexes is also \equivariant-complex with the action defined component-wise, as $-(v_1,\dots,v_n) = (-v_1,\dots,-v_n)$.
Every graph homomorphism $f\colon \rel H\to \rel G$ induces a~\equivariant-simplicial map $f'\colon \Bip{\rel H} \to \Bip{\rel G}$ defined by $f'((a,b)) = (f(a),f(b))$.\footnote{This is an instance of a more general fact that $f$ induces a simplical map $f'\colon \Hom(\rel K,\rel H) \to \Hom(\rel K,\rel G)$ for each $\rel K$. This map is defined by $f'(g)\colon x \mapsto f(g(x))$.}

\subsection{Topological spaces}

The spaces assigned to \equivariant-complexes inherit the $\mathbb Z_2$ symmetry.

\begin{definition} \label{def:z2-map}
  A~\emph{\equivariant-space} is a~topological space $\top X$ with a~distinguished continuous function ${-}\colon \top X\to \top X$ such that $-(-x) = x$ for each $x$.
  A~\emph{$\ZZ_2$-map} between two \equivariant-spaces $\top X$ and $\top Y$ is a~continuous function $f\colon \top X\to \top Y$ which preserves the action of $\mathbb Z_2$, i.e., $f(-x) = -f(x)$ for each $x\in \top X$ (note that the first $-$ is taken in $\top X$, while the second is taken in $\top Y$).
\end{definition}

As is the case for \equivariant-complexes, \equivariant-spaces are topological spaces with an action of the group $\mathbb Z_2$ by continuous functions.

\begin{example}
  Prime examples of \equivariant-spaces are spheres: We define $\Sphere^n$ as a~subspace of $\mathbb R^{n+1}$ consisting of all unit vectors, i.e., $\Sphere^n = \{ (x_1,\dots,x_{n+1})\in \mathbb R^{n+1} \mid x_1^2 + \dots + x_{n+1}^2 = 1 \}$,
  with \emph{antipodality} as the chosen \equivariant-action, i.e., $-(x_1,\dots,x_{n+1}) = (-x_1,\dots,-x_{n+1})$.
  Other common \equivariant-spaces are toruses. An~$n$-torus $\Torus^n$ is defined as the $n$-th power $\Sphere^1 \times \dots \times \Sphere^1$, and is therefore naturally equipped with a \equivariant-action defined to act coordinatewise.
\end{example}

A~\emph{\equivariant-complex} $\cox K$ is \emph{free} if $-\sigma \neq \sigma$ for each $\sigma \in \cox K$ (equivalently, $\{-v,v\} \not \in \cox K$ for all vertices $v$ of $\cox K$).
Note that, for a loopless undirected graph $\rel G$, the complex $\Bip{\rel G}$ is always a~free \equivariant-complex.
To ease a technical annoyance in the proofs below, we rephrase the definition of a~geometric realisation (see also \cite[Definition 2.27]{Koz08-book}) of a free \equivariant-complex.

\begin{definition} \label{def:geom}
  Let $\cox K$ be a~free \equivariant-simplicial complex. Let $v_1,-v_1,\dots,v_n,-v_n$ be all vertices of $\cox K$.
  We define $\geom{\cox K}$, a \emph{geometric realisation} of $\cox K$, as a subspace of $\mathbb R^n$. First, we identify the canonical unit vectors with $v_1,\dots,v_n$, so that $v_1 = (1,0,\dots,0)$, etc., and $-v_1,\dots,-v_n$ with their opposites, so $-v_1 = (-1,0\dots,0)$, etc.
  Second, for each face $\sigma\subseteq V(\cox K)$, we define $\Delta^\sigma \subseteq \mathbb R^n$ to be the convex hull of $\sigma$, i.e.,
  \(
    \Delta^\sigma = \{ \sum_{v\in \sigma} \lambda_v v \mid \sum_{v\in \sigma} \lambda_v = 1, \lambda_v \geq 0 \}.
  \)
  Finally, we set
  \[
    \geom{\cox K} = \bigcup_{\sigma\in \cox K} \Delta^\sigma
      = \{ \sum_{v\in \sigma} \lambda_v v \mid
        \sigma \in \cox K,
        \sum_{v\in \sigma} \lambda_v = 1,
        \lambda_v \geq 0
      \}.
  \]
  The action of $\mathbb Z_2$ on $\geom{\cox K}$ maps a~point $\sum_{v\in \sigma} \lambda_v v$ to the point $\sum_{{-v}\in {-\sigma}} \lambda_v (-v)$ which can be equivalently described as reversing the sign of a~vector, i.e., as $-(x_1,\dots,x_n) = (-x_1,\dots,-x_n)$.
\end{definition}

With the above definition, we can view $V(\cox K)$ as a subset of $\geom{\cox K}$ --- this will be useful in the technical proofs below. Also note that $-v$ has two meanings that result in the same object: either it is a \equivariant-counterpart of $v \in V(\cox K)$, or the opposite vector to $v\in \geom{\cox K}$.
Note that the geometric realisation of a free \equivariant-complex is a \emph{free \equivariant-space}, i.e., a \equivariant-space $\top X$ such that $-x \neq x$ for all $x\in \top X$.

To express abstractly what it means for two \equivariant-spaces to be the same, we use the notion of \emph{\equivariant-homeomorphism} which is an analogue of the notion of homeomorphism. We remark that
this is a strong notion of equivalence of topological spaces, akin to isomorphism, and that we will also use weaker notions of topological equivalence (see also Appendix~\ref{app:hom-and-box}).

\begin{definition}
Two \equivariant-topological spaces $\top X$ and $\top Y$ are \emph{\equivariant-homeomorphic} if there are \equivariant-maps $f\colon \top X\to \top Y$ and $g\colon \top Y\to \top X$ such that $fg$ is the identity on $\top Y$ and $gf$ is the identity on $\top X$.
\end{definition}

\begin{example} \label{ex:hom-k3}
  It is not hard to see that the geometric representation $\gBip{\rel C_k}$ of the homomorphism complex of an odd cycle $\rel C_k$ is \equivariant-homeomorphic to $\Sphere^1$ (see Fig.~\ref{fig:box-k3sq} on page~\pageref{fig:box-k3sq}). Let us define one such \equivariant-homeomorphism $f\colon \gBip{\rel C_k} \to \Sphere^1$. Choose $k$ points on the circle in a~regular pattern. Let us denote these vectors $x_1,\dots,x_k$. Note that since $k$ is odd, $-x_i \notin \{x_1,\dots,x_k\}$ for all $i$. We first define a~map $f_0\colon \gBip{\rel C_k} \to \mathbb R^2$ as follows:
  \(
    f_0(v) = x_b - x_a
  \)
  for $v\in V(\cox K)$, $v = (a,b)$, and extend it linearly. Note that the image of $\gBip{\rel C_k}$ forms a regular $2k$-gon centred in the origin. We project the polygon onto $\Sphere^1$ by putting $f(x) = f_0(x)/\size{f_0(x)}$. It is clear that $f$ is continuous and $f(-x) = -f(x)$, and therefore it is a \equivariant-map. It is also not hard to see that it is 1-to-1 and therefore invertible, and that the inverse is a \equivariant-map.
\end{example}

\begin{remark}
  While there is always a~continuous function between two topological spaces $\top X$ and $\top Y$ (simply map everything to a~single point), there might not be a \equivariant-map between two \equivariant-spaces. In particular, the Borsuk-Ulam theorem \cite{Bor33} (see also~\cite{Mat03}) states that there is no \equivariant-map from a~sphere $\Sphere^m$ to a~sphere $\Sphere^n$ of smaller dimension (i.e., if $m > n$).
\end{remark}

Every \equivariant-simplicial map $f\colon \cox K \to \cox K'$ induces a \equivariant-map $\geom f\colon \geom{\cox K} \to \geom{\cox K'}$ defined as a piece-wise linear extension of~$f$:
\[
  \geom f\colon \sum_{v\in \sigma} \lambda_v v \mapsto \sum_{v\in \sigma} \lambda_v f(v).
\]
(Here, we use that $v\in V(\cox K)$ is also a point in $\geom {\cox K}$, and therefore a vector in $\mathbb R^n$.)
Consequently, every graph homomorphism $\rel H \to \rel G$ induces a \equivariant-map from $\geom{\Bip{\rel H}}$ to $\geom{\Bip{\rel G}}$.

\subsubsection{The fundamental group}

We briefly recall the definition of the fundamental group assigned to a topological space $\top X$, denoted $\Fg {\top X}$. For more details, see \cite[Chapter 1]{Hat01}.
The elements of the group are \emph{homotopy classes} of maps $f\colon \Sphere^1 \to \top X$ defined as follows.
Intuitively, two maps are \emph{homotopic} if one can be continuously transformed into the other.

\begin{definition}
We say that two continuous maps $f,g\colon \top X\to \top Y$ are \emph{homotopic} if there is a~continuous map $h\colon \top X \times [0,1] \to \top Y$ such that $h(x,0) = f(x)$ and $h(x,1) = g(x)$ for each $x\in \top X$. Any such map $h$ is called a~\emph{homotopy}. The \emph{homotopy class of $f\colon \top X\to \top Y$} is the set of all continuous maps $g\colon \top X\to \top Y$ that are homotopic to $f$. We denote such a~class by $[f]$.\footnote{We use notation $[*]$ both for sets $\{1,\ldots,n\}$ and for homotopy classes; the meaning will always be clear from the context.}
\end{definition}

Formally, the fundamental group is defined relative to a point $x_0 \in \top X$, but the choice of the point is irrelevant if the space $\top X$ is path connected (see \cite[Proposition 1.5]{Hat01}), i.e., if any two points in $\top X$ are connected by a path. Fix one such choice $x_0\in \top X$. The elements of $\Fg{\top X}$ are all homotopy classes of maps $\ell\colon \Sphere^1 \to \top X$ such that $\ell((1,0)) = x_0$. The group operation is given by so-called \emph{loop composition}: seeing maps $\ell_1,\ell_2\colon \Sphere^1 \to {\top X}$ as closed walks originating in $x_0$, the product $\ell_1\cdot \ell_2$ is the closed walk that follows first $\ell_1$ and then $\ell_2$. While this product is not a group operation as is, it induces a group operation on the homotopy classes defined as $[\ell_1]\cdot[\ell_2] = [\ell_1\cdot \ell_2]$ (see \cite[Proposition 1.3]{Hat01} for a proof).

Any map $f\colon {\top X}\to \top Y$ induces a~group homomorphism $f_*\colon \Fg {\top X}\to \Fg {\top Y}$ defined by $f_*([\ell]) = [f\ell]$ for each $\ell\colon \Sphere^1 \to \top X$.

The fundamental groups of many spaces are described in the literature. For example:

\begin{lemma}[{\cite[Theorem 1.7]{Hat01}}] \label{lem:fg-s1}
  $\Fg{\Sphere^1}$ is isomorphic to $\mathbb Z$.
\end{lemma}

We also define \emph{\equivariant-homotopy} which is a strengthening of homotopy, restricting it to \equivariant-maps.

\begin{definition}\label{def:z2-homotopy}
Let $f,g\colon \top X \to \top Y$ be \equivariant-maps. A homotopy $h$ from $f$ to $g$ is a \emph{\equivariant-homotopy} if the map $h_t\colon x\mapsto h(x,t)$ is a \equivariant-map for each $t\in [0,1]$. We say that $f$ and $g$ are \emph{\equivariant-homotopic}, if there is a \equivariant-homotopy between them.
\end{definition}

\subsection{Polymorphisms of complexes, spaces, and groups}

A~polymorphism from one graph to another is defined as a homomorphism from a~power. In the same way, we can define polymorphisms of any objects as long as we have a notion of a homomorphism and of a power.

\begin{definition}
  \begin{enumerate}
    \item Let $\cox K,\cox K'$ be two \equivariant-simplicial complexes. An $n$-ary polymorphism from $\cox K$ to $\cox K'$ is a \equivariant-simplicial map from the $n$-th power of $\cox K$ to $\cox K'$, i.e.,
      $f\colon V(\cox K)^n \to V(\cox K')$ such that $f(-v_1,\dots,-v_n) = -f(v_1,\dots,v_n)$ for all $v_i\in V(\cox K)$ and
      \[
        f(\sigma_1 \times \dots \times \sigma_n) = \{ f( a_1, \dots, a_n ) \colon a_i \in \sigma_i \} \text{ is in } \cox K'
      \]
      for all $\sigma_1,\dots,\sigma_n \in \cox K$. We denote by $\Pol(\cox K,\cox K')$ the set of all polymorphisms from $\cox K$ to $\cox K'$.
    \item Let $\top X,\top Y$ be two \equivariant-spaces. An $n$-ary polymorphism from $\top X$ to $\top Y$ is a~\equivariant-map from $\top X^n$ to $\top Y$, i.e., a~continuous map $f\colon \top X^n \to \top Y$ such that $f(-x_1,\dots,-x_n) = -f(x_1,\dots,x_n)$ for all $x_i\in X$. Again, $\Pol(\top X,\top Y)$ denotes the set of all polymorphisms from $\top X$ to $\top Y$.
    \item Let $\gr H$, $\gr G$ be two groups. An $n$-ary polymorphism from $\gr H$ to $\gr G$ is a group homomorphism from $\gr H^n$ to $\gr G$, i.e., a mapping $f\colon H^n \to G$ such that
      \[
        f( g_1\cdot h_1, \dots, g_n\cdot h_n ) = f( g_1,\dots,g_n) \cdot f(h_1,\dots,h_n)
      \]
      for all $g_i,h_i \in H$. We denote the set of all polymorphisms from $\gr H$ to $\gr G$ by $\Pol(\gr H,\gr G)$.
  \end{enumerate}
\end{definition}

In all the cases above, it is easy to see that polymorphisms are closed under taking minors, and therefore $\Pol({-},{-})$ is always a~minion. This allows us to talk about minion homomorphisms between minions of polymorphisms of different objects (graphs, simplicial complexes, topological spaces, or groups).

\begin{example}
  By definition, $\Pol(\mathbb Z,\mathbb Z)$ consists of all group homomorphisms from $\mathbb Z^n$ to $\mathbb Z$ for all $n > 0$.
  It is straightforward to check that such an $n$-ary polymorphism in $\Pol(\mathbb Z,\mathbb Z)$ is a linear function, i.e., of the form $(x_1,\dots,x_n) \mapsto \sum_{i=1}^n c_i x_i$ for some $c_1,\dots,c_n \in \mathbb Z$, and conversely, any such function is a~group homomorphism from $\mathbb Z^n$ to $\mathbb Z$.
\end{example}

\subsection{Proofs of Theorems~\ref{thm:K3} and~\ref{thm:main-s1}}

We recall the statement of Theorem~\ref{thm:main-s1}.

\maintheoremtopology*

Theorem \ref{thm:K3} is a~direct corollary of the above and Example~\ref{ex:hom-k3}.
In the rest of the section we prove Theorem~\ref{thm:main-s1} by using Theorem~\ref{thm:bounded-arity}. We show that there is a~minion homomorphism from $\Pol(\rel H,\rel G)$ to the minion $\clo Z_{\leq N}$ for some $N$. Recall that
 $\clo Z_{\leq N}\subset \Pol(\mathbb Z,\mathbb Z)$ is defined to consist of all linear functions $f$ on $\mathbb Z$ of the form
\[
  f(x_1,\dots,x_n) = \sum_{i=1}^n c_i x_i
\]
for some $c_i\in \mathbb Z$ such that $\sum_{i=1}^N c_i$ is odd and $\sum_{i=1}^N \size{c_i} \leq N$.
This is achieved in two steps.

In the first step, we provide a minion homomorphism from $\Pol(\rel H,\rel G)$ to $\Pol(\mathbb Z,\mathbb Z)$. This is achieved by following the constructions described above, i.e., the transformations
\[
  \text{graph} \buildrel {\Bip\ast}\over\longrightarrow
  \text{\equivariant-complex} \buildrel {\geom\ast}\over\longrightarrow
  \text{\equivariant-space} \buildrel {\Fg\ast}\over\longrightarrow
  \text{group},
\]
and showing that pushing a~polymorphism through this sequence of constructions preserves minors. This essentially follows from the well-known facts that these constructions behave well with respect to products. A detailed proof is presented in Section~\ref{sec:minion-homomorphism}.

When we push a polymorphism $f\in \Pol(\rel H,\rel G)$ through these constructions, we first obtain a $\ZZ_2$-simplicial map $f'\in \Pol(\Bip{\rel H},\Bip{\rel G})$, which in turn induces a $\ZZ_2$-map $\geom{f'} \in \Pol(\gBip{\rel H},\gBip{\rel G})$. Then, by composing with the assumed $\ZZ_2$-map
\(
  s \colon \gBip{\rel G} \to \Sphere^1
\)
(and assuming without loss of generality that $\rel H$ is an odd cycle), we obtain a polymorphism $g$ of $\Sphere^1$, and then finally a polymorphism $g_*$ of the group $\pi_1(\Sphere^1) \simeq \mathbb Z$.
As discussed before, $g_*$ can be described more concretely as a linear function whose coefficients $c_i$ are the winding numbers of maps $\Sphere^1 \to \Sphere^1$ defined by $t \mapsto s \circ g(x_0,\dots,t,\dots,x_0)$ where $x_0\in \Sphere^1$ is an arbitrary (but fixed) point.

In the second step, we use the discrete structure of the graphs $\rel H$ and $\rel G$, as well as the action of $\ZZ_2$, to show that  the image of $\Pol(\rel H,\rel G)$ under the constructed minion homomorphism is contained in $\clo Z_{\leq N}$ for some $N$. This is described in Section~\ref{sec:bounding-arity}.

We note that there are several ways to present the proof of Theorem~\ref{thm:K3}.
These presentations would look different on the surface, but in fact they use the same underlying topological concepts, just hidden to various extent.  For example, the proof that appeared in the conference version \cite{KO19} hides topology in a more direct combinatorial approach. Yet another version of the proof can be given in the language of \emph{recolourings}: the required minion homomorphism would map two polymorphisms $f$ and $f'$ from $\Pol(\rel H,\rel G)$ to the same function if and only if $f$ can be recoloured to $f'$ by changing one output value at a time, that is, if there is a sequence of polymorphisms $f_0,\ldots,f_n\in \Pol(\rel H,\rel G)$ such that $f_0=f, f_n=f'$ and, for each $i\in \{0,\dots,n-1\}$, there is a unique tuple $\bar{t}_i$ with $f_i(\bar{t}_i)\ne f_{i+1}(\bar{t}_i)$.
However, as shown in \cite{Wrochna15}, recolourability is inherently a topological notion, so the resulting proof would have the same essence as the one presented here.
We chose the current presentation because we believe that it reflects the `true essence' of the proof.

\subsubsection{A minion homomorphism}
  \label{sec:minion-homomorphism}

\begin{figure}
  \centering
  \raisebox{-.5\height}{\begin{tikzpicture} [scale = 1.3]
    \draw (0,0) circle (1cm);
    \foreach \i/\c/\d in {0/(0/4),36/(0/1),72/(2/1),108/(2/3),144/(4/3),180/(4/0),216/(1/0),252/(1/2),288/(3/2),324/(3/4)}
      \node [circle,fill,inner sep=1,label={\i:\c,\d}] at (\i:1) {};
  \end{tikzpicture}}
  \qquad
  \raisebox{-.5\height}{\includegraphics[width=.4\textwidth]{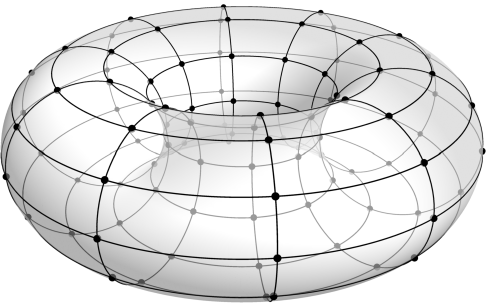}
  }
  \caption{A representation of $\Bip{\rel C_5}$ as $\Sphere^1$ and $\Bip {\rel C_5^2}$ as $\Torus^2$.}
    \label{fig:box-k3sq}
\end{figure}

As mentioned above, the required minion homomorphism is obtained as a composition of three mappings.
The first one is a minion homomorphism from polymorphisms of graphs to polymorphisms of simplicial complexes. We implicitly use that for any graphs $\rel H_1$ and $\rel H_2$, there is a natural isomorphism of \equivariant-simplicial complexes
\[
  \Hom(\rel K_2, \rel H_1) \times \Hom(\rel K_2, \rel H_2) \simeq
  \Hom(\rel K_2, \rel H_1 \times \rel H_2)
\]
given by the \equivariant-simplicial map:
\(
  ((a,b),(a',b')) \mapsto ((a,a'),(b,b')).
\)
See Fig.~\ref{fig:box-k3sq} for an example.

\begin{lemma} \label{lem:mu_1}
  For graphs $\rel H$, $\rel G$, the mapping
  \[
    \mu_1\colon \Pol(\rel H,\rel G) \to \Pol(\Hom (\rel K_2,\rel H),\Hom(\rel K_2,\rel G))
  \]
  defined as
  \[
    \mu_1(f) ( (u_1,v_1),\dots,(u_n,v_n) ) := ( f(u_1,\dots,u_n), f(v_1,\dots,v_n) )
  \]
  is a~minion homomorphism.
\end{lemma}

\begin{proof}
  Let us first check that $\mu_1(f)$ is indeed a simplicial map. Assume that $\sigma_1,\dots,\sigma_n$ are faces in $\Hom(\rel K_2,\rel H)$, i.e., $\sigma_i$ is a subset of edges of some complete directed bipartite subgraph of $\rel H$. We may assume without loss of generality that $\sigma_i$ is the set of all edges of a complete directed bipartite subgraph of $\rel H$ which gives $\sigma_i = U_i\times V_i$ for some $U_i,V_i \subseteq V(G)$, which form a bipartition of some complete bipartite subgraph of $\rel G$. By definition,
  \[
    \mu_1(f)(\sigma_1 \times \dots \times \sigma_n) = \{ (f(u_1,\dots,u_n),f(v_1,\dots,v_n)) : u_i\in U_i, v_i\in V_i \} = f(U)\times f(V),
  \]
  where $U = U_1\times \cdots \times U_n$ and $V = V_1\times \cdots\times V_n$. By the definition of graph product, all edges between $U$ and $V$ are present in $\rel H^n$. Consequently, $f(U)\times f(V)$ is a complete directed bipartite subgraph of $\rel G$,  since $f$ is a~polymorphism. This implies that $\mu_1(f)(\sigma_1 \times \dots \times \sigma_n)$ is a face of $\Hom(\rel K_2,\rel G)$.

  It is straightforward that the mapping $\mu_1$ preserves both the \equivariant-action and taking minors.
\end{proof}

The next step is from \equivariant-simplicial complexes to \equivariant-spaces. The map that we construct will not be a minion homomorphism, it will satisfy a weaker condition that will be sufficient later.

\begin{definition}
Assume that $\top X,\top Y$ are \equivariant-spaces and let $\clo M$ be a~minion. We say that a~mapping $\xi\colon \clo M \to \Pol(\top X,\top Y)$ \emph{preserves minors up to \equivariant-homotopy} if for all $n,m>0$, $f\in \clo M^{(n)}$, and $\pi\colon [n]\to [m]$, we have that $\xi(f^\pi)$ is \equivariant-homotopic to $\xi(f)^\pi$.
\end{definition}

We recall that the points in the geometric representation of $\cox K$ can be viewed as convex combinations of vertices of $\cox K$, more precisely
\(
  \geom{\cox K} = \{ \sum_{v\in \sigma} \lambda_v v \mid \sigma \in \cox K, \lambda_v \geq 0, \sum_{v\in \sigma} \lambda_v = 1 \}
\) and the representation of each point is unique (as we identified $v \in V(K)$ with affinely independent vectors in $\RR^{V(K)}$, namely the basis vectors).
This is used in the following lemma.

\begin{lemma} \label{lem:mu_2}
  Let $\cox K$, $\cox K'$ be two \equivariant-simplicial complexes. Let $\mu_2\colon \Pol(\cox K,\cox K') \to \Pol(\geom {\cox K},\geom {\cox K'})$ be the linear extension, i.e., $\mu_2$ takes $f \in  \Pol(\cox K,\cox K')$ to
  \[
    \mu_2(f)\colon ( \sum_{v\in \sigma_1} \lambda_{1,v} v, \dots, \sum_{v\in \sigma_n} \lambda_{n,v} v )
    \mapsto \sum_{v_1 \in \sigma_1,\dots,v_n\in \sigma_n} \lambda_{1,v_1}\cdots\lambda_{n,v_n} f(v_1,\dots v_n)
  \]
  for a~point in $\Delta^{\sigma_1}\times\dots\times \Delta^{\sigma_n} \subseteq \geom{\cox K}^n$.
  Then $\mu_2$ preserves minors up to \equivariant-homotopy.
\end{lemma}

\begin{proof}
  Let $f \colon \cox K^n \to \cox K'$ be a~simplicial map and pick any $\pi\colon [n] \to [m]$.
  Then
  \begin{align*}
    \mu_2(f^\pi)(\sum_{v\in \sigma_1} \lambda_{1,v} v, \dots, \sum_{v\in \sigma_m} \lambda_{m,v} v)
      &= \sum_{v_1 \in \sigma_1,\dots,v_m\in \sigma_m} \lambda_{1,v_1}\cdots\lambda_{m,v_m} f^\pi(v_1,\dots v_m)\\
      &= \sum_{v_1 \in \sigma_1,\dots,v_m\in \sigma_m} \lambda_{1,v_1}\cdots\lambda_{m,v_m} f(v_{\pi(1)},\dots v_{\pi(n)}).
  \end{align*}
  On the other hand, if we take the induced map first and then the minor, we obtain:
  \begin{multline*}
    \mu_2(f)(\sum_{v\in \sigma_{\pi(1)}} \lambda_{\pi(1),v} v, \dots, \sum_{v\in \sigma_{\pi(n)}} \lambda_{\pi(n),v} v) \\
      = \sum_{v_1 \in \sigma_{\pi(1)},\dots,v_n\in \sigma_{\pi(n)}}
      \lambda_{\pi(1),v_1}\cdots\lambda_{\pi(n),v_n} f(v_1,\dots v_n).
  \end{multline*}
  Both points lie in $\Delta^\sigma \subseteq \geom{\cox K'}$ for $\sigma = \{f(v_1,\dots,v_n) \mid v_i\in \sigma_{\pi(i)}\} \in \cox K'$. We can thus continuously move from one to the other.
  Formally, we define a~homotopy $h\colon \geom {\cox K}^m \times [0,1] \to \geom {\cox K'}$ by
  \[
    \textstyle h(x_1,\dots,x_m,t) =
      t \mu_2(f)(x_{\pi(1)},\dots,x_{\pi(n)}) + (1-t) \mu_2(f^\pi)(x_1,\dots,x_m).
  \]
  It is clear that $h_t\colon \tup x \mapsto h(\tup x,t)$ is a~well-defined \equivariant-map, and therefore $h$ is the required $\ZZ_2$-homotopy.
  (As a side note, observe that the homotopy is constant on vertices: for any vertices $v_1,\dots,v_m$ of $\cox K$ and each $t\in [0,1]$, $h_t(v_1,\dots,v_m)$ is equal to $f^\pi(v_1,\dots,v_m) = f(v_{\pi(1)},\dots,v_{\pi(n)})$.)
\end{proof}

\begin{lemma} \label{lem:r}
  Let $\rel H$ be a non-bipartite graph. Then there is a~\equivariant-map $r\colon \Sphere^1 \to \gBip {\rel H}$.
\end{lemma}

\begin{proof}
  Since the graph $\rel H$ is not bipartite, there is a homomorphism $h\colon \rel C_k \to \rel H$ for some odd $k$. This induces a \equivariant-map $\geom{h'}\colon \gBip{\rel C_k} \to \gBip{\rel H}$, and since $\gBip{\rel C_k}$ is \equivariant-homeomorphic to $\Sphere^1$ (see Example \ref{ex:hom-k3}), the claim follows.
\end{proof}

\begin{lemma} \label{lem:mu}
  Let $\rel H,\rel G$ be two graphs such that $\rel H$ is non-bipartite, $\rel H \to \rel G$, and there is a \equivariant-map $s\colon \gBip{\rel G} \to \Sphere^1$.
  Then $\mu\colon \Pol({\rel H}, {\rel G}) \to \Pol( \Sphere^1, \Sphere^1 )$ defined as
  \[
    \mu(f)(x_1,\dots,x_n) := s(\mu_2\mu_1(f)(r(x_1),\dots,r(x_n)))
  ,\]
  where $\mu_1$, $\mu_2$, and $r$ are from Lemmas~\ref{lem:mu_1}, \ref{lem:mu_2}, and \ref{lem:r}, respectively, preserves minors up to \equivariant-homotopy.
\end{lemma}

\begin{proof}
  Assume that $f$ is a~polymorphism from $\rel H$ to $\rel G$ of arity $n$ and let $\pi\colon [n]\to [m]$. We want to prove that $\mu(f)^\pi$ is \equivariant-homotopic to $\mu(f^\pi)$.
  From Lemmas~\ref{lem:mu_1} and \ref{lem:mu_2}, we have that $\mu_2\mu_1(f)^\pi$ and $\mu_2\mu_1(f^\pi)$ are \equivariant-homotopic; let $h'$ be a~\equivariant-homotopy that witnesses this fact.
  We define a~\equivariant-homotopy $h\colon \Torus^m \times [0,1] \to \Sphere^1$ by
  \[
    h(x_1,\dots,x_m,t) = sh'(r(x_1),\dots,r(x_m),t)
  .\]
  Indeed, for $t=0$,
  \[
    h(x_1,\dots,x_m,0) = s\left(\mu_2\mu_1(f)\left(r(x_{\pi(1)}), \dots, r(x_{\pi(n)})\right)\right) = \mu(f)^\pi(x_1,\dots,x_m)
  \]
  while for $t=1$,
  \[
    h(x_1,\dots,x_m,1) = s\big(\mu_2\mu_1(f^\pi)\big(r(x_1),\dots,r(x_m)\big)\big) = \mu(f^\pi)(x_1,\dots,x_m).
  \]
  This concludes the proof.
\end{proof}

The final step is from \equivariant-spaces to the fundamental groups. Recall that, for a~continuous function $f:\top X\to\top Y$, we have a group homomorphism $f_*\colon \Fg{\top X}\to \Fg{\top Y}$ defined as $f_*([\ell]) = [f\ell]$.
We will also need a group homomorphism from $\Fg{\top X}^n$ to $\Fg{\top X^n}$, that is guaranteed to exist for any path connected space $\top X$ by {\cite[Proposition 1.12]{Hat01}}. One such homomorphism is the mapping $e_n\colon \Fg {\top X}^n \to \Fg{\top X^n}$ defined as
\[
  e_n( [\ell_1],\dots,[\ell_n] ) := [ t \mapsto (\ell_1(t),\dots,\ell_n(t)) ]
\]
for $\ell_1,\dots,\ell_n\colon \Sphere^1 \to \top X$.

\begin{lemma} \label{lem:nu}
Let $\top X,\top Y$ be two path connected \equivariant-spaces. Then the mapping
\[
  \nu(f) := f_* e_n
\]
is a minion homomorphism from $\Pol(\top X,\top Y)$ to $\Pol(\Fg {\top X},\Fg
  {\top Y})$. Moreover, if $f$ and $g$ are homotopic then $\nu(f) = \nu(g)$.
\end{lemma}

\begin{proof}
   Assume that $f\in \Pol(\top X,\top Y)$ is of arity $n$, $\pi\colon [n] \to [m]$, and $\ell_1,\dots,\ell_m\colon \Sphere^1 \to \top X$.
   To simplify notation, let $\ell(t) = (\ell_1(t),\dots,\ell_m(t))$ go from $\Sphere^1$ to $\top X^n$ and $\ell_\pi(t) = (\ell_{\pi(1)}(t),\dots,\ell_{\pi(n)}(t))$ go from $\Sphere^1$ to $\top X^m$.
   Using the definitions of $f_*$ and $e_n$, we get
   \[
      \nu(f^\pi) (\ell_1,\dots,\ell_m)
      = [ f^\pi\ell ]
      = [ f\ell_\pi ]
      = \nu(f)^\pi (\ell_1,\dots,\ell_m)
   .\]
   Finally, since $f_* = g_*$ if $f$ and $g$ are homotopic, we also get that $\nu(f) = \nu(g)$.
\end{proof}

We recall that $\Fg{\Sphere^1} \simeq \mathbb Z$ (Lemma~\ref{lem:fg-s1}). In the following statement we identify the two isomorphic groups to obtain a minion homomorphism to $\Pol(\mathbb Z,\mathbb Z)$, the minion of all linear functions on $\mathbb Z$.

\begin{corollary}
  Let $\rel H,\rel G$ be two graphs such that $\rel H$ is non-bipartite, $\rel H \to \rel G$, and there is a \equivariant-map $s\colon \gBip{\rel G} \to \Sphere^1$.
  The mapping
  \(
    \nu\mu
  \)
  is a~minion homomorphism from $\Pol(\rel H,\rel G)$ to $\Pol(\mathbb Z,\mathbb Z)$ assuming $\mu$ is as in Lemma~\ref{lem:mu} and $\nu$ is as in Lemma~\ref{lem:nu}.
\end{corollary}

\begin{proof}
  Clearly, $\nu\mu\colon \Pol(\rel H,\rel G) \to \Pol(\mathbb Z,\mathbb Z)$ is a well-defined mapping that preserves arities. We need to show that it also preserves minors. This follows from the facts that $\mu$ preserves minors up to \equivariant-homotopy (Lemma \ref{lem:mu}) and that $\nu$ is a~minion homomorphism that is constant on \equivariant-homotopy classes (Lemma \ref{lem:nu}). More precisely, assume $f\in \Pol(\rel H,\rel G)$ is of arity $n$ and $\pi\colon [n] \to [m]$. Then $\mu(f^\pi)$ and $\mu(f)^\pi$ are \equivariant-homotopic, and therefore
  \[
    \nu(\mu(f^\pi)) = \nu(\mu(f)^\pi) = \nu(\mu(f))^\pi
  \]
  where the second equality follows from minor preservation by $\nu$.
\end{proof}

We remark that, for any ($n$-ary) function $f\in \Pol(\rel H,\rel G)$, the coefficients of the linear function $\nu\mu(f)=\sum_{i=1}^n {c_ix_i}$ can be naturally thought of as the degrees of $f$ at the corresponding coordinates.  Such degrees can be defined in a combinatorial way (see~\cite{KO19}) --- the intuitions in that approach are still topological, but the technical proofs become somewhat ad-hoc.

\subsubsection{Bounding essential arity}
  \label{sec:bounding-arity}

To finish the analysis of polymorphisms from $\rel H$ to $\rel G$ necessary for applying Theorem~\ref{thm:bounded-arity}, we need to bound the essential arity of functions in the image of $\nu\mu$ (defined above) and show that none of these functions is a constant function.
We achieve this by proving that the image of $\nu\mu$ is contained in the~minion $\clo Z_{\leq N}$ for some $N$. Recall that this minion is defined to be the set of all functions $f\colon \mathbb Z^n \to \mathbb Z$ of the form
\(
  f(x_1,\dots,x_n) = c_1x_1 +\dots + c_nx_n
\)
for some $c_1$, \dots, $c_n\in \mathbb Z$ with $\sum_{i=1}^n \lvert c_i\rvert \leq N$ and $\sum_{i=1}^n c_i$ odd.

The oddness of the sum of coefficients follows from a well-known fact about \equivariant-maps on~$\Sphere^1$. We recall that \emph{the degree} of a map $f\colon \Sphere^1 \to \Sphere^1$ is the integer $d_f$ such that the induced map $f_*$ on $\Fg{\Sphere^1} = \mathbb Z$ is $x\mapsto d_f \cdot x$.

\begin{lemma}[{\cite[Proposition 2B.6]{Hat01}}]
  \label{lem:odd}
  The degree of any \equivariant-map $f\colon \Sphere^1\to \Sphere^1$ is odd.
\end{lemma}

\begin{lemma} \label{lem:odd-2}
  Let $\rel H,\rel G$ be two graphs such that $\rel H$ is non-bipartite, $\rel H
  \to \rel G$, and there is a \equivariant-map $s\colon \gBip{\rel G} \to
  \Sphere^1$. Let $\mu$ be as in Lemma~\ref{lem:mu}, $\nu$ as in Lemma~\ref{lem:nu}, and let $\clo H = \Pol(\rel H,\rel G)$.
  If $f \in \Pol(\rel H,\rel G)$ and
  \(
    \nu\mu(f)\colon (x_1,\dots,x_n) \mapsto \sum_{i=1}^n c_ix_i,
  \)
  then $\sum_{i=1}^n c_i$ is odd.
\end{lemma}

\begin{proof}
  Consider the (unique) unary minor $h(x) := f(x,\dots,x)$ of $f$. Since $\mu(h)$ is a~\equivariant-map, by Lemma~\ref{lem:odd} it has an odd degree, i.e., $\nu\mu(h)\colon x \mapsto d_hx$ for some odd $d_h$. Finally, since $\nu\mu$ preserves minors, we get that $d_h = \sum_{i\in [n]} c_i$ which we wanted to show to be odd.
\end{proof}

The bound on the sum of absolute values of coefficients is given by the discrete structure of the involved graphs. The key here is that there are only finitely many polymorphisms of a~fixed arity between two given finite graphs.

\begin{lemma} \label{lem:bound}
  Let $\clo M$ be a minion on finite sets $A,B$. Assume that $\xi\colon \clo M \to \Pol(\mathbb Z,\mathbb Z)$ is a minion homomorphism. Then there exists $N$ such that for all $f\in \clo M$,
  if $\xi(f)\colon (x_1,\dots,x_n) \mapsto \sum_{i=1}^n c_ix_i$,
  then $\sum_{i=1}^n \lvert c_i \rvert \leq N$.
\end{lemma}

\begin{proof}
  We first consider binary functions. There are only finitely many functions $f\in \clo M^{(2)}$, so clearly the sum of the absolute values of the coefficients of $\xi(f)$ is bounded by some $N$. We argue that the same $N$ provides a bound for all other arities as well. Let $f\in \clo M$ and
  \(
    \xi(f)(x_1,\dots,x_n) = c_1x_1 +\dots + c_nx_n
  \).
  Let $\sigma\colon [n] \to \{0,1\}$ be defined as
  \[
    \sigma(i) = \begin{cases}
      0 & \text{if $c_i \leq 0$, and} \\
      1 & \text{if $c_i > 0$.}
    \end{cases}
  \]
  I.e., $\sigma^{-1}(1)$ is the set of all coordinates of $\xi(f)$ with positive coefficients, and $\sigma^{-1}(0)$ of those with negative or zero coefficients. Now, let $g = f^\sigma$, that is, $g$ is the minor of $f$ defined by
  \[
    g(x_0,x_1) = f(x_{\sigma(1)},\dots,x_{\sigma(n)}),
  \]
  so $g$ is obtained by identifying all variables of $f$ that induce positive coefficients on $\xi(f)$, and also all those that induce negative coefficients. We let $\xi(g)(x_0,x_1) = c^+ x_1 + c^- x_0$. Since $\xi$ preserves minors, we get that
  $c^+ = \sum_{c_i > 0} c_i$ and $c^- = \sum_{c_i < 0} c_i$. Note that $c^+ \geq 0$ and $c^- \leq 0$. Finally,
  \[
    \sum_{i\in [n]} \lvert c_i \rvert = \sum_{c_i > 0} c_i - \sum_{c_i < 0} c_i = c^+ - c^- = \lvert c^+ \rvert + \lvert c^- \rvert \leq N
  \]
  where the last inequality follows from the definition of $N$. \end{proof}

We can now finish the proof of Theorem~\ref{thm:main-s1}.

\begin{proof}[Proof of Theorem~\ref{thm:main-s1}]
  We assume that $\rel H,\rel G$ are two graphs such that $\rel H$ is non-bipartite, $\rel H
  \to \rel G$, and there is a \equivariant-map $s\colon \gBip{\rel G} \to
  \Sphere^1$. Let $\mu$ be as in Lemma~\ref{lem:mu}, $\nu$ as in Lemma~\ref{lem:nu}, and $N$ be the bound obtained from Lemma~\ref{lem:bound} for $\clo M = \Pol(\rel H,\rel G)$ and $\xi = \nu\mu$. We claim that $\nu\mu(f) \in \clo Z_{\leq N}$ for each $f\in \Pol(\rel H,\rel G)$. Assume $\nu\mu(f)\colon (x_1,\dots,x_n) \mapsto \sum_{i=1}^n c_ix_i$. We have that $\sum_{i=1}^n c_i$ is odd from Lemma~\ref{lem:odd-2}, and $\sum_{i=1}^n \lvert c_i \rvert \leq N$ by the choice of $N$. This concludes that there is a minion homomorphism from $\Pol(\rel H,\rel G)$ to a minion of bounded essential arity, namely $\clo Z_{\leq N}$, and thus Theorem~\ref{thm:main-s1} follows from Theorem~\ref{thm:bounded-arity}.
\end{proof}

\subsection{Proofs of Corollaries~\ref{thm:circular-cliques} and~\ref{thm:square-free}}

 To show that Theorem \ref{thm:main-s1} implies Corollaries \ref{thm:circular-cliques} and \ref{thm:square-free}, we need the following facts about the structure of $\gBip{\rel G}$ for the relevant graphs $\rel G$. These facts seem to be folklore, but we include proofs for completeness.

\begin{figure}
  \centering
  \begin{tikzpicture}[scale = 1.2, baseline={([yshift=-.5ex]current bounding box.center)}]
    \draw [gray] (0,0) circle (1cm);

    \foreach \i/\c in {0/0,51.42/1,102.85/2,205.71/4,257.14/5}
      \node (\c) [circle,fill=red,inner sep=1.5,label={\i:$\c$}] at (\i:1) {};
    \foreach \i/\c in {154.28/3,308.57/6}
      \node (\c) [circle,draw,fill=white,inner sep=1.5,label={\i:$\c$}] at (\i:1) {};

    \foreach \i/\j in {1/3,3/5,4/6,6/1,0/3,3/6,6/2}
      \draw (\i) -- (\j);
    \foreach \i/\j in {4/2,5/0,4/1,5/2,4/0,5/1}
      \draw [very thick,->,red] (\i) -- (\j);
  \end{tikzpicture}
  \quad
  \begin{tikzpicture}[scale = 1.2, baseline={([yshift=-.5ex]current bounding box.center)}]
      \fill [gray!20!white] (-1.22,0.97) -- (-1.9,0.43) -- (-1.9,-0.43) -- (-1.22,-0.97) -- cycle;
\fill [gray!20!white] (1.22,-0.97) -- (1.9,-0.43) -- (1.9,0.43) -- (1.22,0.97) -- cycle;
\fill [gray!20!white] (-1.52,-0.35) -- (-1.52,-1.22) -- (-0.85,-1.76) -- (0.0,-1.56) -- cycle;
\fill [gray!20!white] (1.52,0.35) -- (1.52,1.22) -- (0.85,1.76) -- (0.0,1.56) -- cycle;
\fill [gray!20!white] (-0.68,-1.41) -- (0.0,-1.95) -- (0.85,-1.76) -- (1.22,-0.97) -- cycle;
\fill [gray!20!white] (0.68,1.41) -- (0.0,1.95) -- (-0.85,1.76) -- (-1.22,0.97) -- cycle;
\fill [gray!20!white] (0.68,-1.41) -- (1.52,-1.22) -- (1.9,-0.43) -- (1.52,0.35) -- cycle;
\fill [gray!20!white] (-0.68,1.41) -- (-1.52,1.22) -- (-1.9,0.43) -- (-1.52,-0.35) -- cycle;
\fill [gray!20!white] (1.52,-0.35) -- (1.9,0.43) -- (1.52,1.22) -- (0.68,1.41) -- cycle;
\fill [gray!20!white] (-1.52,0.35) -- (-1.9,-0.43) -- (-1.52,-1.22) -- (-0.68,-1.41) -- cycle;
\fill [gray!20!white] (1.22,0.97) -- (0.85,1.76) -- (0.0,1.95) -- (-0.68,1.41) -- cycle;
\fill [gray!20!white] (-1.22,-0.97) -- (-0.85,-1.76) -- (0.0,-1.95) -- (0.68,-1.41) -- cycle;
\fill [gray!20!white] (0.0,1.56) -- (-0.85,1.76) -- (-1.52,1.22) -- (-1.52,0.35) -- cycle;
\fill [gray!20!white] (0.0,-1.56) -- (0.85,-1.76) -- (1.52,-1.22) -- (1.52,-0.35) -- cycle;

\fill [red!50!white] (1.9,0.43) -- (1.52,1.22) -- (0.85,1.76) -- (0.0,1.95) -- (0.68,1.41) -- (1.22,0.97) -- cycle;

\draw [thin,gray] (-1.22,0.97) -- (-1.22,0.97);
\draw [thin,gray] (-1.22,0.97) -- (-1.9,0.43);
\draw [thin,gray] (-1.22,0.97) -- (-1.9,-0.43);
\draw [thin,gray] (-1.22,0.97) -- (-1.22,-0.97);
\draw [thin,gray] (-1.22,0.97) -- (-0.85,1.76);
\draw [thin,gray] (-1.9,0.43) -- (-1.22,0.97);
\draw [thin,gray] (-1.9,0.43) -- (-1.9,0.43);
\draw [thin,gray] (-1.9,0.43) -- (-1.9,-0.43);
\draw [thin,gray] (-1.9,0.43) -- (-1.22,-0.97);
\draw [thin,gray] (-1.9,0.43) -- (-1.52,-0.35);
\draw [thin,gray] (-1.9,0.43) -- (-0.68,1.41);
\draw [thin,gray] (-1.9,0.43) -- (-1.52,1.22);
\draw [thin,gray] (-1.9,-0.43) -- (-1.22,0.97);
\draw [thin,gray] (-1.9,-0.43) -- (-1.9,0.43);
\draw [thin,gray] (-1.9,-0.43) -- (-1.9,-0.43);
\draw [thin,gray] (-1.9,-0.43) -- (-1.22,-0.97);
\draw [thin,gray] (-1.9,-0.43) -- (-1.52,-1.22);
\draw [thin,gray] (-1.9,-0.43) -- (-0.68,-1.41);
\draw [thin,gray] (-1.9,-0.43) -- (-1.52,0.35);
\draw [thin,gray] (-1.22,-0.97) -- (-1.22,0.97);
\draw [thin,gray] (-1.22,-0.97) -- (-1.9,0.43);
\draw [thin,gray] (-1.22,-0.97) -- (-1.9,-0.43);
\draw [thin,gray] (-1.22,-0.97) -- (-1.22,-0.97);
\draw [thin,gray] (-1.22,-0.97) -- (-0.85,-1.76);
\draw [thin,gray] (-1.22,-0.97) -- (0.0,-1.95);
\draw [thin,gray] (-1.22,-0.97) -- (0.68,-1.41);
\draw [thin,gray] (-1.52,-0.35) -- (-1.9,0.43);
\draw [thin,gray] (-1.52,-0.35) -- (-1.52,-0.35);
\draw [thin,gray] (-1.52,-0.35) -- (-1.52,-1.22);
\draw [thin,gray] (-1.52,-0.35) -- (-0.85,-1.76);
\draw [thin,gray] (-1.52,-0.35) -- (0.0,-1.56);
\draw [thin,gray] (-1.52,-0.35) -- (-0.68,1.41);
\draw [thin,gray] (-1.52,-0.35) -- (-1.52,1.22);
\draw [thin,gray] (-1.52,-1.22) -- (-1.9,-0.43);
\draw [thin,gray] (-1.52,-1.22) -- (-1.52,-0.35);
\draw [thin,gray] (-1.52,-1.22) -- (-1.52,-1.22);
\draw [thin,gray] (-1.52,-1.22) -- (-0.85,-1.76);
\draw [thin,gray] (-1.52,-1.22) -- (0.0,-1.56);
\draw [thin,gray] (-1.52,-1.22) -- (-0.68,-1.41);
\draw [thin,gray] (-1.52,-1.22) -- (-1.52,0.35);
\draw [thin,gray] (-0.85,-1.76) -- (-1.22,-0.97);
\draw [thin,gray] (-0.85,-1.76) -- (-1.52,-0.35);
\draw [thin,gray] (-0.85,-1.76) -- (-1.52,-1.22);
\draw [thin,gray] (-0.85,-1.76) -- (-0.85,-1.76);
\draw [thin,gray] (-0.85,-1.76) -- (0.0,-1.56);
\draw [thin,gray] (-0.85,-1.76) -- (0.0,-1.95);
\draw [thin,gray] (-0.85,-1.76) -- (0.68,-1.41);
\draw [thin,gray] (0.0,-1.56) -- (-1.52,-0.35);
\draw [thin,gray] (0.0,-1.56) -- (-1.52,-1.22);
\draw [thin,gray] (0.0,-1.56) -- (-0.85,-1.76);
\draw [thin,gray] (0.0,-1.56) -- (0.0,-1.56);
\draw [thin,gray] (0.0,-1.56) -- (0.85,-1.76);
\draw [thin,gray] (0.0,-1.56) -- (1.52,-1.22);
\draw [thin,gray] (0.0,-1.56) -- (1.52,-0.35);
\draw [thin,gray] (1.22,-0.97) -- (1.22,-0.97);
\draw [thin,gray] (1.22,-0.97) -- (-0.68,-1.41);
\draw [thin,gray] (1.22,-0.97) -- (0.0,-1.95);
\draw [thin,gray] (1.22,-0.97) -- (0.85,-1.76);
\draw [thin,gray] (1.22,-0.97) -- (1.9,-0.43);
\draw [thin,gray] (-0.68,-1.41) -- (-1.9,-0.43);
\draw [thin,gray] (-0.68,-1.41) -- (-1.52,-1.22);
\draw [thin,gray] (-0.68,-1.41) -- (1.22,-0.97);
\draw [thin,gray] (-0.68,-1.41) -- (-0.68,-1.41);
\draw [thin,gray] (-0.68,-1.41) -- (0.0,-1.95);
\draw [thin,gray] (-0.68,-1.41) -- (0.85,-1.76);
\draw [thin,gray] (-0.68,-1.41) -- (-1.52,0.35);
\draw [thin,gray] (0.0,-1.95) -- (-1.22,-0.97);
\draw [thin,gray] (0.0,-1.95) -- (-0.85,-1.76);
\draw [thin,gray] (0.0,-1.95) -- (1.22,-0.97);
\draw [thin,gray] (0.0,-1.95) -- (-0.68,-1.41);
\draw [thin,gray] (0.0,-1.95) -- (0.0,-1.95);
\draw [thin,gray] (0.0,-1.95) -- (0.85,-1.76);
\draw [thin,gray] (0.0,-1.95) -- (0.68,-1.41);
\draw [thin,gray] (0.85,-1.76) -- (0.0,-1.56);
\draw [thin,gray] (0.85,-1.76) -- (1.22,-0.97);
\draw [thin,gray] (0.85,-1.76) -- (-0.68,-1.41);
\draw [thin,gray] (0.85,-1.76) -- (0.0,-1.95);
\draw [thin,gray] (0.85,-1.76) -- (0.85,-1.76);
\draw [thin,gray] (0.85,-1.76) -- (1.52,-1.22);
\draw [thin,gray] (0.85,-1.76) -- (1.52,-0.35);
\draw [thin,gray] (1.9,-0.43) -- (1.22,-0.97);
\draw [thin,gray] (1.9,-0.43) -- (1.9,-0.43);
\draw [thin,gray] (1.9,-0.43) -- (1.52,0.35);
\draw [thin,gray] (1.9,-0.43) -- (0.68,-1.41);
\draw [thin,gray] (1.9,-0.43) -- (1.52,-1.22);
\draw [thin,gray] (1.52,0.35) -- (1.9,-0.43);
\draw [thin,gray] (1.52,0.35) -- (1.52,0.35);
\draw [thin,gray] (1.52,0.35) -- (0.68,-1.41);
\draw [thin,gray] (1.52,0.35) -- (1.52,-1.22);
\draw [thin,gray] (1.52,0.35) -- (0.0,1.56);
\draw [thin,gray] (0.68,-1.41) -- (-1.22,-0.97);
\draw [thin,gray] (0.68,-1.41) -- (-0.85,-1.76);
\draw [thin,gray] (0.68,-1.41) -- (0.0,-1.95);
\draw [thin,gray] (0.68,-1.41) -- (1.9,-0.43);
\draw [thin,gray] (0.68,-1.41) -- (1.52,0.35);
\draw [thin,gray] (0.68,-1.41) -- (0.68,-1.41);
\draw [thin,gray] (0.68,-1.41) -- (1.52,-1.22);
\draw [thin,gray] (1.52,-1.22) -- (0.0,-1.56);
\draw [thin,gray] (1.52,-1.22) -- (0.85,-1.76);
\draw [thin,gray] (1.52,-1.22) -- (1.9,-0.43);
\draw [thin,gray] (1.52,-1.22) -- (1.52,0.35);
\draw [thin,gray] (1.52,-1.22) -- (0.68,-1.41);
\draw [thin,gray] (1.52,-1.22) -- (1.52,-1.22);
\draw [thin,gray] (1.52,-1.22) -- (1.52,-0.35);
\draw [thin,gray] (1.9,0.43) -- (1.22,-0.97);
\draw [thin,gray] (1.9,0.43) -- (1.9,-0.43);
\draw [thin,gray] (1.9,0.43) -- (1.52,-0.35);
\draw [thin,gray] (1.52,1.22) -- (1.52,0.35);
\draw [thin,gray] (1.52,1.22) -- (1.52,-0.35);
\draw [thin,gray] (1.52,1.22) -- (0.0,1.56);
\draw [thin,gray] (0.68,1.41) -- (-1.22,0.97);
\draw [thin,gray] (0.68,1.41) -- (1.52,-0.35);
\draw [thin,gray] (0.68,1.41) -- (-0.85,1.76);
\draw [thin,gray] (1.52,-0.35) -- (0.0,-1.56);
\draw [thin,gray] (1.52,-0.35) -- (0.85,-1.76);
\draw [thin,gray] (1.52,-0.35) -- (1.52,-1.22);
\draw [thin,gray] (1.52,-0.35) -- (1.52,-0.35);
\draw [thin,gray] (1.22,0.97) -- (1.22,-0.97);
\draw [thin,gray] (1.22,0.97) -- (1.9,-0.43);
\draw [thin,gray] (1.22,0.97) -- (-0.68,1.41);
\draw [thin,gray] (0.85,1.76) -- (1.52,0.35);
\draw [thin,gray] (0.85,1.76) -- (-0.68,1.41);
\draw [thin,gray] (0.85,1.76) -- (0.0,1.56);
\draw [thin,gray] (0.0,1.95) -- (-1.22,0.97);
\draw [thin,gray] (0.0,1.95) -- (-0.68,1.41);
\draw [thin,gray] (0.0,1.95) -- (-0.85,1.76);
\draw [thin,gray] (-0.68,1.41) -- (-1.9,0.43);
\draw [thin,gray] (-0.68,1.41) -- (-1.52,-0.35);
\draw [thin,gray] (-0.68,1.41) -- (-0.68,1.41);
\draw [thin,gray] (-0.68,1.41) -- (-1.52,1.22);
\draw [thin,gray] (0.0,1.56) -- (1.52,0.35);
\draw [thin,gray] (0.0,1.56) -- (0.0,1.56);
\draw [thin,gray] (0.0,1.56) -- (-0.85,1.76);
\draw [thin,gray] (0.0,1.56) -- (-1.52,1.22);
\draw [thin,gray] (0.0,1.56) -- (-1.52,0.35);
\draw [thin,gray] (-0.85,1.76) -- (-1.22,0.97);
\draw [thin,gray] (-0.85,1.76) -- (0.0,1.56);
\draw [thin,gray] (-0.85,1.76) -- (-0.85,1.76);
\draw [thin,gray] (-0.85,1.76) -- (-1.52,1.22);
\draw [thin,gray] (-0.85,1.76) -- (-1.52,0.35);
\draw [thin,gray] (-1.52,1.22) -- (-1.9,0.43);
\draw [thin,gray] (-1.52,1.22) -- (-1.52,-0.35);
\draw [thin,gray] (-1.52,1.22) -- (-0.68,1.41);
\draw [thin,gray] (-1.52,1.22) -- (0.0,1.56);
\draw [thin,gray] (-1.52,1.22) -- (-0.85,1.76);
\draw [thin,gray] (-1.52,1.22) -- (-1.52,1.22);
\draw [thin,gray] (-1.52,1.22) -- (-1.52,0.35);
\draw [thin,gray] (-1.52,0.35) -- (-1.9,-0.43);
\draw [thin,gray] (-1.52,0.35) -- (-1.52,-1.22);
\draw [thin,gray] (-1.52,0.35) -- (-0.68,-1.41);
\draw [thin,gray] (-1.52,0.35) -- (0.0,1.56);
\draw [thin,gray] (-1.52,0.35) -- (-0.85,1.76);
\draw [thin,gray] (-1.52,0.35) -- (-1.52,1.22);
\draw [thin,gray] (-1.52,0.35) -- (-1.52,0.35);
\draw [thin,red] (1.9,0.43) -- (1.9,0.43);
\draw [thin,red] (1.9,0.43) -- (1.52,1.22);
\draw [thin,red] (1.9,0.43) -- (0.68,1.41);
\draw [thin,red] (1.9,0.43) -- (1.22,0.97);
\draw [thin,red] (1.52,1.22) -- (1.9,0.43);
\draw [thin,red] (1.52,1.22) -- (1.52,1.22);
\draw [thin,red] (1.52,1.22) -- (0.85,1.76);
\draw [thin,red] (1.52,1.22) -- (0.68,1.41);
\draw [thin,red] (0.85,1.76) -- (1.52,1.22);
\draw [thin,red] (0.85,1.76) -- (0.85,1.76);
\draw [thin,red] (0.85,1.76) -- (0.0,1.95);
\draw [thin,red] (0.85,1.76) -- (1.22,0.97);
\draw [thin,red] (0.0,1.95) -- (0.85,1.76);
\draw [thin,red] (0.0,1.95) -- (0.0,1.95);
\draw [thin,red] (0.0,1.95) -- (0.68,1.41);
\draw [thin,red] (0.0,1.95) -- (1.22,0.97);
\draw [thin,red] (0.68,1.41) -- (1.9,0.43);
\draw [thin,red] (0.68,1.41) -- (1.52,1.22);
\draw [thin,red] (0.68,1.41) -- (0.0,1.95);
\draw [thin,red] (0.68,1.41) -- (0.68,1.41);
\draw [thin,red] (1.22,0.97) -- (1.9,0.43);
\draw [thin,red] (1.22,0.97) -- (0.85,1.76);
\draw [thin,red] (1.22,0.97) -- (0.0,1.95);
\draw [thin,red] (1.22,0.97) -- (1.22,0.97);
\node [circle,fill,inner sep=1] (02) at (-1.22,0.97) {};
\node [circle,fill,inner sep=1] (03) at (-1.9,0.43) {};
\node [circle,fill,inner sep=1] (04) at (-1.9,-0.43) {};
\node [circle,fill,inner sep=1] (05) at (-1.22,-0.97) {};
\node [circle,fill,inner sep=1] (13) at (-1.52,-0.35) {};
\node [circle,fill,inner sep=1] (14) at (-1.52,-1.22) {};
\node [circle,fill,inner sep=1] (15) at (-0.85,-1.76) {};
\node [circle,fill,inner sep=1] (16) at (0.0,-1.56) {};
\node [circle,fill,inner sep=1] (20) at (1.22,-0.97) {};
\node [circle,fill,inner sep=1] (24) at (-0.68,-1.41) {};
\node [circle,fill,inner sep=1] (25) at (0.0,-1.95) {};
\node [circle,fill,inner sep=1] (26) at (0.85,-1.76) {};
\node [circle,fill,inner sep=1] (30) at (1.9,-0.43) {};
\node [circle,fill,inner sep=1] (31) at (1.52,0.35) {};
\node [circle,fill,inner sep=1] (35) at (0.68,-1.41) {};
\node [circle,fill,inner sep=1] (36) at (1.52,-1.22) {};
\node [circle,fill,inner sep=1] (46) at (1.52,-0.35) {};
\node [circle,fill,inner sep=1] (53) at (-0.68,1.41) {};
\node [circle,fill,inner sep=1] (61) at (0.0,1.56) {};
\node [circle,fill,inner sep=1] (62) at (-0.85,1.76) {};
\node [circle,fill,inner sep=1] (63) at (-1.52,1.22) {};
\node [circle,fill,inner sep=1] (64) at (-1.52,0.35) {};
\node [circle,fill,red,inner sep=1,label={[red,anchor=west,rotate=13]13:$(4,0)$}] (40) at (1.9,0.43) {};
\draw [very thick,->,red] (0,0) -- (40);
\node [circle,fill,red,inner sep=1,label={[red,anchor=west,rotate=39]39:$(4,1)$}] (41) at (1.52,1.22) {};
\draw [very thick,->,red] (0,0) -- (41);
\node [circle,fill,red,inner sep=1,label={[red,anchor=west,rotate=64]64:$(5,1)$}] (51) at (0.85,1.76) {};
\draw [very thick,->,red] (0,0) -- (51);
\node [circle,fill,red,inner sep=1,label={[red,anchor=west,rotate=90]90:$(5,2)$}] (52) at (0.0,1.95) {};
\draw [very thick,->,red] (0,0) -- (52);
\node [circle,fill,red,inner sep=1,label={[red,anchor=north east,rotate=64]-116:$(4,2)$}] (42) at (0.68,1.41) {};
\draw [very thick,->,red] (0,0) -- (42);
\node [circle,fill,red,inner sep=1,label={[red,anchor=north east,rotate=38]-142:$(5,0)$}] (50) at (1.22,0.97) {};
\draw [very thick,->,red] (0,0) -- (50);

   \end{tikzpicture}
  \caption{$\rel K_{7/2}$ with a complete bipartite subgraph (on the left) and the corresponding face of $\gBip{\rel K_{7/2}}$ after mapping to $\mathbb R^2$ (on the right).}
    \label{fig:faces-of-circular-cliques}
\end{figure}

\begin{lemma} \label{lem:graphs-and-s1}
  For any $2<p/q<4$ and any square-free non-bipartite graph $\rel G$, there exist \equivariant-maps
  \begin{enumerate}
    \item \( s_1\colon \gBip{\rel K_{p/q}} \to \Sphere^1 \), and \item \( s_2\colon \gBip{\rel G} \to \Sphere^1 \). \end{enumerate}
\end{lemma}

\begin{proof}
  (1)
  We first define a \equivariant-map $g\colon \geom{\Bip{\rel K_{p/q}}} \to \mathbb R^2$. We will show that the origin $0$ is not in the image of $g$, which then implies that the map $x \mapsto g(x)/\size {g(x)}$ is a \equivariant-map to $\Sphere^1$.

  First, we define $g$ on vertices of the complex, i.e., oriented edges of $\rel K_{p/q}$: Place $p$ points $x_0,\dots,x_{p-1}$ on $\Sphere^1$ in a regular $p$-gon pattern. Map the edge $(a,b)$ to the point $x_b - x_a$. Then extend $g$ linearly to the interior points of faces. Clearly, $g$ is a \equivariant-map. See Fig.~\ref{fig:faces-of-circular-cliques} for a visualisation of $g$ for $\rel K_{7/2}$.

  Let $\sigma$ be a face. That is, $\sigma \subseteq A \times B\subseteq E(K_{p/q})$ for some non-empty sets of vertices $A,B$.
  For an edge $(a,b)$, the distance between $x_a,x_b$ on the circle (the length of the shorter arc between them) is at least $2\pi\cdot q/p$.
  Since $p/q < 4$, this is greater than $\pi/2$.
  Hence there are no $a,a' \in A, b,b' \in B$ such that $x_a,x_b,x_{a'},x_{b'}$ occur in this order on the circle, as the distances would add up to more than $2\pi$.
  Therefore, there is a line in $\RR^2$ that strictly separates $\{x_a \mid a \in A\}$ from $\{x_b \mid b \in B\}$ (indeed, scanning the circle clockwise, there is exactly one interval from $A$ to $B$ and exactly one from $B$ to $A$, both of length at least $\pi/2$; any line crossing these intervals will do).
  This implies that the convex hull of vectors $x_a - x_b$ cannot contain 0 (since each such vector has a positive dot product with the normal vector of the line).

  (2)
  This statement follows from observing that $\gBip{\rel G}$ is essentially 1-dimensional which loosely follows from the facts that there are no copies of the complete bipartite graph $\rel K_{2,2}$ and that every free \equivariant-space of dimension 1 maps to $\Sphere^1$ (see \cite[Proposition 5.3.2(v)]{Mat03}). We present a compressed version of the two arguments.

  Let $E^+ \cup E^- = V(\Bip{\rel G})$ be an arbitrary partition into two sets that are swapped by reversing the edges, i.e., $-E^+ = E^-$ and $-E^- = E^+$. This means that we choose an orientation for each edge of $\rel G$, and denote by $E^+$ the set of all edges of $\rel G$ oriented the chosen way, while $E^-$ is the set of all edges oriented the opposite way. We define a mapping $h\colon \gBip{\rel G} \to \Sphere^1$ on the vertices of $\gBip{\rel G}$ by setting $h(e) = (1,0)$ if $e\in E^+$ and $h(e) = (-1,0)$ if $e\in E^-$.

  We extend this mapping to inner points of faces. First, observe that for every face $\sigma \in \Bip{\rel G}$ with at least two elements, there is a vertex $u\in V(G)$ such that either $\sigma \subseteq \{ (u,v) \in E(G) \}$ or $\sigma \subseteq \{ (v,u) \in E(G) \}$, as otherwise we can find a copy of $\rel C_4$ in $\rel G$.
  We map $\geom\sigma$ for faces of the first form to the arc connecting $(1,0)$ and $(-1,0)$ with positive $y$ coordinates, and $\geom\sigma$ for faces of the second form to the arc with negative $y$ coordinates. More precisely, if $\sigma = \{ (u,v_1),\dots,(u,v_n) \}$, we let
  \(
    a = \sum_{i=1}^n \lambda_i \ (u,v_i)
  \)
  be a point of $\geom\sigma$. Put
  \[
    x = \sum_{(u,v_i)\in E^+} \lambda_i - \sum_{(u,v_i) \in E^-} \lambda_i
  \]
  and $y = \sqrt{1-x^2}$ (note that $\size x \leq 1$, so $y$ is well-defined), and define $h(a) = (x,y)$.
  Now to preserve the \equivariant-action, we map the geometric representations of the faces of the second form to the arc with negative $y$ coordinates analogously putting $y = -\sqrt{1-x^2}$. Clearly, the mapping $h$ defined this way is continuous and it is easy to check that indeed $h(-a) = -h(a)$ for each $a\in \gBip{\rel G}$.
\end{proof}
 
\section{Adjunction} \label{sec:adjunction}

In this section we will use both graphs and digraphs, which by default are allowed to have loops.
We will work with certain (di)graph constructions that can be seen as functions from the set of all finite (di)graphs to itself. On one occasion in this section (Subsection~\ref{subsubsec:all-adj}), we will allow the image of a finite digraph to be an infinite digraph; this will be specified.  We denote the set of all finite graphs and digraphs by $\graf$ and $\dgraf$, respectively. The class of all (finite and infinite) digraphs is denoted by $\dgra$.

In this section, we explain what \emph{adjunction} is and how it can be used to obtain reductions between PCSPs. The notion of adjointness we present is a special case of the more general notion of adjoint functors in category theory. We restrict our attention to an order-theoretic version thereof (i.e., to posetal or thin categories), which only considers the existence of homomorphisms; this is also known as a (monotone) Galois connection.
Generally, a monotone Galois connection between two preordered sets $P_1$ and $P_2$ is pair of maps
$\lambda\colon P_1\rightarrow P_2$ and $\gamma\colon P_2\rightarrow P_1$ such that, for all $a\in P_1$ and $b\in P_2$,
\begin{equation} \label{eq:Galois}
    \lambda (a)\le b
      \quad\text{if and only if}\quad
    a\le \gamma(b).
  \end{equation}
For us, the preorder $\le$ will always be the homomorphism preorder $\rightarrow$, and the sets $P_1$ and $P_2$ will be either $\dgraf$ or $\graf$.
In this case, $\Lambda$ and $\Gamma$ are adjoint if, for all (di)graphs $\rel H$ and $\rel G$, we have
\begin{equation} \label{eq:adjunction}
    \Lambda \rel H \to \rel G
          \quad\text{if and only if}\quad
   \rel H \to \Gamma \rel G.
  \end{equation}
In this case $\Lambda$ is a~\emph{left adjoint} and $\Gamma$ is a~\emph{right adjoint}. If, for some $\Lambda$, there exists such $\Gamma$ we also say that $\Lambda$ \emph{has (or admits) a right adjoint}. Similarly, we say that $\Gamma$ \emph{has a~left adjoint} if there exists such $\Lambda$.

Adjunction is an abstraction of a few concepts that are already present in the theory of (P)CSPs: notably, the $\Inv$-$\Pol$ Galois correspondences of Geiger, Bodnarchuk, Kaluzhnin, Kotov, and Romov \cite{Gei68,BKKR69,BKKR69a} and Pippenger \cite{Pip02} between sets of functions and sets of relations can also been seen as adjunctions where, in~(\ref{eq:Galois}), the preorder on one side is the inclusion and the preorder on the other side is the inverse inclusion.
We remark that many constructions described in~\cite[Sections 3 and 4]{BBKO19} (see e.g.\ Lemma~4.4 there)
form pairs of adjoint functions. We also remark that condition~(\ref{eq:adjunction}) makes perfect sense when $\Lambda$ and $\Gamma$ are maps between the sets of relational structures of different signatures (say, between the set of all finite digraphs and the set of all finite 3-uniform hypergraphs), and all results in Subsection~\ref{subsec:gen-adj} hold in this more general setting.

This section is organised as follows. In Subsection~\ref{subsec:adj-csp}, we show that the standard gadget reductions from the algebraic approach to the CSP can be seen as a special case of adjunction. In Subsection~\ref{subsec:gen-adj}, we give general results about adjunctions and reductions between PCSPs. In Subsection~\ref{sec:righthard}, we apply specific cases of adjunction to prove our results about the hardness of approximate graph colouring and demonstrate that the reductions between PCSPs obtained there cannot be captured by the algebraic theory from~\cite{BKO19,BBKO19}. Finally, in Subsection~\ref{subsec:secondMainProof}, we use another specific adjunction to prove that, in a precise technical sense, the complexity of promise graph homomorphism problem depends only on the topological properties of graphs.

To emphasise that many of our proofs in this section do not assume computability of reductions, we will use the following definition.

\begin{definition}
Let $\Lambda$ be a function from $\dgraf$ to $\dgraf$ or from $\graf$ to $\graf$.
We say that $\Lambda$ is
\begin{itemize}
  \item \emph{a reduction} from $\PCSP(\rel H_1,\rel G_1)$ to $\PCSP(\rel H_2,\rel G_2)$ if it preserves the \yes- and \no-answers of the two problems, i.e., for any $\rel I$, $\rel I \to \rel H_1$ implies $\Lambda \rel I \to \rel H_2$ and $\rel I \not\to \rel G_1$ implies $\Lambda \rel I \not\to \rel G_2$. Preserving \yes-answers is also called \emph{completeness} and preserving \no-answers \emph{soundness};
  \item \emph{log-space/polynomial-time computable} if there is a log-space/polynomial-time algorithm that on input $\rel I$ outputs $\Lambda \rel I$;
\end{itemize}
\end{definition}

\subsection{Adjunction in CSPs}
  \label{subsec:adj-csp}

  The algebraic approach to the CSP studies certain constructions on templates of CSPs called \emph{pp-powers}, and it asserts that if $\rel A$ is a pp-power of $\rel B$ then there is a log-space reduction from $\CSP(\rel A)$ to $\CSP(\rel B)$. This reduction is a function from instances of $\CSP(\rel A)$ to $\CSP(\rel B)$ computable in log-space that we call a \emph{gadget replacement}.
  In fact, any such reduction is a left adjoint to the corresponding pp-power construction.
  We present the notions simplified for digraphs and refer to \cite[Section 3.1]{BKW17} for more background. We will use the constructions from Example~\ref{ex:walk-power} as a running example in this subsection.

  Both functions are parameterised by a \emph{gadget} or a~\emph{primitive positive formula} (a \emph{pp-formula}), thus giving a reduction for each gadget. A~\emph{digraph gadget formula} of arity $n$ is a logical formula $\phi(x_1,\dots,x_n,y_1,\dots,y_n)$ of the form
  \[
    \exists {z_1,\dots,z_m}\;.\;(u_1,v_1) \in E\ \wedge\ \dots\ \wedge\ (u_k,v_k) \in E\ \wedge\ (u'_1=v'_1)\ \wedge\ \dots\ \wedge\ (u'_\ell=v'_\ell).
  \]
  where $u_i,v_i,u'_i,v'_i \in \{x_1,\dots,x_n,y_1,\dots,y_n,z_1,\dots,z_m\}$. Such a formula can also be represented by a \emph{gadget digraph} $\rel J_\phi$, which is a digraph with distinguished vertices labelled $x_1,\dots,x_n,y_1,\dots,y_n$, obtained from vertices $\{x_1,\dots,x_n,y_1,\dots,y_n,z_1,\dots,z_m\}$ and edges $\{ (u_i,v_i) \mid i\in [k] \}$ by identifying some of the vertices (according to the equalities in $\phi$).

  If we want a gadget to transform an undirected graph to an undirected graph, we require that the gadget is \emph{symmetric}, i.e., that the formula $\phi(x_1,\dots,x_n,y_1,\dots,y_n)$ is logically equivalent to $\phi(y_1,\dots,y_n,x_1,\dots,x_n)$ for all graphs, or that the gadget graph has an automorphism switching $x_i$ and $y_i$ for each $i$.

  The subdivision from Example~\ref{ex:walk-power} is defined by the following digraph gadget formula of arity $1$:
  \[
    \phi_k(x, y) = \exists z_1, \dots, z_{k-1} \;.\;
      (x, z_1) \in E \wedge (z_1, z_2) \in E \wedge \dots \wedge (z_{k-1}, y) \in E
  \]
  which is symmetric if we only consider undirected graphs. The corresponding gadget would be an (unoriented) path of length $k$ connecting the two distinguished vertices $x$ and~$y$.

  \begin{definition}
  The \emph{gadget replacement} $\Lambda_\phi$ assigned to a~digraph gadget $\phi$ is then defined by applying the following construction. Starting with a~digraph $\rel H$,
  \begin{enumerate}
    \item for each vertex $v \in V(H)$, introduce new vertices $v_1,\dots,v_n \in V(\Lambda_\phi H)$,
    \item for each edge $(u,v) \in E(H)$, introduce a fresh copy of the gadget digraph $\rel J_\phi$ while identifying $x_1,\dots,x_n$ with $u_1,\dots,u_n$ and $y_1,\dots,y_n$ with $v_1,\dots, v_n$; we denote the remaining vertices of this copy of $\rel J_\phi$ by $z_{u,v}$ for $z\in \{z_1,\dots,z_m\}$.
  \end{enumerate}
  Note that some of the vertices $u_i,v_j$ above might get identified, which can also result in long chains of identifications.
  Nevertheless, $\Lambda_\phi$ is log-space computable.
  \end{definition}

  \begin{figure}
    \[\begin{matrix}
      \begin{tikzpicture}[scale=.9, baseline={([yshift=-.5ex]current bounding box.center)}]
        \foreach \i/\c in {0/0,120/1,240/2} {
          \node [circle,fill,inner sep=1.5,label={\i:$\c$}] at (\i:1) {};
          \draw [thick] (\i:1) -- (120+\i:1);
        }
      \end{tikzpicture} & \!\!\stackrel{\Lambda_{\phi_3}}\mapsto\!\! &
      \begin{tikzpicture}[scale=.9, baseline={([yshift=-.5ex]current bounding box.center)}]
        \foreach \i/\c in {0/0,40/\relax,80/\relax,120/1,160/\relax,200/\relax,240/2,280/\relax,320/\relax} {
          \node [circle,fill,inner sep=1.5,label={\i:$\c$}] at (\i:1) {};
          \draw [thick] (\i:1) -- (40+\i:1);
        }
      \end{tikzpicture} & &
      \begin{tikzpicture}[scale=.9, baseline={([yshift=-.5ex]current bounding box.center)}]
        \foreach \i/\c in {0/0,40/1,80/2,120/3,160/4,200/5,240/6,280/7,320/8} {
          \node [circle,fill,inner sep=1.5,label={\i:$\c$}] at (\i:1) {};
          \draw [thick] (\i:1) -- (40+\i:1);
          \draw [thick] (\i:1) -- (120+\i:1);
        }
      \end{tikzpicture} & \!\!\stackrel{\Gamma_{\phi_3}}\mapsfrom\!\! &
      \begin{tikzpicture}[scale=.9, baseline={([yshift=-.5ex]current bounding box.center)}]
        \foreach \i/\c in {0/0,40/1,80/2,120/3,160/4,200/5,240/6,280/7,320/8} {
          \node [circle,fill,inner sep=1.5,label={\i:$\c$}] at (\i:1) {};
          \draw [thick] (\i:1) -- (40+\i:1);
        }
      \end{tikzpicture}
    \end{matrix}\]
    \caption{Example of gadget replacement and a pp-power.}
    \label{ex:4.1}
  \end{figure}

  The subdivision $\Lambda_k$ of Example~\ref{ex:walk-power} is the same as the gadget replacement $\Lambda_{\phi_k}$. We show an example of application in Figure~\ref{ex:4.1}.

  \begin{definition}
  The \emph{pp-power} $\Gamma_\phi \rel G$ of a~digraph $\rel G$ defined by $\phi$ is obtained by the following construction.
  \begin{enumerate}
    \item $V(\Gamma_\phi G) = V(G)^n$, and
    \item $((u_1,\dots,u_n),(v_1,\dots,v_n)) \in E(\Gamma_\phi G)$ if $\phi(u_1,\dots,u_n,v_1,\dots,v_n)$ is true in $\rel G$, in other words, there exists a homomorphism $e_{u,v}$ from the gadget digraph $\rel J_\phi$ to $\rel G$ such that $e_{u,v}(x_i) = u_i$ and $e_{u,v}(y_i) = v_i$ for all $i$.
  \end{enumerate}
  \end{definition}

  Again, it is not hard to check that $\Gamma_k\colon\dgraf\to\dgraf$ from Example~\ref{ex:walk-power} is the same as the pp-power $\Gamma_{\phi_k}$, i.e., it is the graph on the same vertex set where two vertices are connected by an edge iff they are connected by a path of length $k$ in the original graph. Again, see Figure~\ref{ex:4.1} for an application.

  \medskip

  The standard reductions used in the algebraic approach are of the form $\Lambda_\phi$: it is well-known (see \cite[Theorem 13]{BKW17}) that $\Lambda_\phi$ is a~reduction from $\CSP(\Gamma_\phi \rel G)$ to $\CSP(\rel G)$ for any digraph $\rel G$.
  For example, it is not hard to see that $\Lambda_{\phi_3}$ from our running example is really a reduction from $\CSP(\rel K_3)$ to $\CSP(\rel C_9)$ since $\rel K_3$ is homomorphically equivalent to $\Gamma_{\phi_3}\rel C_9$, and hence it has the same CSP.
  This observation also follows immediately from the fact that $\Lambda_\phi$ and $\Gamma_\phi$ are adjoint, which we show here directly.

  \begin{observation}
    Let $\phi$ be a pp-formula. Then $\Lambda_\phi$ and $\Gamma_\phi$ are adjoint.
  \end{observation}

  \begin{proof}
  To prove that indeed $\Lambda_\phi$ and $\Gamma_\phi$ are adjoint, first assume $h\colon \rel H \to \Gamma_\phi \rel G$ is a homomorphism.
  Such a~homomorphism is a~map $h\colon V(H) \to V(G)^n$ which can be seen as an $n$-tuple of mappings $h_1,\dots,h_n\colon V(H) \to V(G)$. Further, since $h$ preserves edges, we have that for each $(u,v) \in E(H)$, $\phi(h_1(u),\dots,h_n(u),h_1(v),\dots,h_n(v))$ is true in $\rel G$, which gives a homomorphism $e_{u,v}\colon \rel J_\phi \to \rel G$ such that $e_{u,v} (x_i) = h_i(u)$ and $e_{u,v} (y_i) = h_i(v)$. We use these $e_{u,v}$'s to define a homomorphism $h'\colon \Lambda_\phi \rel H \to \rel G$:
  \begin{enumerate}
    \item put $h'(u_i) = h_i(u)$ for each $u\in \rel H$ and $i$;
    \item extend $h'$ to new vertices introduced by the second step of gadget replacement of the edge $(u,v) \in E(H)$ by putting $h'(z) = e_{u,v}(z_{u,v})$ for all $z\in \{ z_1,\dots,z_m \}$.
  \end{enumerate}
  Clearly, $h'$ is a homomorphism since each $e_{u,v}$ is and there are no edges in $\Lambda_\phi \rel H$ that are not included in some copy of $\rel J_\phi$.
  For the other implication, assume $g\colon \Lambda_\phi \rel H \to \rel G$. We define $g'\colon \rel H \to \Gamma_\phi \rel G$ as
  \(
    g'(u) = (g(u_1),\dots,g(u_n))
  \)
  for each $u \in V(H)$. It is straightforward to check that $g'$ is indeed a homomorphism. This concludes the proof.
  \end{proof}

  One of the main strengths of the algebraic approach lies in a description of when such reductions apply, by means of polymorphisms and minion homomorphisms; see \cite[Theorem 38]{BKW17}
  (originally appeared in \cite{BOP18}) and~\cite[Theorem 4.12]{BBKO19} for the promise setting. We return to this later in this section (Example~\ref{ex:gadgetVsHom}).

\subsection{General results about adjunction for PCSPs}
\label{subsec:gen-adj}

In the following lemma, we give a few basic and useful properties of adjoint functors that are well-known in category theory. We provide proofs of these facts for completeness.
We say that a function $\Lambda\colon \dgraf \to \dgraf$ is \emph{monotone} if $\Lambda \rel H \to \Lambda \rel G$ for all $\rel H,\rel G$ such that $\rel H \to \rel G$; and it \emph{preserves disjoint unions} if $\Lambda(\rel H_1 + \rel H_2)$ and $\Lambda \rel H_1 + \Lambda \rel H_2$ are homomorphically equivalent for all digraphs $\rel H_1,\rel H_2$ (we denote disjoint union with $+$).

\begin{lemma} \label{lem:adj-technical}
  Let $\Lambda,\Gamma\colon \dgraf \to \dgraf$ be adjoint.
  Then
  \begin{enumerate}
    \item $\rel G \to \Gamma\Lambda \rel G$ for all digraphs $\rel G$;
    \item $\Lambda\Gamma \rel H \to \rel H$ for all digraphs $\rel H$;
    \item both $\Lambda$ and $\Gamma$ are monotone; and
    \item $\Lambda$ preserves disjoint unions.
  \end{enumerate}
\end{lemma}

\begin{proof}
  We start by proving (1): $\Lambda \rel G\to \Lambda \rel G$ implies that $\rel G \to \Gamma\Lambda \rel G$ by adjunction.
  Similarly for (2), observe that $\Gamma \rel H\to \Gamma \rel H$ implies $\Lambda\Gamma \rel H \to \rel H$ by adjunction.

  For (3) assume $\rel H \to \rel G$. Then by (1), we have $\rel H \to \rel G \to \Gamma\Lambda \rel G$, and therefore by adjunction $\Lambda \rel H \to \Lambda \rel G$. This concludes that $\Lambda$ is monotone.
  Similarly from (2), we have $\Lambda\Gamma \rel H \to \rel H \to \rel G$ and hence by adjunction $\Gamma \rel H \to \Gamma \rel G$, so $\Gamma$ is monotone.

  For (4), consider the disjoint union of digraphs $\rel H_1,\rel H_2$.
  Note that $\rel H_i \to \rel H_1 + \rel H_2$ for $i=1,2$ implies $\Lambda \rel H_1 + \Lambda \rel H_2 \to \Lambda(\rel H_1 + \rel H_2)$ by monotonicity.
  To show the other direction, observe that $\rel H_i \to \Gamma \Lambda \rel H_i \to \Gamma(\Lambda \rel H_1 + \Lambda \rel H_2)$ for $i=1,2$ by (1) and monotonicity of $\Gamma$, hence $\rel H_1 + \rel H_2 \to \Gamma(\Lambda \rel H_1 + \Lambda \rel H_2)$ and therefore $\Lambda(\rel H_1 + \rel H_2) \to \Lambda \rel H_1 + \Lambda \rel H_2$ by adjunction.
\end{proof}

The next result is the main theorem of this subsection. It describes when an adjunction provides a reduction between two PCSPs. This theorem will be applied in the following two subsections to provide new reductions between promise digraph homomorphism problems of the sort that has not been described before.

\begin{theorem}\label{thm:adj1}
  Let $\Lambda,\Gamma\colon \dgraf \to \dgraf$ be adjoint.
  Let $\rel H_1,\rel G_1,\rel H_2,\rel G_2$ be digraphs such that $\rel H_i \to \rel G_i$ for $i=1,2$.
  Then $\Lambda$ is a reduction from $\PCSP(\rel H_1, \rel G_1)$ to $\PCSP(\rel
  H_2, \rel G_2)$ if and only if $\ \rel H_1 \to \Gamma \rel H_2$ and $\Gamma \rel G_2 \to \rel G_1$.
\end{theorem}
\begin{proof}
  Assume first that $\rel H_1 \to \Gamma \rel H_2$ and $\Gamma \rel G_2 \to \rel G_1$. Then $\Lambda$ preserves \yes-instances because $\rel I \to \rel H_1$ implies $\rel I \to \Gamma \rel H_2$ (since $\rel H_1 \to \Gamma \rel H_2$ by assumption) and then $\Lambda \rel I \to \rel H_2$ by adjointness.  It also preserves \no-instances because $\Lambda \rel I \to \rel G_2$ implies $\rel I \to \Gamma \rel G_2$ by adjointness and then $\rel I \to \rel G_1$ because $\Gamma \rel G_2 \to \rel G_1$ (by assumption). Hence $\Lambda$ is a reduction, as claimed.

  For the converse, preserving \yes-instances means that for $\rel I \in \dgraf$,
  $\rel I \to \rel H_1$ implies $\Lambda \rel I  \to \rel H_2$. Using this with $\rel I=\rel H_1$, we get that
  $\Lambda \rel H_1  \to \rel H_2$ and thus $\rel H_1 \to \Gamma \rel H_2$ by adjointness.
  Preserving \no-instances means that $\Lambda \rel I \to \rel G_2$ implies $\rel I \to \rel G_1$. Take $\rel I= \Gamma \rel G_2$. Since $\Lambda\Gamma \rel G_2 \to \rel G_2$ by Lemma~\ref{lem:adj-technical}(2), we have $\Gamma \rel G_2 \to \rel G_1$.
\end{proof}

Naturally, we use the above theorem in the case that $\Lambda$ is log-space computable to obtain a log-space reduction between the two PCSPs. In the same way, it can also be applied if $\Lambda$ is polynomial-time computable, if the goal is to get a~polynomial-time reduction, etc. Note that, in such applications, $\Gamma$ need not be computable to guarantee the correctness of the reduction.

\begin{remark}
  We note (a well-known categorical fact) that any two right adjoints $\Gamma_1$ and $\Gamma_2$ of $\Lambda$ are homomorphically equivalent in the following sense: for all $\rel G$, $\Gamma_1 \rel G$ and $\Gamma_2 \rel G$ are homomorphically equivalent. This follows, for example, from the above theorem: we have that $\Lambda$ is a reduction from $\PCSP(\Gamma_1 \rel G,\Gamma_1 \rel G)$ to $\PCSP(\rel G,\rel G)$ since $\Gamma_1$ is a right adjoint to $\Lambda$, and consequently, $\Gamma_1 \rel G \leftrightarrow \Gamma_2 \rel G$ since $\Gamma_2$ is a right adjoint.
\end{remark}

\begin{corollary}\label{cor:adj1}
  Let $\Lambda,\Gamma\colon \dgraf \to \dgraf$ be adjoint. Then
  \begin{enumerate}
  \item  $\Lambda$ is a reduction from $\PCSP(\rel H, \Gamma \rel G)$ to $\PCSP(\Lambda \rel H, \rel G)$, for all graphs $\rel H, \rel G$ such that $\rel H\to \Gamma \rel G$ (or equivalently, $\Lambda \rel H \to \rel G$);
  \item $\Lambda$ is a reduction from $\PCSP(\Gamma \rel H, \Gamma \rel G)$ to $\PCSP(\rel H, \rel G)$, for all graphs $\rel H, \rel G$ such that $\rel H\to \rel G$.
  \end{enumerate}
\end{corollary}
\begin{proof}
  For (1), the first condition of Theorem~\ref{thm:adj1} is equivalent to $\rel H \to \Gamma \Lambda \rel H$, which holds by adjunction (see Lemma~\ref{lem:adj-technical}(1)); the second condition is trivial: $\Gamma \rel G \to \Gamma \rel G$.
  For (2), both conditions are trivial.
\end{proof}

We remark that all reductions described in Theorem~\ref{thm:adj1} can be deduced from the special case in Corollary~\ref{cor:adj1}(1) by composing it with trivial reductions (that map every instance to itself).
Recall that there is a~trivial reduction from $\PCSP(\rel H',\rel G')$ to $\PCSP(\rel H,\rel G)$ if (and only if) $\rel H' \to \rel H$ and $\rel G \to \rel G'$; this is referred to as a~\emph{homomorphic relaxation} \cite[Definition 4.6]{BBKO19}.
If digraphs $\rel H_1, \rel G_1, \rel H_2, \rel G_2$ satisfy the conditions of Theorem~\ref{thm:adj1},
then we have the following sequence of reductions:
\[
  \PCSP(\rel H_1, \rel G_1) \xrightarrow{\text{triv.}}
  \PCSP(\rel H_1, \Gamma \rel G_2) \xrightarrow{\text{Cor.~\ref{cor:adj1}(1)}}
  \PCSP(\Lambda \rel H_1, \rel G_2) \xrightarrow{\text{triv.}}
  \PCSP(\rel H_2, \rel G_2).
\]
Similarly, Corollary~\ref{cor:adj1}(2) implies all reductions in Theorem~\ref{thm:adj1}:
\[
  \PCSP(\rel H_1, \rel G_1) \xrightarrow{\text{triv.}}
  \PCSP(\Gamma \rel H_2, \Gamma \rel G_2) \xrightarrow{\text{Cor.~\ref{cor:adj1}(2)}}
  \PCSP(\rel H_2, \rel G_2).
\]

\begin{example}\label{ex:gadgetVsHom}
  What we described in Subsection~\ref{subsec:adj-csp} in the context of CSPs, can be generalised to PCSPs as follows.
  The following are equivalent:
  \begin{enumerate}
   	\item there is a pp-formula $\phi$ such that $\Lambda_\phi$ is a log-space reduction from $\PCSP(\rel H_1, \rel G_1)$ to $\PCSP(\rel H_2, \rel G_2)$ (i.e., there exists some gadget reduction between the two);
  	\item there is a pp-formula $\phi$ such that $\rel H_1 \to \Gamma_\phi \rel H_2$ and $\Gamma_\phi \rel G_2 \to \rel G_1$ (i.e., $(\rel H_1, \rel G_1)$ is a homomorphic relaxation of a pp-power of $(\rel H_2, \rel G_2)$);   	
  	\item there is a minion homomorphism $\Pol(\rel H_2, \rel G_2) \to \Pol(\rel H_1, \rel G_1)$.
  \end{enumerate}

  The equivalence of (2) and (3) is by \cite[Theorem~4.12]{BBKO19}. The equivalence of the last two items and (1) is implicit in \cite{BBKO19} (see e.g.\ Lemma 4.11 there), but the equivalence of (1) and (2) follows directly from Theorem~\ref{thm:adj1} above.

  For example, all \NP-hard (non-promise) CSPs are reducible to one another in this way.  The understanding that one can get simple reductions between CSPs by relating their sets of polymorphisms goes back at least as far as~\cite{Jeavons97:jacm}.  The use of pp-formulas and minion homomorphisms was initiated in \cite{Bulatov05:classifying} and~\cite{BOP18}, respectively.
\end{example}

\subsubsection{Are all reductions given by adjunctions?}\label{subsubsec:all-adj}

Theorem~\ref{thm:adj1} raises a question whether all reductions between PCSPs are given by adjunctions, in the sense that every reduction is a left adjoint from some adjoint pair. By Lemma~\ref{lem:adj-technical}(3--4), we have to restrict this question to reductions that are monotone and preserve disjoint unions.
We will show that the answer to this question is positive, with a small technical caveat that the right adjoint might produce infinite digraphs on a finite input. Note that this caveat is not an issue, since the right adjoint does not need to be computable.

This suggests that looking at classes of adjoints that generalise the simple gadget constructions $\Lambda_\phi$ could lead to understanding an essential part of all reductions between PCSPs. In particular, we hope that the use of the PCP theorem in proving \NP-hardness of PCSPs (see~\cite[Section 5]{BBKO19}) can be superseded this way.
We remark that, e.g., the reduction in Dinur's proof of the PCP theorem \cite{Din07} is not necessarily monotone: this is due to the fact that the number of repetitions of the gap amplification depends on the size of the input.

The core of the argument in the proof of the following theorem is a well-known categorical statement (the adjoint functor theorem). Again, we provide a full proof for completeness.

\begin{theorem}\label{thm:all-red}
  Let $\rel H_1,\rel G_1,\rel H_2,\rel G_2$ be finite digraphs such that $\rel H_i \to \rel G_i$ for $i = 1, 2$, and let $\Lambda\colon \dgraf \to \dgraf$ be a reduction from $\PCSP(\rel H_1,\rel G_1)$ to $\PCSP(\rel H_2,\rel G_2)$. Assume additionally that $\Lambda$ is monotone and preserves disjoint unions. Then there is a~function $\Gamma\colon \dgraf \to \dgra$ with possibly infinite images such that, for all finite digraphs $\rel H$ and $\rel G$, we have \( \rel H \to \Gamma \rel G \) if and only if \( \Lambda \rel H \to \rel G \).
  Moreover, we have $\rel H_1 \to \Gamma \rel H_2$ and $\Gamma \rel G_2 \to \rel G_1$ for any such $\Gamma$.
\end{theorem}

\begin{proof}
  We define $\Gamma \rel G$ to be the disjoint union of all finite digraphs $\rel I$ such that $\Lambda \rel I \to \rel G$.  Assuming that $\rel H$ and $\rel G$ are finite digraphs, we immediately get that $\Lambda \rel H \to \rel G$ implies $\rel H \to \Gamma \rel G$. We first prove the other implication for connected $\rel H$: assuming that $\rel H \to \Gamma \rel G$, we get that $\rel H$ maps to some connected component of $\Gamma \rel G$ and thus $\rel H \to \rel I$ for some finite $\rel I$ such that $\Lambda \rel I \to \rel G$. This gives that $\Lambda \rel H \to \Lambda \rel I \to \rel G$ since $\Lambda$ is monotone. For disconnected $\rel H$, we use that $\Lambda$ preserves disjoint unions, so we may repeat the above argument for each component separately.

  The ``moreover'' claim is proved similarly to Theorem~\ref{thm:adj1}. In particular, the proof that if $\Lambda$ is a reduction then $\rel H_1 \to \Gamma \rel H_2$ is identical to the one in Theorem~\ref{thm:adj1}. To prove that $\Gamma \rel G_2 \to \rel G_1$, we cannot simply use preservation of \no-instances on the possibly infinite $\Gamma \rel G_2$. Instead, we get that for every $\rel I$ finite, $\rel I \to \Gamma \rel G_2$ implies $\Lambda \rel I \to \rel G_2$, which implies $\rel I \to \rel G_1$ (because $\Lambda$ is a reduction). A homomorphism from the possibly infinite $\Gamma \rel G_2$ to the finite $\rel G_1$ is then given by compactness.
\end{proof}

While monotonicity is a key assumption in Theorem~\ref{thm:all-red}, preservation of disjoint unions can always be enforced on any reduction by first precomputing connected components of the input (which can be done in log-space due to \cite{Rei08}), and then applying the original reduction on each of the components separately.

We note that all the proofs in this section reduce between decision problems; they can be adapted for search problems. For that we need to additionally assume that there is an efficient way to find a~homomorphism $\rel I\to \Gamma \rel G$ given a~homomorphism $\Lambda \rel I \to \rel G$ on input (note that $\rel G$ is fixed here). All the adjoint pairs that we use in the following subsections indeed have this property.

From now on, we return to considering only finite digraphs.

\subsubsection{Reductions that have both a left and a right adjoint}

In the two applications below, we use reductions that are a left adjoint from one adjoint pair and, at the same time, the right adjoint from another adjoint pair. (In fact, these reductions will be of the form a~pp-power $\Gamma_\phi$, as described in Subsection~\ref{subsec:adj-csp}, for special gadgets $\phi$). The property of being both left and right adjoint has the following consequence.

\begin{theorem}\label{thm:adj2}
  Let $\Gamma$ be a log-space computable function that has a right adjoint $\Omega$, and a log-space computable left adjoint $\Lambda$.
  Then $\PCSP(\Gamma \rel H, \rel G)$ and $\PCSP(\rel H, \Omega \rel G)$ are log-space equivalent for all digraphs $\rel H,\rel G$ such that $\Gamma \rel H \to \rel G$.
\end{theorem}	
\begin{proof}
  Corollary~\ref{cor:adj1}(1) applied for $\Gamma$ and $\Omega$ gives that $\Gamma$ is a reduction from $\PCSP(\rel H,\Omega \rel G)$ to $\PCSP(\Gamma \rel H,\rel G)$. We claim that $\Lambda$ is a~reduction from $\PCSP(\Gamma \rel H,\rel G)$ to $\PCSP(\rel H,\Omega \rel G)$. This follows from Theorem~\ref{thm:adj1} applied to $\Lambda, \Gamma$: We need to check that $\Gamma \rel H \to \Gamma \rel H$, which holds trivially, and that $\Gamma \Omega \rel G \to \rel G$, which follows by Lemma~\ref{lem:adj-technical}(2), since $\Gamma$ and $\Omega$ are adjoint.
\end{proof}

\subsection{The arc digraph construction}
  \label{sec:righthard}

Let $\rel D$ be a digraph. The \emph{arc digraph} (or \emph{line digraph}) of $\rel D$, denoted $\delta \rel D$ , is the digraph whose vertices are arcs (directed edges) of $\rel D$ and whose arcs are pairs of the form $((u,v),(v,w))$.
In other words, $\delta\colon \dgraf\to\dgraf$ is the pp-power $\Gamma_\phi$ corresponding to ($n=2$ and) the following gadget digraph:
\[\begin{tikzpicture}
	\node[circle,fill,inner sep=1,label={left:$x_1$}] (a) at (0,0) {};
	\node[circle,fill,inner sep=1,label={above:$x_2=y_1$}] (b) at (1,0.5) {};	
	\node[circle,fill,inner sep=1,label={right:$y_2$}] (c) at (2,0) {};
	\draw[->] (a)--(b); \draw[->] (b)--(c);
\end{tikzpicture}\]
or to the pp-formula $\phi = (x_1,x_2)\in E \wedge (y_1,y_2)\in E \wedge x_2 = y_1$.  It thus has a left adjoint $\delta_L = \Lambda_\phi$, though we will not need it.  More surprisingly, $\delta$ has a right adjoint $\delta_R\colon \dgraf\to\dgraf$.

\begin{definition} \label{def:delta_R}
  For a digraph $\rel D$, let $\delta_R \rel D$ be the digraph that has a vertex for each pair $S,T \subseteq V(D)$, where $S$ or $T$ can be empty, such that $S \times T \subseteq E(D)$, and an arc from $(S,T)$ to $(S',T')$ if and only if $T \cap S' \neq \emptyset$.
\end{definition}

We give a proof of the adjunction below for completeness.
While $\delta$ will be the reduction we use, $\delta_R$ will be useful for understanding the best reduction we can get from $\delta$.

\begin{lemma}[\cite{FoniokT15}]
	$\delta$ and $\delta_R$ are adjoint.
\end{lemma}
\begin{proof}
	Let $\rel H, \rel G$ be digraphs and let $h \colon \delta \rel H \to \rel G$ be a homomorphism.
	That is, $h(u,v)$ is a vertex of $\rel G$ for each arc $(u,v)$ of $\rel H$,
	and for every pair of arcs $(u,v),(v,w)$ in $\rel H$,
	there is an arc from $h(u,v)$ to $h(v,w)$ in $\rel G$.
	We can define a homomorphism $\rel H \to \delta_R \rel G$ as $v \mapsto \left(s(v),t(v)\right)$, where $s(v) := \{h(u,v) \mid (u,v) \in E(H)\}$ and $t(v) := \{h(v,w) \mid (v,w) \in E(H)\}$.
	Then $s(v) \times t(v) \subseteq E(G)$, so $(s(v),t(v))$ is indeed a vertex of $\delta_R \rel G$.
	Moreover, for every arc $(u,v)$ of $\rel H$, $t(u) \cap s(v)$ is non-empty, as it contains $(u,v)$; hence $(s,t)$ is a homomorphism to $\delta_R \rel G$.
	
	Conversely, let $(s,t)$ define a homomorphism $\rel H \to \delta_R \rel G$.
	That is, $s(v),t(v)$ are subsets of $V(G)$ such that $s(v) \times t(v) \subseteq E(G)$ and for every arc $(u,v)$ of $\rel H$, $t(u) \cap s(v) \neq \emptyset$.
	We define a homomorphism $h\colon \delta\rel H \to \rel G$ as follows:
	choose $h(u,v)$ to be an arbitrary vertex in $t(u) \cap s(v)$.
	For any two arcs $(u,v),(v,w)$ in $\rel H$, we have that $h(u,v)$ is a vertex in $s(v)$ and $h(v,w)$ is a vertex in $t(v)$, hence $(h(u,v),h(v,w))$ is an arc of $\rel G$.
	Thus $h$ is indeed a homomorphism $\delta\rel H \to \rel G$.
\end{proof}

By Corollary~\ref{cor:adj1}(2), $\delta$ is a reduction from $\PCSP(\delta_R \rel H, \delta_R \rel G)$ to $\PCSP(\rel H, \rel G)$, for all digraphs $\rel H,\rel G$.  Let us see what this gives for classical colourings, i.e., when $\rel H$ and $\rel G$ are cliques.  Let us denote $b(n)\defeq \B{n}$.

\begin{observation}\label{obs:approxPoljakRodl} \label{obs:cliques-and-delta}
  For all $n\in\NN$, there are homomorphisms
  \[
    \rel K_{b(n)} \to \rel \delta_R \rel K_n \to \rel K_{2^n}
  .\]
\end{observation}
\begin{proof}
  Consider vertices of the form $(S,V(K_n)\setminus S)$ in $\delta_R \rel K_n$, for subsets $S$ of $V(K_n)$ of size exactly $\lfloor n/2 \rfloor$.
  Clearly for any two such different $S,S'$, the set $S'$ intersects $V(K_n) \setminus S$, so these vertices from a clique of size $b(n)$ in $\delta_R \rel K_n$.
  For the other bound, note that mapping a vertex $(S,T)$ to $(S, V(K_n) \setminus S)$ gives a homomorphism from $\rel \delta_R \rel K_n$ to its subgraph of size at most $2^n$, and therefore to the clique~$\rel K_{2^n}$.
\end{proof}

In other words,  if $\chi(\delta \rel G) \leq n$ (i.e., if $\delta \rel G \to \rel K_n$) then $\rel G \to \delta_R \rel K_n \to \rel K_{2^n}$, hence $\chi(\rel G) \leq 2^n$.
Similarly, if $\chi(\rel G) \leq b(n)$, then $\chi(\delta \rel G) \leq n$.
Therefore, $\delta$ has the remarkable property of decreasing the chromatic number roughly logarithmically (even though it is computable in log-space!). This was first proved by Harner and Entringer in \cite{HarnerE72}.

Observation~\ref{obs:approxPoljakRodl} can be made tight
if we use another, somewhat trivial adjunction between digraphs and graphs:
Let $\sym \rel D$ be the symmetric closure of a digraph $\rel D$ and let $\sub \rel D$ be the maximal symmetric subgraph of $\rel D$; so $\sub \rel D \to \rel D \to \sym \rel D$ by the identity maps. Observe that $\sym$ and $\sub$ are adjoint: $\sym \rel D \to \rel D'$ if and only if $\rel D \to \sub \rel D'$ for all digraphs $\rel D,\rel D'$.\footnote{This is in fact the composition of two adjoint pairs: taking $\sym$ and $\sub$ as functions from digraphs to graphs and the inclusion function $\iota$ from graphs to digraphs, we have $\sym \rel D \to \rel G$ if and only if $\rel D \to \iota \rel G$ and $\iota \rel G \to \rel D$ if and only if $\rel G \to \sub \rel D$ for all graphs $\rel G$ and digraphs $\rel D$.}
Composing the two adjunctions, we get that $\delta \sym$ is adjoint to $\sub \delta_R$.
Therefore, for any digraphs $\rel H, \rel G$ with $\rel H\to\rel G$, $\delta\sym$ is a reduction from $\PCSP(\sub\delta_R \rel H, \sub\delta_R \rel G)$ to $\PCSP(\rel H, \rel G)$ by Corollary~\ref{cor:adj1}(2).
For cliques, Poljak and R\"odl~\cite{PoljakR81} showed the following.

\begin{lemma}[\cite{PoljakR81}]\label{lem:poljakRoedl} \label{lem:PR}
	For all $n \in \NN$, $\sub \delta_R \rel K_n$ is homomorphically equivalent to $\rel K_{b(n)}$.
\end{lemma}
\begin{proof}
	As before, mapping a vertex $(S,T)$ of $\sub \delta_R \rel K_n$ to $(S, V(K_n)\setminus S)$ gives a homomorphism to the subgraph induced by vertices of the form $(S, V(K_n)\setminus S)$,
	so we can restrict our attention to it.
	A (bidirected) clique corresponds exactly to an antichain of sets $S$ (in the subset lattice),
	so by Sperner's theorem (on the maximal size of such antichains) the largest clique has size $b(n)$.
	Independent sets in this subgraph correspond exactly to chains of sets $S$,
	thus by Dilworth's theorem (on poset width) the subgraph can be covered with $b(n)$ independent sets, giving a $b(n)$-colouring.
\end{proof}

Lemma~\ref{lem:poljakRoedl} is equivalent to the statement that for a undirected graph $\rel G$, $\delta \rel G = \delta \sym \rel G \to \rel K_n$ if and only if $\rel G \to \rel K_{b(n)}$.
This, in particular, means that the chromatic number of $\delta \rel G$ is determined by $\chi(\rel G)$, namely $\chi(\delta \rel G) = \min\{n \mid \chi(\rel G) \leq b(n)\}$.
This together with Corollary~\ref{cor:adj1}(2) implies that $\delta \sym$ gives the following reduction for approximate colouring:

\begin{lemma}\label{lem:red}
  $\PCSP(\rel K_{b(k')},\rel K_{c'})$ log-space reduces to $\PCSP(\rel K_{k},\rel K_{c})$, for all $k' > 1$, $c' \geq b(k') > 1$, $c \geq k$ such that $b(c)\leq c'$ and $k' \leq k$.
\end{lemma}

\begin{proof}
  By Theorem~\ref{thm:adj1}, $\delta\sym$ is a reduction between the two problems if $\rel K_{b(k')} \to \sub\delta_R \rel K_{k}$ and $\sub\delta_R \rel K_{c} \to \rel K_{c'}$. Lemma~\ref{lem:PR} then implies that the first condition is satisfied if $b(k') \leq b(k)$, which is implied by $k'\leq k$, and the second is satisfied if $b(c) \leq c'$.
\end{proof}

\begin{remark}
We remark that the reduction in Lemma~\ref{lem:red} cannot be obtained by using a standard gadget reduction captured by the algebraic approach~\cite{BBKO19} (see Example~\ref{ex:gadgetVsHom}).
In detail, since $b(4)=6$, $\PCSP(\rel K_6,\rel K_{b(c)})$ log-space reduces to $\PCSP(\rel K_4,\rel K_c)$ for all $c \geq 4$. This contrasts with~\cite[Proposition~10.3]{BBKO19} which says that there exists a $c$ such that $\Pol(\rel K_4,\rel K_c)$ admits no minion homomorphism to any $\Pol(\rel K_{k'},\rel K_{c'})$ for $c' \geq k' > 4$.
Therefore, constructions like $\delta$ change the set of polymorphisms in an essential way and we believe that understanding the relation between $\Pol(\rel H,\delta_R \rel G)$ and $\Pol(\delta \rel H,\rel G)$ is an important question for future work.
\end{remark}

\subsubsection{Proof of Theorem \ref{thm:asymp}}

One consequence we derive from Lemma~\ref{lem:red} is a strengthening of Huang's result:

\begin{theorem}[Huang~\cite{Hua13}]\label{thm:Huang}
	For all sufficiently large $k$ and $c = 2^{\Omega(k^{1/3})}$, $\PCSP(\rel K_k, \rel K_c)$ is \NP-hard.
\end{theorem}

We improve the asymptotics from a~sub-exponential $c$ to single-exponential: from $c = \Theta(2^{k^{1/3}})$ to $c = b(k) - 1 = \Theta(2^k / \sqrt k)$ while at the same time relaxing the condition from ``sufficiently large $k$'' to $k\geq 4$.

\maintheoremasymptotics*

This theorem is proved by starting from Theorem~\ref{thm:Huang} and repeatedly using the reduction $\delta \sym$. Roughly speaking, each step improves the asymptotics a little. After a few steps, this results in a single-exponential function, and with slightly more precision, this results in exactly $b(k)-1$.
Moreover, one can notice that the requirements on ``sufficiently large $k$'' gets relaxed with every step. This allows us after sufficiently many steps to arrive at any $k \geq 4$.

We note it would \emph{not} be sufficient to start from a quasi-polynomial $c=k^{\Theta(\log k)}$ in Khot's \cite{Khot01} earlier result in place of Huang's Theorem~\ref{thm:Huang}.

\begin{proof}[Proof of Theorem~\ref{thm:asymp}]
	We start with Theorem~\ref{thm:Huang} which asserts a constant $C>0$ such that
  \begin{equation} \label{eq:Hua}
  \PCSP(\rel K_{k}, \rel K_{2^{\lfloor C \cdot k^{1/3}\rfloor}})
    \text{ is \NP-hard, for sufficiently large $k$.}
  \end{equation}

  After one reduction using Lemma~\ref{lem:red}, we obtain the following.

  \begin{claim}
    $\PCSP(\rel K_{k}, \rel K_{\floor{2^{k/4}}})$ is \NP-hard, for sufficiently large $k$.
  \end{claim}

  \begin{proof}
    Substituting $b(k)$ for $k$ in (\ref{eq:Hua}), we get that $\PCSP(\rel K_{b(k)}, \rel K_{2^{\lfloor C \cdot b(k)^{1/3}\rfloor }})$ is \NP-hard, for sufficiently large $k$. We apply Lemma~\ref{lem:red}; to show that it implies the claim, we need $b(\floor{2^{k/4}}) \leq 2^{\floor{C \cdot b(k)^{1/3}}}$. This follows since $b(m) \leq 2^m$ for all $m$ and $2^{k/4} \leq (2^k/k)^{1/3} \leq \floor{C\cdot b(k)^{1/3}}$ for sufficiently large $k$, and therefore
    \[
      b(\floor{2^{k/4}}) \leq 2^{\floor{2^{k/4}}} \leq 2^{\floor{C \cdot b(k)^{1/3}}}
    \]
    as we wanted to show.
  \end{proof} 

  The second reduction gives:
  \begin{claim}
    $\PCSP(\rel K_{k}, \rel K_{\floor{2^k/4k}})$ is \NP-hard, for sufficiently large $k$.
  \end{claim}
  \begin{proof}
    Substitute $b(k)$ for $k$ in the first claim to get that $\PCSP(\rel K_{b(k)},\rel K_{\floor{2^{b(k)/4}}})$ is \NP-hard. Observe that
    \(
      2^k/4k \leq b(k)/4
    \)
    for sufficiently large $k$, and therefore
    \[
      b(\floor{2^k/4k}) \leq b(\floor{b(k)/4}) \leq 2^{\floor{b(k)/4}} \leq \floor{2^{b(k)/4}}.
    \]
    Hence by Lemma~\ref{lem:red}, $\PCSP(\rel K_{b(k)},\rel K_{\floor{2^{b(k)/4}}})$
    reduces to $\PCSP(\rel K_{k}, \rel K_{\floor{2^k/4k}})$.
  \end{proof} 

  After the third reduction, we get the following.
  \begin{claim} \label{cl:3}
    $\PCSP(\rel K_{k}, \rel K_{b(k-1)})$ is \NP-hard for sufficiently large $k$.
  \end{claim}
  \begin{proof}
    Again, substitute $b(k)$ for $k$ in the second claim to obtain that $\PCSP(\rel K_{b(k)}$, $\rel K_{\floor{2^{b(k)}/4b(k)}})$ is \NP-hard.
    Observe that $b(k) \geq \frac{3}{2} b(k-1)$ for all $k\geq 1$, since
    \begin{align*}
      b(2k) &\textstyle
        = \binom{2k}{k} = \binom{2k-1}{k-1} \frac{2k}{k} = 2 \cdot b(2k-1) \geq \frac{3}{2} b(2k-1), \text{ and} \\
      b(2k+1) &\textstyle
        = \binom{2k+1}{k} = \binom{2k}{k} \frac{2k+1}{k+1} = b(2k) (2-\frac{1}{k+1})\geq \frac{3}{2} b(2k).
    \end{align*}
    Therefore,
    \[\textstyle
      b(b(k-1)) \leq b(\frac23 b(k)) \leq 2^{\frac23 b(k)} \leq 2^{b(k)}/4b(k).
    \]
    for sufficiently large $k$,
    and consequently, by Lemma~\ref{lem:red}, $\PCSP(\rel K_{b(k)},\rel K_{\floor{2^{b(k)}/4b(k)}})$ reduces to $\PCSP(\rel K_{k}, \rel K_{b(k-1)})$.
  \end{proof}

  Finally after the fourth reduction:
  \begin{claim} \label{cl:4}
    $\PCSP(\rel K_{k}, \rel K_{b(k)-1})$ is \NP-hard for sufficiently large $k$.
  \end{claim}
  \begin{proof}
    Substitute $b(k)$ for $k$ in the third claim to get that $\PCSP(\rel K_{b(k)},\rel K_{b(b(k)-1)})$ is \NP-hard.
    By Lemma~\ref{lem:red} this reduces to $\PCSP(\rel K_{k}, \rel K_{b(k)-1})$.
  \end{proof}

  This concludes the improvement in asymptotics.
  To relax the requirements for $k$, we repeat the reduction enough times. Each step is given by the following claim.
  \begin{claim} \label{cl:5}
    $\PCSP(\rel K_{b(k)}, \rel K_{b(b(k))-1})$ log-space reduces to
    $\PCSP(\rel K_{k}, \rel K_{b(k)-1})$ for all $k \geq 4$.
  \end{claim}
  \begin{proof}
    Lemma~\ref{lem:red} gives the reduction since $b(k)$ is strictly increasing and $b(k)>4$ for $k \geq 4$, and hence $b(b(k)-1) \leq b(b(k))-1$.
  \end{proof}

  To finish the proof, assume that $k'\geq 4$ and let $k_0$ be sufficiently large so Claim~\ref{cl:4} is true for all $k\geq k_0$. Since $k'\geq 4$ there is $n$ such that $k_0 \leq b^{(n)}(k')$ (where $b^{(n)}$ denotes the $n$-fold composition of $b$). Applying Claim~\ref{cl:5} $n$ times gives a log-space reduction from $\PCSP(\rel K_{b^{(n)}(k')},\rel K_{b^{(n+1)}(k')-1})$ to $\PCSP(\rel K_{k'},\rel K_{b(k')-1})$, and since the first problem is \NP-hard (Claim~\ref{cl:4}) this concludes that $\PCSP(\rel K_{k'},\rel K_{b(k')-1})$ is \NP-hard for all $k'\geq 4$.
\end{proof}

\subsubsection{Proof of Theorem~\ref{thm:conditional}}

In this section we prove a slightly more general result than Theorem~\ref{thm:conditional}:

\begin{theorem}\label{thm:conditional-ext}
    If, for some loopless digraph $\rel H$ and all loopless digraphs $\rel G$ such that $\rel H\to \rel G$, $\PCSP(\rel H,\rel G)$ is \NP-hard, then $\PCSP(\rel K_3,\rel G)$ is \NP-hard for all loopless digraphs $\rel G$ such that $\rel K_3 \to \rel G$.
\end{theorem}

Indeed, it is easy to see that the assumption of Theorem~\ref{thm:conditional-ext} is slightly weaker than the assumption of Theorem~\ref{thm:conditional}, while the conclusions of the two theorems are equivalent.

We prove Theorem~\ref{thm:conditional-ext} by iterating the reduction given by $\delta$ in a similar way as in the last paragraph of the proof of Theorem~\ref{thm:asymp}; in fact that part of the proof could be used with only minor changes to prove that $\PCSP(\rel K_4,\rel K_c)$ is \NP-hard for all $c\geq 4$ if $\PCSP(\rel K_k,\rel K_c)$ is \NP-hard for some $k$ and all $c\geq k$. We get to $\rel K_3$ by omitting the intermediate use of $\sym$, i.e., we keep orientation of the edges. One step of the reduction is given by the following lemma.

\begin{lemma}\label{lem:right-hard-digraph}
    Let $\rel H$ be a loopless digraph.
    If, for all loopless digraphs $\rel G'$ such that $\rel H \to \rel G'$, $\PCSP(\rel H, \rel G')$ is \NP-hard, then $\PCSP(\delta\rel H,\rel G)$ is \NP-hard for all loopless digraphs $\rel G$ such that $\delta \rel H \to \rel G$.
\end{lemma}
\begin{proof}
    Let $\rel H$ be a digraph that satisfies the premise, and let $\rel G$ be a loopless digraph such  that $\delta\rel H \to \rel G$. We aim to prove that $\PCSP(\delta \rel H,\rel G)$ is \NP-hard.
    Corollary \ref{cor:adj1}(1) gives a log-space reduction from $\PCSP(\rel H, \delta_R \rel G)$ to $\PCSP(\delta \rel H,\rel G)$. We claim that the digraph $\delta_R \rel G$ is loopless, which follows from the construction of $\delta_R$ (see Definition~\ref{def:delta_R}) and the assumption that $\rel G$ is loopless: Indeed, if a vertex $(S,T)$ in $\delta_R\rel G$ has a loop, then $S\cap T \neq \emptyset$. Consequently, $\rel G$ has a loop on any vertex $v\in S\cap T$ since $S\times T \subseteq E(G)$. Clearly also $\rel H\to \delta_R \rel G$ since $\delta$ and $\delta_R$ are adjoint. Therefore, we have that $\PCSP(\rel H,\delta_R \rel G)$ is \NP-hard by the assumption of the lemma, and we conclude that $\PCSP(\delta \rel H,\rel G)$ is \NP-hard as well.
\end{proof}

To finish the proof of Theorem~\ref{thm:conditional-ext}, we will need the following two lemmas. The first one, which is a special case of the second, was independently discovered by Zhu~\cite{zhu1998survey}, Poljak~\cite{Poljak1991}, and Schmerl (unpublished, see~\cite{Tardif08:survey}). For the sake of completeness, we include the proof of Zhu~\cite{zhu1998survey}.

\begin{lemma}\label{lem:deltadelta}
    There is a homomorphism $\delta(\delta\rel K_4) \to \rel K_3$.
\end{lemma}

\begin{proof}
  The vertices of $\delta(\delta \rel K_4)$ are two consecutive pairs of arcs, i.e., they correspond to triples $(i,j,k)$ in $\{0,1,2,3\}$ such that $i\neq j$ and $j\neq k$. Two such triples $(i,j,k)$ and $(i',j',k')$ are adjacent if $j = i'$ and $k = j'$. We define $h\colon \delta(\delta \rel K_4) \to \rel K_3$ so that $h\colon (i,j,k) \mapsto j$ if $j \in \{0,1,2\}$, and $h\colon (i,3,k) \mapsto c$ for some $c \in \{0,1,2\} \setminus \{i,k\}$. It is straightforward to check that such $h$ is a~valid colouring.
\end{proof}

We note that $\delta(\sym\delta\rel K_4)$ is \emph{not} 3-colourable, so it is important here to iterate $\delta$ rather than $\sym\delta$.

The next lemma essentially shows that iterating $\delta$ many times can bring a chromatic number of any finite loopless digraph down to $3$. We use $\delta^{(i)}\rel D$ to denote the digraph obtained from $\rel D$ by applying $\delta$ $i$ times.

\begin{lemma} \label{lem:delta^i}
    For every loopless digraph $\rel D$ there exists $i\geq 0$ such that $\delta^{(i)} \rel D \to \rel K_3$.
\end{lemma}

\begin{proof}
    If $\rel D$ is a loopless digraph then $\rel D \to \rel K_{\size{V(D)}}$. Therefore, since $\delta$ is monotone, it is enough to prove the statement for $\rel D = \rel K_c$ and all $c > 3$ (the case $c\leq 3$ is trivial).
    Recall that $\rel K_{b(n)} \to \delta_R\rel K_n$ for all $n\geq 1$
    (Observation~\ref{obs:cliques-and-delta}), and consequently $\delta\rel
    K_{b(n)} \to \rel K_n$. Recall that $b(n)$ is strictly increasing. Now since $b(n)>4$ for $n\geq 4$, there is $j$ such that $b^{(j)}(4) \geq c$. Since $\rel K_c \to \rel K_{b^{(j)}(4)}$, the monotonicity of $\delta$ implies $\delta\rel K_c \to \delta\rel K_{b^{(j)}(4)}$, which, with $\delta\rel K_{b^{(j)}(4)}\to \rel K_{b^{(j-1)}(4)}$, implies $\delta\rel K_c \to \rel K_{b^{(j-1)}(4)}$. Similarly, we get $\delta^{(2)}\rel K_c \to \delta\rel K_{b^{(j-1)}(4)}\to \rel K_{b^{(j-2)}(4)}$. Proceeding in the same way, we eventually get $\delta^{(j)} \rel K_c \to \rel K_4$. This together with Lemma~\ref{lem:deltadelta} gives that $\delta^{(j+2)} \rel K_c \to \rel K_3$ which gives the claim for $i = j+2$.
\end{proof}

\begin{proof}[Proof of Theorem~\ref{thm:conditional-ext}]
    Assume that $\rel H$ is a~loopless digraph s.t.\ $\PCSP(\rel H,\rel G)$ is \NP-hard for all loopless digraphs $\rel G$ with $\rel H\to \rel G$. Let $i$ be such that $\delta^{(i)}\rel H \to \rel K_3$ which exists from Lemma~\ref{lem:delta^i}. Now, iterating Lemma~\ref{lem:right-hard-digraph} $i$ times gives that $\PCSP(\delta^{(i)} \rel H, \rel G)$ is \NP-hard for all $\rel G$ such that $\delta^{(i)} \rel H \to \rel G$. Since $\delta^{(i)} \rel H \to \rel K_3$, $\PCSP(\delta^{(i)} \rel H, \rel G)$ trivially reduces to $\PCSP(\rel K_3, \rel G)$.
\end{proof}

We remark that, if one iterates the reduction using $\delta$ further, one cannot improve Theorem~\ref{thm:conditional-ext} to imply the full extent of Conjecture~\ref{conj:bg}. In particular, iterating $\delta$ cannot be used to show $\PCSP(\rel C_5,\rel K_c)$ is \NP-hard for all $c\geq 3$ given that $\PCSP(\rel K_3,\rel K_c)$ is \NP-hard for all $c\geq 3$. This is because $\delta^{(i)} \rel K_3$ contains a directed cycle of length three, for all integers $i$.

\begin{remark}
    Theorems~\ref{thm:conditional} and~\ref{thm:conditional-ext} can be phrased in algebraic terms using so-called \emph{$\rel H$-loop conditions}, which recently gained popularity in universal algebra (see, e.g.~\cite{Ols18}) and featured in~\cite[Section 6]{BBKO19}. Fix a digraph $\rel H$ and let $(a_1,b_1),\ldots,(a_m,b_m)$ be the full list of arcs of $\rel H$.
    The $\rel H$-loop condition is the following identity (i.e., function equation) involving two minors of a function $f$
    \begin{equation}\label{eq:h-loop}
    f(x_{a_1},\ldots,x_{a_m})=f(x_{b_1},\ldots,x_{b_m}).
    \end{equation}
    (This definition is equivalent to the one used in~\cite[Section 6]{BBKO19}).
    One says that a minion $\clo M$ satisfies the $\rel H$-loop condition if it contains a function $f$ satisfying \eqref{eq:h-loop}.

    The following two statements are equivalent for each loopless digraph $\rel H$:
    \begin{itemize}
        \item $\PCSP(\rel H,\rel G)$ is \NP-hard for all loopless digraphs $\rel G$ with $\rel H\to \rel G$; and
        \item $\PCSP(\rel A,\rel B)$ is \NP-hard for all pairs of similar structures $\rel A$, $\rel B$ such that $\Pol(\rel A,\rel B)$ does not satisfy the $\rel H$-loop condition.
    \end{itemize}
    This claim can be proved in a similar way as \cite[Theorems 6.9 and 6.12]{BBKO19}.

    In this interpretation, Lemma~\ref{lem:right-hard-digraph} can be rephrased as follows: If, for all PCSPs, the failure to satisfy the $\rel H$-loop condition implies \NP-hardness, then so does the failure to satisfy the $\delta \rel H$-loop condition.
    Can this perspective be used to bring some ideas from algebra (e.g.\ \cite{Ols18}) to obtain better conditional hardness?
\end{remark}

\subsection{Only topology matters}
  \label{subsec:secondMainProof}

All graphs in this subsection are undirected.

Recall Example~\ref{ex:walk-power}. The functions $\Lambda_k$ and $\Gamma_k$ from this example are adjoint, for all $k$. More surprisingly, for odd $k$, $\Gamma_k$ is itself the \emph{left} adjoint of a certain function~$\Omega_k\colon \graf\to\graf$, i.e., for all graphs $\rel H$ and $\rel G$, $\Gamma_k \rel H \to \rel G$ if and only if $\rel H \to \Omega_k \rel G$.
The graph $\Omega_k \rel G$ for $k = 2\ell+1$ is defined as follows; the vertices of $\Omega_k \rel G$ are tuples $(A_0,\dots,A_\ell)$ of vertex subsets $A_i \subseteq V(G)$ such that $A_0$ contains exactly one vertex.  Two such tuples $(A_0,\dots,A_\ell)$ and $(B_0,\dots,B_\ell)$ are adjacent if $A_i \subseteq B_{i+1}$, $B_i \subseteq A_{i+1}$ for all $i=0\dots \ell-1$, and $A_\ell \times B_\ell \subseteq E(G)$.

If there is a homomorphism $f \colon \Gamma_k \rel H \to \rel G$, then a homomorphism $\rel H \to \Omega_k \rel G$ is obtained by mapping $v$ to $(f(N^0(v)),\dots,f(N^\ell(v)))$, where $N^i$ denotes the set of vertices reachable via a walk of length exactly $i$.
Conversely, if there is a homomorphism $f \colon \rel H \to \Omega_k \rel G$, then a homomorphism $\Gamma_k \rel H \to \rel G$ is obtained by mapping $v$ to the unique vertex in the first set of $f(v)=(A_0,\dots, A_\ell)$.

We note that $\Lambda_k$ and $\Gamma_k$ are log-space computable, for all odd $k$; however, $\Omega_k$ is not: $\Omega_k \rel G$~is exponentially larger than $\rel G$.
Below, we will use the following observation that can be found together with more properties of the functions $\Lambda_k,\Gamma_k$, and $\Omega_k$ in~\cite[Lemma 2.3]{Wrochna17b}.

\begin{lemma} \label{lem:lambda-to-omega}
  For any graph $\rel G$ and all odd $k$, $\Lambda_k \rel G \to \Omega_k \rel G$.
\end{lemma}

\begin{proof}
  Since $\Gamma_k$ and $\Omega_k$ are adjoint, the claim is equivalent to $\Gamma_k\Lambda_k \rel G \to \rel G$ (note that $\rel G \to \Gamma_k\Lambda_k\rel G$ since $\Lambda_k$ and $\Gamma_k$ are adjoint, so we will prove that the two graphs are homomorphically equivalent). We describe one such homomorphism $h\colon \Gamma_k\Lambda_k \rel G\to \rel G$. Note that $V(\Gamma_k\Lambda_k G) = V(\Lambda_k G) \supseteq V(G)$. We put $h(v) = v$ for each $v\in V(G)$ and extend this to vertices introduced to $\Lambda_k \rel G$ by replacing an edge $(u,v)\in E(G)$ with a path of length $k$ by alternatively mapping vertices on this path to $u$ and $v$ in such a way that $h$ restricted to the path is a homomorphism from the odd path to $(u,v)$. It is straightforward to check that $h$  maps any two vertices of $\Lambda_k \rel G$ that are connected by a path of length $k$ to an edge of $\rel G$, and therefore, it is a homomorphism from $\Gamma_k \Lambda_k \rel G$ to $\rel G$.
\end{proof}

The next lemma gives the key reduction for the main result of this subsection, Theorem~\ref{thm:topoOnly}.

\begin{lemma}\label{lem:left-hard}
  Let $k$ be odd and $\rel G$ be a graph. If $\PCSP(\rel H',\Omega_k \rel G)$ is \NP-hard for all non-bipartite graphs $\rel H' \to \Omega_k \rel G$, then $\PCSP(\rel H,\rel G)$ is \NP-hard for all non-bipartite graphs $\rel H \to \rel G$.
\end{lemma}

\begin{proof}
  Let $\rel H, \rel G$ be non-bipartite with $\rel H \to \rel G$.
  By Corollary~\ref{cor:adj1}(2), $\Gamma_k$ is a reduction from $\PCSP(\Omega_k \rel H,\Omega_k \rel G)$ to $\PCSP(\rel H,\rel G)$.
  To conclude that $\PCSP(\rel H, \rel G)$ is \NP-hard it remains to show that $\Omega_k \rel H$ is non-bipartite.  Observe that since $\rel H$ is non-bipartite and $k$ is odd then also $\Lambda_k\rel H$ is non-bipartite. Furthermore, from the above lemma, $\Lambda_k \rel H \to \Omega_k \rel H$ which implies that $\Omega_k \rel H$ is also non-bipartite.
\end{proof}

Thus (by the above lemma and homomorphic relaxation) if we know one graph $\rel G'$ such that $\PCSP(\rel H',\rel G')$ is \NP-hard for all non-bipartite $\rel H'$, then we can conclude the same for all $\rel G$ such that $\Omega_k \rel G \to \rel G'$, for some odd $k$.
When does such a $k$ exists?  The answer, given in~\cite{Wrochna17b}, turns out to be topological.
We remark that the results in~\cite{Wrochna17b} use the so-called box complex of $\rel G$ instead of $\Bip{\rel G}$. However, there exist $\ZZ_2$-maps (in both directions) between the two complexes, see~\cite[Proposition 4(M2,M3,M7)]{MZ04} for explicit maps. This is enough for our purposes, but a stronger claim is true --- the two complexes are $\ZZ_2$-homotopy equivalent (as defined in Appendix~\ref{app:hom-and-box})~\cite{Zivaljevic05,Csorba08}.

Intuitively, while the operation $\Gamma_k$ gives a ``thicker'' graph, the operation $\Omega_k$ gives a ``thinner'' one.  In fact, $\Omega_k$ behaves similarly to barycentric subdivision in topology: it preserves the topology of a graph (formally, $\gBip{\Omega_k \rel G}$ is $\ZZ_2$-homotopy equivalent to $\gBip{\rel G}$~\cite{Wrochna17b}) but refines its geometry.  With increasing $k$, this eventually allows to model any continuous map with a graph homomorphism; in particular we have the following.

\begin{theorem}[\cite{Wrochna17b}] \label{thm:approx}
  There exists a $\ZZ_2$-map $\gBip{\rel H} \to_{\ZZ_2} \gBip{\rel G}$ if and only~if $\Omega_k \rel H \to \rel G$ for some odd $k$.
\end{theorem}

We now conclude the proof that whether the Brakensiek-Guruswami conjecture holds for a graph $\rel G$ (and all relevant $\rel H$) depends only on the topology of $\rel G$  --- this was informally stated earlier in Theorem~\ref{thm:topoOnlyInformal}).
In fact, it only matters which $\ZZ_2$-maps $\geom{\Bip{\rel G}}$ admits.

\begin{theorem} \label{thm:topoOnly}
  Let $\rel G$, $\rel G'$ be graphs such that $\geom{\Bip{\rel G}}$ admits a $\ZZ_{2}$-map to $\geom{\Bip{\rel G'}}$ and suppose $\PCSP(\rel H,\rel G')$ is \NP-hard for all non-bipartite graphs $\rel H$ such that $\rel H \to \rel G'$. Then $\PCSP(\rel H,\rel G)$ is \NP-hard for all non-bipartite graphs $\rel H$ such that $\rel H \to \rel G$.
\end{theorem}

\begin{proof}
  By Theorem~\ref{thm:approx}, $\Omega_k \rel G \to \rel G'$ for some odd $k$. Since $\PCSP(\rel H,\rel G')$ is \NP-hard for all non-bipartite graphs $\rel H$, we also have that $\PCSP(\rel H,\Omega_k \rel G)$ is \NP-hard by a trivial reduction. Now, Lemma~\ref{lem:left-hard} gives the claim.
\end{proof}

In particular, Theorem~\ref{thm:topoOnly} implies that Theorems~\ref{thm:K3} and~\ref{thm:main-s1} are equivalent.

\section{Conclusion}

We presented two new methodologies, based on topology and adjunction, to analyse the complexity of PCSPs and provided some applications
of these methodologies to considerably improve state-of-the-art in the complexity of approximate graph colouring and promise graph homomorphism problems.

As mentioned before, there are many ways in which topology can potentially be applied in the analysis of polymorphisms from $\rel H$ to $\rel G$, for graphs or for general relational structures.
With the approach that we used, we made a few choices for our analysis.  Specifically, we used (a) the graph $\rel K_2$ to construct simplicial complexes $\Bip{\rel H}$ and $\Bip{\rel G}$, (b) $\mathbb{Z}_2$-action on our complexes and topological spaces, and (c) fundamental groups of topological spaces.
One can try to perform similar analysis, but (a) replacing $\rel K_2$ by any other graph $\rel K$ (or, for general PCSPs, by another appropriate structure), (b) using any subgroup of the automorphism group of $\rel K$ to account for symmetry (called ``equivariance'' in the topological literature) of the complexes and topological spaces, and (c) replacing the fundamental group with a different topological invariant of spaces or continuous functions involved.
Some examples of different choices, though not in the context of polymorphisms, can be found, e.g. in~\cite{Koz08-book,Mat03}. These are the obvious first choices of adapting the approach. Naturally, it can be changed in a more fundamental way: the most prominent seems to be directly analyse the topological structure of the simplicial complexes $\Hom(\rel H^n,\rel G)$ (see~\cite[Section~9.2.4]{Koz08-book} for related general suggestions). One advantage of this approach is that the analysis would depend only on the function minion $\Pol(\rel H,\rel G)$, rather than on the specific choice of $\rel H$ and $\rel G$. 

In this paper, we applied topology together with the algebraic theory from~\cite{BBKO19} to prove complexity results about promise graph homomorphism. However, our application can be seen as plugging the topological analysis into an algebraic result. Since topology appears to be naturally present in minions of polymorphisms, it would be interesting to further develop the algebraic theory from~\cite{BBKO19} to properly incorporate topology. 
Similarly, we used adjunction to obtain some reductions for approximate graph colouring problems that are provably cannot be captured by the algebraic theory from~\cite{BBKO19} --- it is natural to ask whether a more general theory can be constructed that incorporates both the current algebraic theory and adjunction.

It would be interesting to find further specific applications of our methodologies, for example, in approximate graph and hypergraph colouring and their variations, or possibly even beyond constraint satisfaction. Naturally, one would want to extend our methodologies to approximate graph colouring problems $\PCSP(\rel K_3,\rel K_c)$ or promise graph homomorphism problems $\PCSP(\rel C_k,\rel K_c)$ for $c\ge 4$. If one applies the same transformation of these graph problems into homomorphism complexes and topological spaces as we use in this paper, one would need to analyse the (\equivariant-)polymorphisms from $\Sphere^1$ to $\Sphere^m$ for $m\ge 2$. (Note that $\pi_1(\Sphere^m)$ is trivial for $m\ge 2$, so the fundamental group is of no use in this case).
These polymorphisms are \equivariant-maps from tori $\Torus^n$, $n\ge 1$, to $\Sphere^m$. This is somewhat related to some well-known hard open questions from algebraic topology, such as classification of maps from one sphere $\Sphere^{m_1}$ to another $\Sphere^{m_2}$. However, to the best of our knowledge, the equivariant version of such questions was never considered.  Moreover, for our purposes, it would suffice to get any classification of $\Pol(\Sphere^1,\Sphere^m)$ that can be connected with the algebraic theory, e.g. with Theorem~\ref{thm:bounded-arity} above, or with~\cite[Theorem~5.22]{BBKO19} or~\cite[Corollary~4.2]{BWZ19}. Of course, it is possible that some modification of our approach will need to be used.  In any case, we believe that topology will play an important part in settling the complexity of approximate graph colouring and the Brakensiek-Guruswami conjecture.

Finally, we remark that the standard reductions from the algebraic approach, i.e., reductions of the form $\Lambda_\phi$ (see Section~\ref{subsec:adj-csp}), can be thought of as replacing individual constraints in an instance with gadgets (possibly consisting of many constraints).  Similarly, certain reductions of the form $\Gamma_\phi$, such as $\delta$ and $\Gamma_k$ presented in Sections~\ref{sec:righthard} and~\ref{subsec:secondMainProof}, can be thought of as replacing gadgets (i.e., certain parts of input) with individual constraints.  The latter turned out to be particularly useful when they themselves admit some right adjoint $\Omega$ (as was the case for $\delta$ and $\Gamma_k$); however, such reductions must have a rather restricted form~\cite[Theorem 2.5]{FoniokT15}.  Thus, a natural extension would be to investigate reductions which replace gadgets with gadgets (i.e., introduce a copy of one gadget for each homomorphism from another gadget).

\appendix
\section{Equivalence of homomorphism complexes}

There is a superficial distinction between the abstract simplicial complex $\Hom(\rel H,\rel G)$ as we defined it and the definition of the homomomorphism complex in~\cite{BK06,Koz08-book}. We will show that the two definitions give topological spaces that are equivalent in the following sense.
\label{app:hom-and-box}

\begin{definition}
  Two $\ZZ_2$-spaces $\top X, \top Y$ are \emph{$\ZZ_2$-homotopy equivalent} if there are $\ZZ_2$-maps $\alpha \colon \top X \to \top Y$ and $\beta \colon \top Y \to \top X$ such that $\beta\alpha$ and $\alpha\beta$ are $\ZZ_2$-homotopic (recall Definition~\ref{def:z2-homotopy}) to the identity maps on $\top X$ and $\top Y$, respectively.
\end{definition}

This notion is coarser than $\ZZ_2$-homeomorphism (which required $fg$ and $gf$ to be equal to identity maps); for example, $\RR^2 \setminus \{(0,0)\}$ is $\ZZ_2$-homotopy equivalent to $\Sphere^1$ but not $\ZZ_2$-homeomorphic to it.
Nevertheless, $\ZZ_2$-homotopy equivalent spaces admit the same $\ZZ_2$-maps and have isomorphic fundamental groups, for example, they are thus indeed equivalent for our purposes.

We remark, that in \cite{Koz08-book} and other topological literature,  very little attention is given to the distinction between abstract (simplicial) complexes and their geometric realisations. In particular, in~\cite{BK06,Koz08-book}, the Hom complex of graphs $\rel H$, $\rel G$ is defined as a topological space and not an abstract simplicial complex (it is in fact a so-called \emph{prodsimplicial complex}; see \cite[p.~28]{Koz08-book}). The following definition is an equivalent formulation of Definition~9.23 in \cite{Koz08-book} using the comments thereafter.

\begin{definition}
  For a set $V$, we denote by $\Delta^V$ the \emph{standard simplex} with vertices $V$ that is defined as a subspace of $\mathbb R^V$, where the canonical unit vectors are identified with elements of $V$, obtained as the convex hull of $V$, i.e.,
  \(
    \Delta^V = \{ \sum_{v\in V} \lambda_v v \mid
      \lambda_v \in [0,1]\text{ for each $v\in V$, and }
      \sum_{v\in V} \lambda_v = 1 \}
  \).

  $\pHom(\rel K_2, \rel G)$ is a subspace of $\Delta^{V(G)} \times \Delta^{V(G)}$.
  Thus a point of this space is described as a pair of formal sums $(\sum_{u\in V(G)} \lambda_u u, \sum_{v\in V(G)} \rho_v v)$ such that $\sum_{u\in V(G)} \lambda_u = \sum_{v\in V(G)} \rho_v = 1$.
  Using this description, $\pHom(\rel K_2, \rel G)$ is defined as the subspace consisting of those points $(\sum_{u\in V(G)} \lambda_u u, \sum_{v\in V(G)} \rho_v v)$ such that
  $\{ u \mid \lambda_u > 0\} \times \{ v \mid \rho_v > 0 \} \subseteq E(G)$ is a complete bipartite graph.
  The action of $\mathbb Z_2$ on this complex is given by switching the two coordinates, i.e.,
  \( \textstyle
    - (\sum_{u\in V(G)} \lambda_u u, \sum_{v\in V(G)} \rho_v v) =
    (\sum_{v\in V(G)} \rho_v v, \sum_{u\in V(G)} \lambda_u u)
  \).
\end{definition}

On the other hand, the geometric realisation of our $\Hom(\rel K_2, \rel G)$ can be described in the following way.
It consists of convex combinations $\sum_{(u,v)\in E(G)} \lambda_{u,v} (u,v)$ (the points $(u,v) \in E(G)$ are identified with certain unit vectors in $\mathbb R^n$ where $n = \size{E(G)} / 2$ so that $-(u,v) = (v,u)$) such that
\(
  \{ u \mid \lambda_{u,v} > 0 \} \times \{ v \mid \lambda_{u,v} > 0 \} \subseteq E(G)
\)
is a complete bipartite subgraph.

\begin{figure}[t]
  \centering
  \raisebox{-.5\height}{
    \includegraphics[width=.4\textwidth]{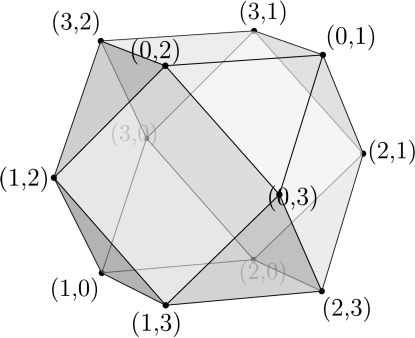}
  }
  \caption{$\pHom(\rel K_2,\rel K_4)$}
    \label{fig:koz-k2-k4}
\end{figure}

Both complexes are therefore defined using complete bipartite subgraphs, and the similarity is apparent. As an example that highlights the small differences, let us note that $\pHom(\rel K_2,\rel K_4)$ is \equivariant-homeomorphic to $\Sphere^2$; the space is depicted in Fig.~\ref{fig:koz-k2-k4}. This is since, unlike in $\gBip{\rel K_4}$, the tetragonal faces on the picture correspond to actual squares in $\pHom(\rel K_2,\rel K_4)$.

\begin{proposition}
  Let $\rel H$, $\rel G$ be graphs.
  Then the \equivariant-spaces $\gHom{\rel K_2,\rel G}$ and $\pHom(\rel K_2,\rel G)$ are \equivariant-homotopy equivalent.
\end{proposition}

\begin{proof}
We define continuous maps $\alpha$ and $\beta$ between the two spaces that witness the \equivariant-homotopy equivalence.
\[\begin{array}{rcccl}
& \geom{\Hom(\rel K_2, \rel G)} & & \pHom(\rel K_2, \rel G) & \\	
\alpha \colon
    & \sum_{(u,v) \in E(G)} \lambda_{(u,v)} (u,v) & \mapsto
    & \textstyle \left(
      \sum_{(u,v) \in E(G)} \lambda_{(u,v)} u,
      \sum_{(u,v) \in E(G)} \lambda_{(u,v)} v
    \right) & \\
& \sum_{(u,v)\in E(G)} \lambda_u \rho_v (u,v)
    & \mapsfrom
    & (\sum_{u\in V(G)} \lambda_u u, \sum_{v\in V(G)} \rho_v v)
    & {:\!\beta}
\end{array}\]
It is straightforward to check that $\alpha$ and $\beta$ are \equivariant-maps, that $\alpha\beta = \id$, and that $\beta\alpha$ maps each point to a point in the same face. Thus, a homotopy from $\beta\alpha$ to $\id$ can be defined by linearly interpolating between the two:
\(
  (p, t) \mapsto (1-t)p + t\beta(\alpha(p))
\)
for $p \in \gHom{\rel K_2, \rel G}$ and $t \in [0,1]$.
\end{proof}

We remark that the above proposition and its proof generalizes to arbitrary Hom complexes $\Hom(\rel H,\rel G)$ with the action of $\Aut(\rel H)$, more precisely, the complex $\pHom(\rel H,\rel G)$ (as defined in \cite[Definition 9.23]{Koz08-book}) is $\Aut(\rel H)$-homotopy equivalent to $\gHom{\rel H,\rel G}$ for any two graphs $\rel H,\rel G$.
 
\subsection*{Acknowledgements}

A.\,K.\ and J.\,O.\ would like to thank John Hunton for consultations on algebraic topology and Libor Barto and Antoine Mottet for inspiring discussions.

\bibliographystyle{alphaurl}
\newcommand{\etalchar}[1]{$^{#1}$}

\end{document}